\documentclass[runningheads]{llncs}

\usepackage[T1]{fontenc}
\usepackage{lmodern}
\usepackage{cite}
\usepackage{booktabs} 
\usepackage[ruled,noend,linesnumbered]{algorithm2e} 
\usepackage{graphicx} 
\usepackage{placeins} 
\usepackage{needspace} 
\usepackage{array}
\usepackage{xcolor}
\newcommand{\safeincludegraphics}[2][]{%
  \IfFileExists{#2}{%
    \includegraphics[#1]{#2}%
  }{%
    \begin{minipage}[c][0.28\textwidth][c]{0.9\textwidth}
      \centering\small Missing figure file: \texttt{#2}
    \end{minipage}%
  }%
}

\SetAlFnt{\small}
\SetAlCapFnt{\small}
\SetAlCapNameFnt{\small}
\SetAlCapHSkip{0pt}
\IncMargin{-\parindent}

\definecolor{revcolor}{RGB}{160,0,0}

\usepackage{xspace}
\usepackage{todonotes}
\usepackage{mathtools}
\usepackage{amssymb}
\usepackage[hidelinks]{hyperref}

\usepackage[a4paper, left=1in, right=1in, top=1in, bottom=1in]{geometry}

\usepackage{tikz}
\usetikzlibrary{decorations.pathreplacing}

\let\lncsproof\proof
\let\endlncsproof\endproof
\renewenvironment{proof}{\lncsproof}{\qed\endlncsproof}

\newcommand{\SW}{\mathrm{SW}}

\newcommand{\fratio}{fair-ratio\xspace}

\newcommand{\AR}{\mathrm{FR}}
\newcommand{\FR}{\ensuremath{\mathrm{FR}}}

\newcommand{\DR}{\ensuremath{\mathrm{DRF}}\xspace}
\newcommand{\FOne}{\ensuremath{\mathrm{UNB}}\xspace}

\newcommand{\FTwos}{\ensuremath{\mathrm{BAL}^*}\xspace}
\newcommand{\ASF}{\ensuremath{\mathrm{ASF}}\xspace}
\newcommand{\AOne}{\ensuremath{\mathrm{ASF_1}}\xspace}
\newcommand{\ATwo}{\ensuremath{\mathrm{ASF_2}}\xspace}

\newcommand{\Hyb}{\ensuremath{\mathrm{Hybrid}}\xspace}

\newcommand{\OPT}{\ensuremath{\mathrm{OPT}}\xspace}
\newcommand{\CRB}{\ensuremath{\mathrm{CRB}}\xspace}

\newcommand{\citep}[1]{\cite{#1}}
\newcommand{\citet}[1]{\cite{#1}}

\title{The Price of Strategyproofness in Fair Multi-Resource Allocation for Cloud Computing}
\titlerunning{The Price of Strategyproofness in Fair Multi-Resource Allocation}

\author{
Yunpeng Lou 
\and
Junjie Luo
}
\authorrunning{Y. Lou and J. Luo}
\institute{School of Mathematics and Statistics, Beijing Jiaotong University, Beijing 100044, China\\
\email{25121474@bjtu.edu.cn, jjluo1@bjtu.edu.cn}}

\begin{document}

\maketitle

\begin{abstract}

We study fair and strategy-proof allocation of multiple divisible resources with Leontief utilities, motivated by cloud computing. The canonical mechanism, Dominant Resource Fairness (DRF), satisfies sharing incentive (SI), envy-freeness (EF), strategy-proofness (SP), and Pareto optimality (PO), but can be highly inefficient in terms of utilitarian social welfare. Under the classical approximation benchmark, no mechanism satisfying even one of SI, EF, and SP can improve on the trivial worst-case guarantee. We therefore adopt the recently introduced \emph{fair-ratio} benchmark, which compares a mechanism only with the welfare-maximizing allocation that itself satisfies SI and EF.

For two resources, we first introduce Adaptive-Speed Fairness (ASF), a unified parametric framework that captures previous mechanisms as special or boundary cases. Every ASF mechanism satisfies SI, EF, and PO, and we derive a general sufficient condition that guarantees SP. Optimizing within this framework yields a strategy-proof mechanism with asymptotic fair-ratio $2/(2\sqrt{2}-1)\approx 1.094$, substantially improving the previous best guarantee $3-\sqrt{3}\approx1.268$. We complement this upper bound with a lower bound of $1.0789$ for all ASF mechanisms, showing that our mechanism is close to optimal within this framework, but also that no ASF mechanism can approach the fair optimum as the population grows. To overcome this limitation, we introduce Corrected Resource Balancing (CRB), which uses the full resource-load structure together with a one-agent incentive correction. CRB satisfies SI, EF, SP, and PO and achieves fair-ratio at most $1+1/n$. Together with a general $1+\Omega(1/n)$ lower bound, this establishes the optimal asymptotic order $1+\Theta(1/n)$ for two resources. Experiments on synthetic and Google trace-based instances further demonstrate that CRB remains close to the fair optimum in practice.

Finally, we establish a sharp dimensional boundary. For every $m\ge3$, any deterministic mechanism satisfying SI and SP has fair-ratio exactly $m$. The same factor-$m$ lower bound continues to hold for randomized mechanisms satisfying ex-post SI and truthfulness in expectation.

\end{abstract}


\section{Introduction}

A central challenge in multi-resource allocation is to balance fairness, incentive compatibility, and efficiency when agents are strategic. This challenge is particularly salient in cloud computing, where users submit jobs with heterogeneous demands across resources such as CPU, memory, and bandwidth. An allocation rule should provide meaningful fairness guarantees and remain robust to strategic misreporting of these demands, without unnecessarily sacrificing efficiency. Designing mechanisms that simultaneously address these requirements is therefore a fundamental problem in multi-resource allocation.

A canonical solution to this problem is the Dominant Resource Fairness (DRF) mechanism proposed by Ghodsi et al.~\citep{Ghodsi2011}. DRF assumes Leontief preferences, meaning that each agent requires different resources in fixed proportions. For each agent, the resource demanded in the largest fraction of its total supply is called its dominant resource. Under such preferences, DRF generalizes max-min fairness to the multi-resource setting by maximizing and equalizing agents' dominant-resource shares. Remarkably, DRF simultaneously satisfies four fundamental properties for cloud resource allocation:
(1) Sharing incentive (SI): every user receives at least as much utility as under an equal division of every resource;
(2) Envy-freeness (EF): no user prefers another user's allocation to its own;
(3) Strategy-proofness (SP): no user can benefit from misreporting its resource demands; and
(4) Pareto optimality (PO): no feasible allocation can make one user better off without making another worse off.
Consequently, DRF has become a standard benchmark for multi-resource fairness, and SI, EF, SP, and PO have become the principal desiderata for evaluating alternative mechanisms.  

Despite these strong fairness and incentive guarantees, DRF can be highly inefficient in terms of utilitarian social welfare, defined as the sum of agents' utilities. 
Under the classical benchmark against the unconstrained social-welfare optimum, \citet{Parkes2015} showed that DRF can lose a factor $m$, where $m$ is the number of resource types. 
More fundamentally, this loss is unavoidable: no mechanism satisfying even one of SI, EF, and SP can guarantee an approximation ratio better than \(m\).
Thus, under the classical benchmark, retaining even a basic fairness or incentive requirement leaves no room for improving the worst-case welfare guarantee of DRF. 

To overcome this barrier, \citet{BeiLL22,BeiLL26} introduced the \fratio benchmark, which compares a mechanism with the welfare-maximizing allocation that is itself required to satisfy SI and EF. 
The key motivation is that fairness is typically a hard constraint in applications. Benchmarking a fair mechanism against an unconstrained optimum that need not satisfy the same fairness requirements can obscure meaningful efficiency differences among fair mechanisms.
Indeed, under the classical benchmark, every mechanism satisfying SI has worst-case approximation ratio of $m$, whereas under the \fratio benchmark, mechanisms satisfying the same four properties as DRF can have sharply different guarantees.
By restricting the benchmark to fair allocations, the \fratio therefore provides a more informative performance measure for mechanism design and enables meaningful efficiency comparisons among fair and strategy-proof mechanisms.

Within this framework, \citet{BeiLL26} showed that DRF can be far from optimal: its fair-ratio remains \(m\) for \(m\) resource types. 
How much this factor-\(m\) guarantee can be improved, however, depends sharply on the number of resources.
With two resources, substantial improvements are possible.
\citet{BeiLL26} proposed \(\FOne\) and \(\FTwos\), with asymptotic fair-ratios \(3/2\) and \(4/3\), respectively, and further combined them into the \Hyb mechanism, which achieves $3-\sqrt{3}\approx1.268$.
In contrast, for every \(m\ge 4\), no mechanism satisfying SI, EF, SP, and PO can improve upon the factor-\(m\) guarantee of DRF. The only unresolved case is \(m=3\), for which the best known lower bound is \(2\), leaving a gap to the factor-\(3\) upper bound.

This leaves a sharp but incomplete picture. For two resources, how close can a mechanism satisfying SI, EF, SP, and PO come to the fair-welfare optimum? Existing mechanisms demonstrate that substantial improvement over DRF is possible, but their guarantees remain bounded away from \(1\), and it is unclear whether this reflects a fundamental price of strategyproofness or merely limitations of the existing designs. Beyond two resources, where exactly does the impossibility begin? Does the full factor-\(m\) barrier already hold for \(m=3\), and can randomization circumvent it for \(m\ge 3\)? These are the two questions we address.

\subsection{Our Contributions}

We study mechanisms that retain the four standard properties of DRF (SI, EF, SP, and PO) while improving efficiency under the \fratio benchmark. Our results sharpen the picture initiated by \citet{BeiLL26}: with two resources, the welfare loss from strategy-proofness can vanish with the population size, whereas with three or more resources no nontrivial improvement over the trivial guarantee is possible.

\emph{Adaptive-speed mechanisms for two resources.}
We first develop the Adaptive-Speed Fairness (\ASF) framework, which unifies and extends the two-step water-filling mechanisms of \citet{BeiLL22,BeiLL26}. After assigning the SI baseline, \ASF coordinates the two dominant-resource groups through a relative speed that depends on population and remaining-resource imbalance. We derive a general sufficient condition for SP and obtain improved mechanisms, which improve the asymptotic \fratio to $\frac{2}{2\sqrt2-1}\approx1.094$.
We complement these guarantees with an asymptotic lower bound $\frac{191-2\sqrt{514}}{135}\approx1.0789$ for every \ASF mechanism. Thus \(\ATwo\) is close to optimal within this natural aggregate-state framework, but no further refinement of such a speed rule can make the \fratio approach one.

\emph{An optimal-order strategyproof mechanism.}
The \ASF lower bound suggests that further improvement requires information beyond aggregate population and residual-capacity statistics. We therefore introduce Corrected Resource Balancing (\CRB), our main two-resource mechanism. The key idea is to start from the resource-balancing characterization of the fair-welfare optimum and add a one-agent correction that removes the incentive to manipulate the balancing point. Using the full resource-load functions, \(\CRB\) satisfies SI, EF, SP, and PO and loses at most \(1/n\) welfare relative to the fair optimum. Consequently, $\FR(\CRB)\le 1+\frac1n$.
Together with our \(1+\Omega(1/n)\) lower bound, this determines the optimal asymptotic order \(1+\Theta(1/n)\) for two resources. Experiments on synthetic and Google trace-generated instances further show that \CRB remains close to the fair optimum and consistently improves on the earlier mechanisms tested.

\emph{A sharp boundary beyond two resources.}
Finally, we show that the positive two-resource picture does not extend to higher dimensions. For every fixed \(m\ge3\), every deterministic mechanism satisfying SI and SP has \fratio exactly \(m\); hence the same is true when EF and PO are also required. The \(m=3\) case was independently resolved in concurrent work by \citet{fithiansolving}, while our proof gives a unified construction for all \(m\ge3\). We further extend the impossibility to randomized mechanisms: every mechanism satisfying ex-post SI and truthfulness in expectation also has \fratio exactly \(m\). Thus randomization cannot overcome the \(m\)-barrier.

\paragraph{Organization.}
Section~\ref{sec:prelim} introduces the model and the fair-ratio benchmark. Section~\ref{sec:asf} develops ASF and explains its limitation. Section~\ref{sec:crb} gives the general two-resource SP lower bound, presents \CRB, and analyzes its properties and fair-ratio. Section~\ref{sec:experiments} reports the experimental results for two resources. Section~\ref{sec:beyond-two} establishes the higher-dimensional barrier. Section~\ref{sec:con} concludes.

\subsection{Further Related Work}
\label{sec:related-work}

Fair and efficient multi-resource allocation has attracted substantial attention across computer science, economics, and operations research since the seminal work of \citet{Ghodsi2011}. 
A large body of work improves efficiency by relaxing or rebalancing the fairness objectives of DRF \citep{Joe-Wong2013,Bonald2014,Bonald2015,Jiang2021,Jin2016,Tang2016,Tang2020,Grandl2014}. 
These mechanisms address the fairness--efficiency trade-off in important practical settings, and many preserve properties such as SI, EF, and PO, but generally sacrifice SP. 
For Leontief utilities, alternative fairness notions and mechanisms have also been studied \citep{Dolev2012,Gutman2012}, together with axiomatic characterizations of related solution concepts \citep{Nicolo2004,Friedman2011,Li2013}. 
The classical DRF model has also been extended in several directions, including multiple machines \citep{Tahir2015,Wang2016,Friedman2014,Wang2014}, bounded demands \citep{Li2017,Narayana2021}, and dynamic settings \citep{Kash2014,fikioris2024incentives,Friedman2017a}. 
In contrast to these extensions, \citet{BeiLL26} revisited the basic DRF model itself and introduced the \fratio benchmark, showing that strategy-proof mechanisms can substantially outperform DRF when evaluated against the best fair allocation. 
We build on this perspective by remaining in the classical DRF setting and seeking stronger efficiency guarantees under the \fratio benchmark.

Our work is also related to the broader literature on approximate mechanism design without money, initiated by \citet{ProcacciaT13}. 
Several papers study strategy-proof allocation without payments or priors under additive or linear utilities \citep{GuoC10,HanSTZ11,Cheung16}. 
Closer to our preference model, \citet{ColeGG13} study divisible-good allocation with linear or Leontief utilities and introduce the Partial Allocation mechanism, which achieves an \(e\)-approximation for Nash social welfare. 
Strategy-proof mechanisms combining fairness and efficiency have also been studied and characterized in indivisible-resource settings \citep{Amanatidis2017Truthful,BabaioffMorag2025Truthful}. 
These works share our focus on incentive compatibility without monetary transfers, but differ from ours in both the allocation model and efficiency objective.
We consider divisible multi-resource allocation with Leontief utilities, require the DRF properties SI, EF, SP, and PO simultaneously, and evaluate utilitarian efficiency relative to the best SI--EF allocation.

\paragraph{Concurrent work.}
Fithian, Niazadeh, and Nuti~\cite{fithiansolving} independently resolve the \(m=3\) case, proving that every deterministic mechanism satisfying SI and SP has fair-ratio \(3\).
Their result was developed as part of an AI-assisted project for solving open problems in operations research.
We obtained our \(m=3\) result independently and became aware of their work only after submitting the present paper to WINE 2026.
Interestingly, both proofs were developed with AI assistance and independently arrived at essentially the same harmonic-weight argument.
Our treatment additionally provides a unified construction for all \(m\ge3\) and extends the impossibility to randomized mechanisms satisfying ex-post SI and truthfulness in expectation.

\section{Preliminaries}\label{sec:prelim}

We first introduce the standard multi-resource allocation model and the notation used throughout the paper, following the conventions of the DRF literature (e.g., \citet{Parkes2015,BeiLL26}).

\paragraph{Multi-resource allocation.}
Let $N=\{1,\dots,n\}$ be the set of agents and let $R$ be the set of resource types, where $|R|=m$.
Each resource is divisible and has total supply normalized to one.
Each agent $i$ has a (normalized) \emph{demand vector} $\mathbf d_i=(d_{i1},\dots,d_{im})\in(0,1]^m$, with $\max_{r\in R}d_{ir}=1$.
A \emph{dominant resource} of agent $i$ is any resource $r_i^*\in\arg\max_{r\in R}d_{ir}$, and hence $d_{i r_i^*}=1$.
Throughout the formal model, we assume \(d_{ir}>0\) for every agent \(i\) and resource \(r\). In examples, we sometimes write zero entries for notational simplicity; these should be understood as limits as the corresponding positive demands tend to zero.
An instance is represented by an $n\times m$ matrix $\mathbf{I}$ whose $i$th row is $\mathbf{d}_i$.

An \emph{allocation} is an $n\times m$ matrix $\mathbf{A}$, where $A_{ir}$ denotes the fraction of resource $r$ assigned to agent~$i$.
An allocation is \emph{feasible} if
$\sum_{i\in N} A_{ir}\le 1$ for every resource $r\in R$.
Agents have \emph{Leontief utilities}: the utility of agent $i$ under bundle 
$\mathbf{A}_i$ is $u_i(\mathbf{A}_i)=\min_{r \in R} A_{ir}/d_{ir}$.
An allocation is \emph{non-wasteful} if each agent receives resources in proportion to its demand vector, i.e., for every $i$ there exists $y_i\ge 0$ such that $A_{ir}=y_i d_{ir}$ for all $r\in R$.

\paragraph{Mechanism.}
A \emph{mechanism} \(M\) is a function that maps each instance to a feasible allocation.
We write \(M_i(\mathbf{I})\) for the bundle assigned to agent \(i\) on instance \(\mathbf{I}\).
Except for Section~\ref{sec:beyond-two}, we always consider deterministic mechanisms.
We restrict attention to mechanisms that always output non-wasteful allocations.
The \emph{Dominant Resource Fairness} (DRF) rule of \citet{Ghodsi2011} selects the allocation that equalizes agents' dominant shares while respecting all resource constraints.
Let $x$ denote the common dominant share assigned to every agent.
DRF chooses the largest feasible value of $x$, namely $x^*=1/\max_{r\in R}\sum_{i\in N}d_{ir}$, and assigns bundle $x^*\mathbf{d}_i$ to each agent $i$.

\paragraph{Example~1.}
Consider an example with two resource types and three agents whose demand vectors are \(\mathbf d_1=(1,2/5)\), \(\mathbf d_2=(1,1/5)\), and \(\mathbf d_3=(1/5,1)\).
Agents 1 and 2 have resource~1 as their dominant resource, while agent 3 has resource~2 as its dominant resource. Under \DR, the common dominant share is \(x^*=5/11\), so the final allocation is
\(\mathbf A_1=(5/11,2/11)\), \(\mathbf A_2=(5/11,1/11)\), and \(\mathbf A_3=(1/11,5/11)\), with each agent receiving utility $5/11$ and the first resource exhausted.

\paragraph{Properties.}
We evaluate mechanisms using four standard axioms that are central to DRF and multi-resource fairness.
As shown by \citet{Ghodsi2011}, DRF satisfies all four.

\begin{definition}[Sharing Incentive (SI)]
An allocation $\mathbf{A}$ satisfies \emph{sharing incentive} if every agent receives utility at least $1/n$, i.e., $u_i(\mathbf{A}_i)\ge 1/n$ for all $i\in N$.
\end{definition}

\begin{definition}[Envy-Freeness (EF)]
An allocation $\mathbf{A}$ is \emph{envy-free} if no agent prefers another agent's bundle, i.e., $u_i(\mathbf{A}_i)\ge u_i(\mathbf{A}_j)$ for all $i,j\in N$.
\end{definition}

\begin{definition}[Pareto Optimality (PO)]
An allocation $\mathbf{A}$ is \emph{Pareto optimal} if there is no feasible allocation $\mathbf{A}'$ that makes some agent strictly better off without making any agent worse off.
\end{definition}

A mechanism is SI, EF, or PO if, on every instance, its output satisfies the corresponding property.

\begin{definition}[Strategy-Proofness (SP)]
A mechanism \(M\) is \emph{strategy-proof} if no agent can improve its utility by misreporting its demand vector.
Formally, for every instance $\mathbf{I}$, every agent $i$, and every alternative normalized report $\mathbf{d}'_i\in(0,1]^m$ with $\max_{r\in R}d'_{ir}=1$, we have $u_i(M_i(\mathbf{I})) \ge u_i\bigl(M_i(\mathbf{I}')\bigr)$,
where $\mathbf{I}'$ is obtained from $\mathbf{I}$ by replacing agent $i$'s row with $\mathbf{d}'_i$.
The utility after misreporting is evaluated using the true demand vector.
\end{definition}

\paragraph{Fair-ratio benchmark.}

The \emph{social welfare} of an allocation $\mathbf{A}$ is the total utility obtained by all agents: $\SW(\mathbf{A}) = \sum_{i\in N} u_i(\mathbf{A}_i)$.
We use the \fratio benchmark of \citet{BeiLL22}, which compares a mechanism with the best allocation that satisfies the fairness constraints SI and EF, rather than with the unconstrained welfare optimum. For an instance $\mathbf I$, write
\[
\OPT(\mathbf I)
=
\max_{\mathbf A:\,\mathrm{feasible,SI,EF}}\SW(\mathbf A).
\]


\begin{definition}\label{prelim-FR}
The \emph{\fratio} of a mechanism \(M\) is the worst-case ratio, over all instances $\mathbf I$, between the fair-optimal welfare $\OPT(\mathbf I)$ and the social welfare achieved by \(M(\mathbf I)\):
\[
\AR(M)
=
\sup_{\mathbf I\in\mathcal I}
\frac{\OPT(\mathbf I)}{\SW(M(\mathbf I))}.
\]
The supremum allows for worst-case ratios approached only along
sequences of positive-demand instances.
\end{definition}

\section{Adaptive-Speed Fairness}\label{sec:asf}

We begin by developing the Adaptive-Speed Fairness (\ASF) framework for two-resource allocation, generalizing the earlier mechanisms of \citet{BeiLL22,BeiLL26}. Starting from an SI baseline, \ASF uses water filling within each dominant-resource group and coordinates the groups through a relative growth speed determined by population shares and residual capacities. We establish a common sufficient condition for strategyproofness that guides the choice of an adaptive speed rule. We construct two mechanisms with improved welfare guarantees, the stronger of which achieves an asymptotic \fratio below 1.1. Nevertheless, every \ASF mechanism based on these aggregate statistics has a worst-case \fratio bounded away from one, motivating the use of the full resource-load geometry in \CRB in the next section.

\subsection{Canonical water-filling allocations}
\label{sec:canonical}

We first introduce a canonical family of water-filling allocations, parameterized by one water level for each dominant-resource group. This family underlies both \ASF and CRB and, as we show below, contains a fair-optimal allocation.
Put $h=1/n$, and partition agents by their dominant resource:
\[
G_1=\{i:\mathbf d_i=(1,a_i)\},
\qquad
G_2=\{j:\mathbf d_j=(b_j,1)\},
\]
using the fixed convention that $(1,1)\in G_1$. 
Thus $0<a_i\le1$ and $0<b_j<1$.
Starting from the SI baseline, where each agent receives \(h\mathbf d_i\), water filling raises the smallest non-dominant allocations within each group together. At water level \(p\), an agent in \(G_1\) receives either its baseline amount \(a_i h\) or \(p\) of resource 2, whichever is larger; its resource-1 allocation is determined by proportionality. The interpretation of \(q\) for \(G_2\) is symmetric. This yields the following canonical representation.

\begin{definition}[Canonical allocation]\label{def:canonical}
For water levels $p,q\ge0$, define
\[
x_i(p)=\max\{h,p/a_i\}\quad(i\in G_1),
\qquad
y_j(q)=\max\{h,q/b_j\}\quad(j\in G_2).
\]
The canonical allocation $A^{\mathrm c}(p,q)$ assigns the bundle $x_i(p)(1,a_i)$ to each $i\in G_1$ and $y_j(q)(b_j,1)$ to each $j\in G_2$. 
\end{definition}

\begin{lemma}
\label{lem:canonical-fairness}
Every feasible canonical allocation $A^{\mathrm c}(p,q)$ with $p,q\in[0,h]$ satisfies SI and EF.
\end{lemma}

\begin{proof}
Every agent's utility is at least $h$, giving SI. For $i,k\in G_1$, if $x_k(p)\le x_i(p)$, agent $i$ cannot envy $k$ because of resource~1. Otherwise, $x_k(p)>h$, so $a_kx_k(p)=p\le a_ix_i(p)$, ruling out envy because of resource~2. The argument within $G_2$ is symmetric. Across groups, every other-group bundle contains at most $h$ of an agent's dominant resource, while its own utility is at least $h$. Hence cross-group envy is also impossible.
\end{proof}

We justify this representation by showing that every feasible SI--EF allocation can be transformed into a canonical allocation with the same group total utilities and no greater resource loads. The proof is in Appendix~\ref{app:canonical}.

\begin{lemma}[Canonicalization]\label{lem:canonical}
For every feasible SI--EF allocation $A$, there exist $p,q\in[0,h]$ such that $A^{\mathrm c}(p,q)$ has the same total utility in each group as $A$, weakly decreases the load on each resource, and is feasible, SI, and EF. Consequently, a fair-optimal allocation can always be chosen canonical.
\end{lemma}

Thus restricting attention to canonical allocations guarantees SI and EF while entailing no loss in the fair benchmark. The remaining design choice is how to coordinate the two water levels.

\subsection{The ASF framework}\label{sec:framework}

We now define the \ASF framework, which unifies and generalizes the earlier mechanisms of \citet{BeiLL22,BeiLL26}. The framework consists of two steps: Step~1 assigns the SI baseline, and Step~2 increases the two groups' utilities along their canonical trajectories at a fixed relative speed until a resource is exhausted.

\paragraph{Baseline and state variables.}
Step~1 assigns each agent $h\mathbf d_i$. If a group is empty or this baseline exhausts a resource, the mechanism returns it. Otherwise, both groups are nonempty and both residual capacities are positive. Relabel the groups and resources together if necessary so that $|G_1|\ge|G_2|$, and define the \emph{minority population ratio}
\[
\alpha=|G_2|/n\in(0,1/2].
\]
The residual capacities after Step~1 are
\begin{equation}\label{two-residuals}
R_1=\alpha-h\sum_{j\in G_2}b_j,\qquad
R_2=1-\alpha-h\sum_{i\in G_1}a_i.
\end{equation}
Following \FTwos, add back one minimum non-dominant baseline contribution to each residual capacity:
\begin{equation}\label{asf-augmented}
R_1^*=R_1+h\min_{j\in G_2}b_j,\qquad
R_2^*=R_2+h\min_{i\in G_1}a_i.
\end{equation}
This correction helps control an agent's influence on the selected speed and is important for SP. \ASF chooses its speed from the minority population ratio and augmented residual capacities, summarized by the state $(\alpha,R_1^*,R_2^*)$.

\paragraph{Intergroup coordination.}
For a canonical allocation $A^{\mathrm c}(p,q)$, define the groups' utility increments above SI by
\[
\Delta S_1(p)=\sum_{i\in G_1}\bigl(x_i(p)-h\bigr),\qquad
\Delta S_2(q)=\sum_{j\in G_2}\bigl(y_j(q)-h\bigr).
\]
We suppress the arguments and write $\Delta S_1,\Delta S_2$ throughout the \ASF analysis.
An \ASF speed rule is a positive function $\Theta(\alpha,R_1^*,R_2^*)$ on the feasible state space. In Step~2, set $\theta=\Theta(\alpha,R_1^*,R_2^*)$ and increase the group utilities along their canonical trajectories subject to
\begin{equation}\label{eq:asf-speed}
\Delta S_1=\theta\Delta S_2,
\end{equation}
until a resource is exhausted. The speed remains fixed during this execution. A larger $\theta$ directs more additional utility to $G_1$. The balanced rule \FTwos uses $\theta=R_1^*/R_2^*$. The minority-only rule \FOne is recovered by extending the framework to the boundary case $\theta=0$.

The following lemma records the basic properties and running time; its proof is in Appendix~\ref{app:asf-implementation}.

\begin{lemma}\label{lem:ASF-3-properties}
Every \ASF mechanism satisfies SI, EF, and PO. Given its speed, it can be implemented in $O(n^2)$ time.
\end{lemma}

\paragraph{Example~1 continued.}
We illustrate the Step~2 rule on the instance from Example~1; see Figure~\ref{fig:example}. Recall that \(\mathbf d_1=(1,2/5)\), \(\mathbf d_2=(1,1/5)\), and \(\mathbf d_3=(1/5,1)\). Thus \(G_1=\{1,2\}\), \(G_2=\{3\}\), and \(\alpha=1/3\). After Step~1, the remaining resources are \((R_1,R_2)=(4/15,7/15)\), and the augmented remaining resources are \(R_1^*=1/3\) and \(R_2^*=8/15\).
The \(\DR\) allocation admits an equivalent two-step interpretation: after the
baseline SI allocation in Step~1, Step~2 increases all agents' dominant shares
at the same rate until a resource is exhausted. Unlike \(\ASF\), however, this
second step does not use within-group water filling. The resulting utility is
\(5/11\) for every agent, giving social welfare \(15/11\approx 1.36\).
The boundary rule \(\FOne\) assigns all Step~2 growth to agent~3, giving \(\mathbf A_3=(4/25,4/5)\) and social welfare \(22/15\approx 1.47\).
Next consider \(\FTwos\), which sets \(\theta=R_1^*/R_2^*=5/8\). Initially, the active sets are \(P_1=\{2\}\) and \(P_2=\{3\}\). Hence Step~2 increases agents 2 and 3 so that the ratio of their dominant-share increments is $5/8$. Resource~1 is exhausted before agent~1 becomes active. Thus \(\FTwos\) outputs \(\mathbf A_1=(1/3,2/15)\), \(\mathbf A_2=(53/99,53/495)\), and \(\mathbf A_3=(13/99,65/99)\), with social welfare \(151/99\approx1.53\).
The adaptive rule \ATwo introduced in Section~\ref{sec:mechanisms} chooses a smaller speed ratio on this instance, shifting more Step~2 growth to agent~3 and increasing the social welfare to approximately \(1.55\).

\begin{figure}[t]
    \centering

    \begingroup
    \def\deepred{red!80!black!100}
    \def\lightred{red!45}
    \def\deepblue{blue!70}
    \def\lightblue{blue!45}
    \def\deepgreen{green!55!black}
    \def\lightgreen{green!35}

    \begin{tikzpicture}[scale=3.80, every node/.style={font=\scriptsize}]
        \def\w{0.22}
        \def\rOneBot{0.2/3}
        \def\rOneTop{1/3}
        \def\rTwoBot{1/3}
        \def\rTwoTop{0.8}
        \def\sqw{0.13}
        \def\sqh{0.08}

        \begin{scope}[shift={(-0.12,0)}]
            \filldraw[fill=\lightred] (0,0) rectangle (\w,0.2/3);
            \filldraw[fill=white] (0,0.2/3) rectangle (\w,1/3);
            \filldraw[fill=\deepgreen] (0,1/3) rectangle (\w,2/3);
            \filldraw[fill=\deepblue] (0,2/3) rectangle (\w,1);

            \filldraw[fill=\deepred] (\w,0) rectangle (2*\w,1/3);
            \filldraw[fill=white] (\w,1/3) rectangle (2*\w,0.8);
            \filldraw[fill=\lightgreen] (\w,0.8) rectangle (2*\w,0.8+0.2/3);
            \filldraw[fill=\lightblue] (\w,0.8+0.2/3) rectangle (2*\w,1);

            \node[anchor=south,font=\footnotesize] at (\w,1.025) {(a) Step 1};

            \draw[decorate,decoration={brace,raise=3pt},line width=0.35pt] (0,0.2/3) -- (0,1/3);
            \node[anchor=east,font=\scriptsize] at (-0.035,0.205) {$R_1$};

            \draw[decorate,decoration={brace,raise=17pt},line width=0.35pt] (0,0) -- (0,1/3);
            \node[anchor=east,font=\scriptsize] at (-0.155,0.16) {$R^*_1$};

            \draw[decorate,decoration={brace,raise=3pt},line width=0.35pt] (0,1/3) -- (0,1);
            \node[anchor=east,font=\scriptsize] at (-0.035,2/3) {$1-\alpha$};

            \draw[decorate,decoration={brace,mirror,raise=3pt},line width=0.35pt] (2*\w,1/3) -- (2*\w,0.8);
            \node[anchor=west,font=\scriptsize] at (2*\w+0.035,0.57) {$R_2$};

            \draw[decorate,decoration={brace,mirror,raise=17pt},line width=0.35pt] (2*\w,1/3) -- (2*\w,0.8+0.2/3);
            \node[anchor=west,font=\scriptsize] at (2*\w+0.155,0.60) {$R^*_2$};

            \draw[decorate,decoration={brace,mirror,raise=3pt},line width=0.35pt] (2*\w,0) -- (2*\w,1/3);
            \node[anchor=west,font=\scriptsize] at (2*\w+0.035,1/6) {$\alpha$};

            \filldraw[fill=\deepblue] (-0.20,-0.16) rectangle (-0.10,-0.06);
            \filldraw[fill=\lightblue] (-0.10,-0.16) rectangle (0,-0.06);
            \node[anchor=north,font=\scriptsize] at (-0.10,-0.18) {Agent $1$};

            \filldraw[fill=\deepgreen] ({\w-0.10},-0.16) rectangle (\w,-0.06);
            \filldraw[fill=\lightgreen] (\w,-0.16) rectangle ({\w+0.10},-0.06);
            \node[anchor=north,font=\scriptsize] at (\w,-0.18) {Agent $2$};

            \filldraw[fill=\lightred] (2*\w,-0.16) rectangle ({2*\w+0.10},-0.06);
            \filldraw[fill=\deepred] ({2*\w+0.10},-0.16) rectangle ({2*\w+0.20},-0.06);
            \node[anchor=north,font=\scriptsize] at ({2*\w+0.10},-0.18) {Agent $3$};
        \end{scope}

        \def\residualboxes{%
            \filldraw[fill=white] (0,\rOneBot) rectangle (\w,\rOneTop);
            \filldraw[fill=white] (\w,\rTwoBot) rectangle (2*\w,\rTwoTop);
        }
        \def\sq#1#2#3{\filldraw[fill=#3] (#1,#2) rectangle +(\sqw,\sqh);}

        \begin{scope}[shift={(0.92,0)}]
            \residualboxes

            \filldraw[fill=\lightred] (0,0.2/3) rectangle (\w,1/11);
            \filldraw[fill=\deepgreen] (0,1/11) rectangle (\w,7/33);
            \filldraw[fill=\deepblue] (0,7/33) rectangle (\w,1/3);

            \filldraw[fill=\deepred] (\w,1/3) rectangle (2*\w,5/11);
            \filldraw[fill=\lightgreen] (\w,8/11) rectangle (2*\w,0.8-8/165);
            \filldraw[fill=\lightblue] (\w,0.8-8/165) rectangle (2*\w,0.8);

            \node[anchor=south,font=\footnotesize] at (\w,1.025) {(b) DRF Step 2};

            \sq{-0.065}{-0.168}{\deepblue}
            \node[anchor=base,font=\scriptsize] at (0.11,-0.144) {$=$};
            \sq{0.155}{-0.168}{\deepgreen}
            \node[anchor=base,font=\scriptsize] at (0.33,-0.144) {$=$};
            \sq{0.375}{-0.168}{\deepred}
        \end{scope}

        \begin{scope}[shift={(1.62,0)}]
            \residualboxes

            \filldraw[fill=\lightred] (0,0.2/3) rectangle (\w,0.16);
            \filldraw[fill=\deepred] (\w,1/3) rectangle (2*\w,0.8);

            \node[anchor=south,font=\footnotesize] at (\w,1.025) {(c) UNB Step 2};

            \begin{scope}[shift={(\sqw/2,0)}]
                \sq{\w-0.22-\sqw/2}{-0.168}{\deepblue}
                \node[anchor=base,font=\scriptsize] at ({\w-0.11},-0.144) {$=$};
                \sq{\w-\sqw/2}{-0.168}{\deepgreen}
                \node[anchor=base,font=\scriptsize] at ({\w+0.11},-0.144) {$=$};
                \node[anchor=base west,inner sep=0pt,outer sep=0pt,font=\scriptsize] at ({\w+0.22-\sqw/2},-0.144) {$0$};
            \end{scope}
        \end{scope}

        \begin{scope}[shift={(2.32,0)}]
            \residualboxes

            \filldraw[fill=\lightred] (0,0.2/3) rectangle (\w,13/99);
            \filldraw[fill=\deepgreen] (0,13/99) rectangle (\w,1/3);

            \filldraw[fill=\deepred] (\w,1/3) rectangle (2*\w,65/99);
            \filldraw[fill=\lightgreen] (\w,0.8-4/99) rectangle (2*\w,0.8);

            \node[anchor=south,font=\footnotesize] at (\w,1.025) {(d) BAL$^*$ Step 2};

            \sq{0.020}{-0.095}{\deepgreen}
            \draw[line width=0.4pt] (0.010,-0.128) -- (0.160,-0.128);
            \sq{0.020}{-0.241}{\deepred}
            \node[anchor=base,inner sep=0pt,outer sep=0pt,font=\scriptsize] at (\w,-0.144) {$=$};
            \node[anchor=base west,inner sep=0pt,outer sep=0pt,font=\scriptsize] at ({\w+0.06},-0.144) {$\dfrac{R_1^*}{R_2^*}$};
        \end{scope}

        \begin{scope}[shift={(3.02,0)}]
            \residualboxes

            \filldraw[fill=\lightred] (0,0.2/3) rectangle (\w,0.1384665561453043);
            \filldraw[fill=\deepgreen] (0,0.1384665561453043) rectangle (\w,1/3);

            \filldraw[fill=\deepred] (\w,1/3) rectangle (2*\w,0.6923327807265214);
            \filldraw[fill=\lightgreen] (\w,0.7610266445623942) rectangle (2*\w,0.8);
            
            \node[anchor=south,font=\footnotesize] at (\w,1.025) {(e) ASF$_2$ Step 2};

            \sq{0.020}{-0.095}{\deepgreen}
            \draw[line width=0.4pt] (0.010,-0.128) -- (0.160,-0.128);
            \sq{0.020}{-0.241}{\deepred}
            \node[anchor=base,inner sep=0pt,outer sep=0pt,font=\scriptsize] at (\w,-0.144) {$=$};
            \node[anchor=base west,inner sep=0pt,outer sep=0pt,font=\scriptsize] at ({\w+0.06},-0.144)
                {$\Theta(\alpha,R_1^*,R_2^*)$};
        \end{scope}

    \end{tikzpicture}
    \endgroup

    \caption{
    Illustration of the two-step allocation process for Example~1 with
    $\mathbf d_1=(1,2/5)$, $\mathbf d_2=(1,1/5)$, and $\mathbf d_3=(1/5,1)$.
    Panel (a) shows the allocation after Step~1.
    Panels (b)--(e) illustrate the allocation of the remaining resources
    $R_1$ and $R_2$ in Step~2 under $\DR$, $\FOne$, $\FTwos$, and $\ATwo$, respectively; $\ATwo$ is defined in Section~\ref{sec:mechanisms}.
    }

    \label{fig:example}
\end{figure}
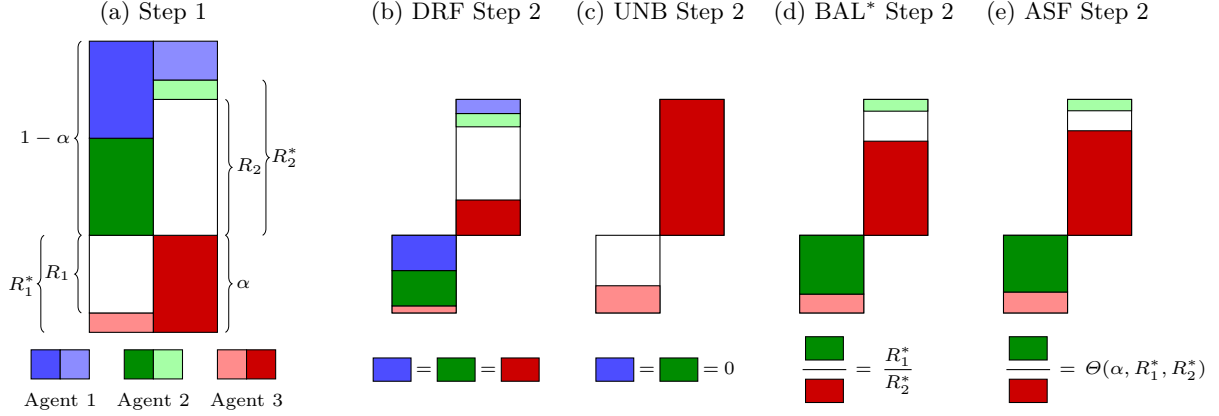

\paragraph{A sufficient condition for strategyproofness.}

A report can change the residual state and hence the selected speed. We control this response through a sufficient condition extending the SP analysis of \FTwos in \citep{BeiLL26}.

Call $\Theta$ \emph{monotone} if, for fixed $\alpha$, it is nondecreasing in $R_1^*$ and nonincreasing in $R_2^*$. Define
\begin{equation}\label{eq:finite-c12}
c_1=\left(1-\frac{R_2^*-h}{1-\alpha}\right)^{-1},\qquad
c_2=\left(1-\frac{R_1^*-h}{\alpha}\right)^{-1}.
\end{equation}
The constant $c_k$ is a lower bound on the additional dominant share obtained by $G_k$ per unit of non-dominant resource consumed. These bounds on within-group resource efficiency determine how strongly the intergroup speed may respond to a report. The following condition quantifies this restriction for a report-induced decrease $\delta$ in an augmented residual capacity.

\begin{lemma}[Sufficient condition for SP]
\label{lem:sufficient-SP}
Let $\Theta>0$ be a monotone speed rule, with $c_1,c_2$ as in~\eqref{eq:finite-c12} and $\bar c=\min\{c_1c_2,2\}$.
The corresponding \ASF mechanism satisfies SP if:
\begin{enumerate}
\item For every feasible decrease $\delta>0$ in $R_2^*$ caused by a $G_1$ report, with $\bar c\delta<R_2^*$,
\[
\frac{\Theta(\alpha,R_1^*,R_2^*-\delta)}
     {\Theta(\alpha,R_1^*,R_2^*)}
\le
\frac{1+c_1\delta/R_1^*}
     {1-\bar c\delta/R_2^*}.
\]
\item For every feasible decrease $\delta>0$ in $R_1^*$ caused by a $G_2$ report, with $\bar c\delta<R_1^*$,
\[
\frac{\Theta(\alpha,R_1^*,R_2^*)}
     {\Theta(\alpha,R_1^*-\delta,R_2^*)}
\le
\frac{1+c_2\delta/R_2^*}
     {1-\bar c\delta/R_1^*}.
\]
\end{enumerate}
\end{lemma}

The proof is in Appendix~\ref{sec:sp-condition}. We next apply the condition to design speed rules that improve the welfare guarantee.

\subsection{Two Improved \ASF Mechanisms}\label{sec:mechanisms}

We use the \ASF framework to improve the \fratio through two forms of adaptation. First, we adjust the balanced speed by a population-dependent factor and consider a speed rule of the following form:
\[
\Theta(\alpha,R_1^*,R_2^*)=f(\alpha)r,\qquad r=R_1^*/R_2^*.
\]
Here $r$ captures residual-capacity imbalance and is the speed used by \FTwos, while $f(\alpha)$ adjusts the response to population imbalance.

The boundary choice $f(\alpha)=0$ recovers \FOne, and $f(\alpha)=1$ recovers \FTwos. Their welfare guarantees favor \FOne when the minority population is small and \FTwos when the population is nearly balanced. The previous best-known mechanism \Hyb combines these rules: it uses \FOne when $\alpha\le2-\sqrt3$ and \FTwos otherwise. We improve this guarantee by selecting the population factor to balance the two cases in the welfare analysis. The speed $r$ satisfies Lemma~\ref{lem:sufficient-SP}, and the following corollary allows any positive population factor without violating SP.

\begin{corollary}[Population rescaling]
\label{cor:population-rescaling}
Suppose that a speed rule
$\Theta(\alpha,R_1^*,R_2^*)$ satisfies the conditions of
Lemma~\ref{lem:sufficient-SP}.  For any positive function
$f(\alpha)$, the speed rule
\[
\widetilde{\Theta}(\alpha,R_1^*,R_2^*)
=
f(\alpha)\Theta(\alpha,R_1^*,R_2^*)
\]
also satisfies those conditions.
\end{corollary}

\begin{proof}
An agent's truthful utility is at least $h$, while a report changing its dominant-resource group gives it at most $h$ of its true dominant resource, which is non-dominant under the report. Such a deviation cannot be profitable; hence profitable deviations preserve the group sizes and $\alpha$.
Therefore $f(\alpha)$ cancels from both multiplicative inequalities in Lemma~\ref{lem:sufficient-SP}.
Monotonicity in $R_1^*$ and $R_2^*$ is also preserved.
\end{proof}

For each fixed $\alpha$, the welfare analysis gives two competing bounds, corresponding to the exhaustion of resource~1 or resource~2. Choosing $f(\alpha)$ to equalize these bounds minimizes their maximum and yields
\begin{equation}\label{eq:asf1-speed}
\Theta_1(\alpha,R_1^*,R_2^*)=f_1(\alpha)r,\qquad
f_1(\alpha)=\frac{2\alpha-\alpha^2}{1-\alpha^2}.
\end{equation}
We call the resulting mechanism \AOne. It improves the asymptotic \fratio of \Hyb from $3-\sqrt3\approx1.268$ to $6/5=1.2$. The derivation of $f_1$ and the proof of the following theorem are in Appendix~\ref{app:asf-speed-rules}.

\begin{theorem}\label{thm:ASF1-ratio}
With two resources, \AOne{} satisfies SI, EF, PO, and SP. Moreover,
\[
\FR(\AOne)=\frac65+O\!\left(\frac1n\right).
\]
\end{theorem}

The tight family for \AOne in Appendix~\ref{app:asf-speed-rules} has a nearly balanced population and $r\to0$. In this family, a fair-optimal allocation concentrates the additional utility in $G_2$, while the linear speed $f_1(\alpha)r$ diverts scarce resource~1 to $G_1$. This motivates a stronger response to residual-capacity imbalance.

Replacing $r$ by $r^c$ with $c>1$ amplifies this response. A fixed exponent such as $c=2$, however, may violate the SP conditions in Lemma~\ref{lem:sufficient-SP}. We therefore choose the exponent adaptively using $c_1,c_2$:
\begin{equation}\label{eq:asf2-speed}
\Theta_2(\alpha,R_1^*,R_2^*)
=
\begin{cases}
r^{\,2-1/c_1}, & r\le 1,\\[2mm]
r^{\,2-1/c_2}, & r\ge 1,
\end{cases}
\end{equation}
We call the resulting mechanism \ATwo. Since $c_k\ge1$, the exponent $2-1/c_k$ lies between one and two. 
For fixed group sizes, a smaller $1/c_k$ means a larger augmented residual capacity $R_{3-k}^*$. The larger exponent then directs more growth toward $G_{3-k}$ on the corresponding branch of the speed rule.
The exponent also satisfies
\[
2-\frac1{c_k}\le\min\{c_k,2\}\le\bar c,
\]
which relates its magnitude to the bounds in Lemma~\ref{lem:sufficient-SP}. 
The following theorem gives the properties and \fratio of \ATwo. Its proof is deferred to Appendix~\ref{app:thm:ASF2-summary}.

\begin{theorem}
\label{thm:ASF2-summary}
With two resources, \ATwo satisfies SI, EF, PO, and SP. Moreover,
\[
\FR(\ATwo)=\frac{2}{2\sqrt2-1}+O\!\left(\frac1n\right),
\]
where $2/(2\sqrt2-1)\approx1.094$.
\end{theorem}

\subsection{The information limitation of \ASF}\label{sec:ASF-limit}

Although \ATwo improves the asymptotic \fratio to approximately $1.094$, no speed rule based only on $(\alpha,R_1^*,R_2^*)$ can make the \fratio approach one. The following lower bound holds even without imposing SP. It also shows that \ATwo is close to the best guarantee attainable within this class.

\begin{theorem}
\label{thm:asymptotic-asf-lower-bound}
Every \ASF mechanism whose Step~2 speed depends only on $(\alpha,R_1^*,R_2^*)$ has asymptotic \fratio at least $(191-2\sqrt{514})/135\approx1.0789$.
\end{theorem}

\begin{proof}[Proof sketch]
The lower bound follows from an information limitation of the state parameters.
We construct two instances with identical state parameters $(\alpha,R_1^*,R_2^*)$ but different efficient ways to allocate the remaining resources.
Let \(n\ge6\) be divisible by three and set $\varepsilon_n=n^{-2}$.
In both instances, \(G_1\) consists of \(2n/3\) agents with demand \((1,\varepsilon_n)\), and $G_2$ contains $n/3$ agents.
In the first instance, every \(G_2\) agent has demand \((1/3,1)\).
The second instance differs only in that one \(G_2\) agent has demand \((\varepsilon_n,1)\).
The definition of \(R_1^*\) excludes the smallest non-dominant demand in \(G_2\), so this replacement leaves \(R_1^*\) unchanged.
Since \(\alpha\) and \(R_2^*\) are also unchanged, the mechanism must choose the same speed $\theta$ on both instances.
The two instances admit different SI--EF comparison allocations.
As \(n\to\infty\), the first admits an SI--EF allocation in which essentially all Step~2 growth goes to \(G_2\), corresponding to speed \(0\).
The second admits an SI--EF allocation whose Step~2 increments in
\(G_1\) and \(G_2\) have asymptotic ratio \(1/3\).
Figure~\ref{fig:asf-lower-bound} illustrates these two comparison allocations.
Comparison with these allocations gives welfare-ratio lower bounds $(15\theta+5)/(11\theta+5)$ and $17/(15+6\theta)$, respectively, up to $O(n^{-2})$ terms. Since the mechanism uses the same speed on both instances, its worst-case ratio is at least the maximum of these bounds. The first increases with $\theta$ and the second decreases, so their maximum is minimized at their intersection, where the common value is $(191-2\sqrt{514})/135$. The full proof is in Appendix~\ref{app:asymptotic-asf-lower-bound}.
\end{proof}

The lower bound shows that further refinement of a speed rule based on $(\alpha,R_1^*,R_2^*)$ cannot eliminate the constant \fratio gap. The missing information is how individual demands determine the resource cost of additional utility along each group's canonical trajectory. This motivates the use of the full resource-load geometry to coordinate the groups in the next section.

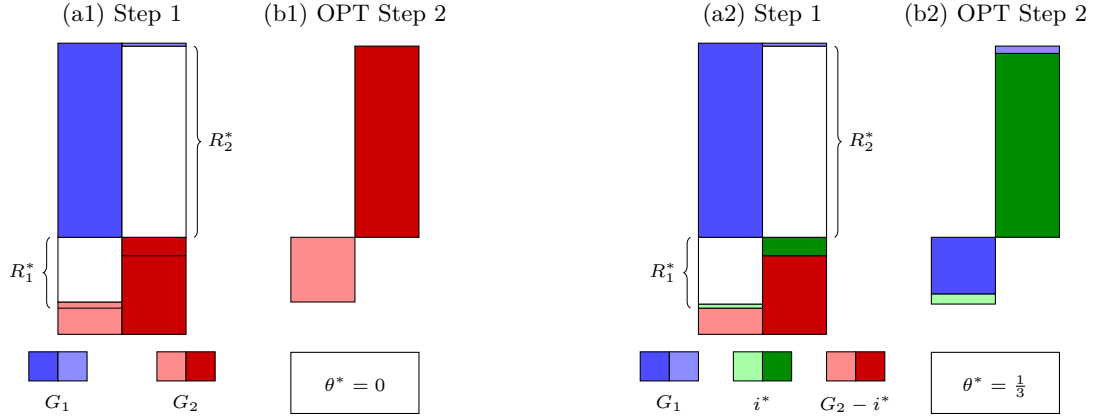
\begin{figure}[t]
    \centering
    \begingroup
    \def\deepblue{blue!70}
    \def\lightblue{blue!45}
    \def\deepred{red!80!black!100}
    \def\lightred{red!45}
    \def\deepgreen{green!55!black}
    \def\lightgreen{green!35}

    \begin{tikzpicture}[
        scale=3.85,
        every node/.style={font=\scriptsize},
       speedbox/.style={
            draw,
            line width=0.4pt,
            minimum width=1.694cm,
            minimum height=23pt,
            text height=12pt,
            text depth=6pt,
            inner xsep=3pt,
            inner ysep=2pt,
            anchor=north,
            font=\scriptsize
        }
    ]
        \def\w{0.22}

        \pgfmathsetmacro{\rOneTop}{1/3}
        \pgfmathsetmacro{\rTwoBot}{1/3}
        \def\rTwoTop{0.990}

        \pgfmathsetmacro{\firstOneBot}{1/9}

        \def\secondRedOneTop{0.090}
        \def\secondOneBot{0.104}
        \def\secondRedTwoTop{0.270}
        \def\secondGreenOneTop{0.139}
        \def\secondGreenTwoTop{0.965}

        \def\remaininglabels#1{%
            \draw[
                decorate,
                decoration={brace,raise=3pt},
                line width=0.35pt
            ]
                (0,#1) -- (0,\rOneTop);
            \node[anchor=east,font=\scriptsize]
                at (-0.045,{(#1+\rOneTop)/2})
                {$R_1^*$};

            \draw[
                decorate,
                decoration={brace,mirror,raise=3pt},
                line width=0.35pt
            ]
                (2*\w,\rTwoBot) -- (2*\w,\rTwoTop);
            \node[anchor=west,font=\scriptsize]
                at (2*\w+0.045,{(\rTwoBot+\rTwoTop)/2})
                {$R_2^*$};
        }

        \begin{scope}[shift={(0,0)}]
            \filldraw[fill=\lightred]
                (0,0)
                rectangle (\w,\firstOneBot);
            \filldraw[fill=white]
                (0,\firstOneBot)
                rectangle (\w,\rOneTop);
            \filldraw[fill=\deepblue]
                (0,\rOneTop)
                rectangle (\w,1);

            \filldraw[fill=\deepred]
                (\w,0)
                rectangle (2*\w,\rTwoBot);
            \filldraw[fill=white]
                (\w,\rTwoBot)
                rectangle (2*\w,\rTwoTop);
            \filldraw[fill=\lightblue]
                (\w,\rTwoTop)
                rectangle (2*\w,1);

            \draw
                (0,\secondRedOneTop) -- (\w,\secondRedOneTop);
            \draw
                (\w,\secondRedTwoTop) -- (2*\w,\secondRedTwoTop);
            
            \node[anchor=south,font=\footnotesize]
                at (\w,1.025) {(a1) Step 1};
                
            \remaininglabels{\secondRedOneTop}

            \filldraw[fill=\deepblue]
                (-0.10,-0.16) rectangle (0,-0.06);
            \filldraw[fill=\lightblue]
                (0,-0.16) rectangle (0.10,-0.06);
            \node[anchor=north,font=\scriptsize]
                at (0,-0.18) {$G_1$};

            \filldraw[fill=\lightred]
                ({2*\w-0.10},-0.16) rectangle (2*\w,-0.06);
            \filldraw[fill=\deepred]
                (2*\w,-0.16) rectangle ({2*\w+0.10},-0.06);
            \node[anchor=north,font=\scriptsize]
                at (2*\w,-0.18) {$G_2$};
        \end{scope}
        
        \begin{scope}[shift={(0.80,0)}]
            \filldraw[fill=\lightred]
                (0,\firstOneBot)
                rectangle (\w,\rOneTop);

            \filldraw[fill=\deepred]
                (\w,\rTwoBot)
                rectangle (2*\w,\rTwoTop);

            \node[anchor=south,font=\footnotesize]
                at (\w,1.025) {(b1) SI--EF Step 2};

            \node[speedbox]
                at (\w,-0.06) {$\theta^*=0$};
        \end{scope}

        \begin{scope}[shift={(2.20,0)}]
            \filldraw[fill=\lightred]
                (0,0)
                rectangle (\w,\secondRedOneTop);
            \filldraw[fill=\lightgreen]
                (0,\secondRedOneTop)
                rectangle (\w,\secondOneBot);
            \filldraw[fill=white]
                (0,\secondOneBot)
                rectangle (\w,\rOneTop);
            \filldraw[fill=\deepblue]
                (0,\rOneTop)
                rectangle (\w,1);

            \filldraw[fill=\deepred]
                (\w,0)
                rectangle (2*\w,\secondRedTwoTop);
            \filldraw[fill=\deepgreen]
                (\w,\secondRedTwoTop)
                rectangle (2*\w,\rTwoBot);
            \filldraw[fill=white]
                (\w,\rTwoBot)
                rectangle (2*\w,\rTwoTop);
            \filldraw[fill=\lightblue]
                (\w,\rTwoTop)
                rectangle (2*\w,1);

            \node[anchor=south,font=\footnotesize]
                at (\w,1.025) {(a2) Step 1};

            \remaininglabels{\secondRedOneTop}

            \filldraw[fill=\deepblue]
                (-0.20,-0.16) rectangle (-0.10,-0.06);
            \filldraw[fill=\lightblue]
                (-0.10,-0.16) rectangle (0,-0.06);
            \node[anchor=north,font=\scriptsize]
                at (-0.10,-0.18) {$G_1$};

            \filldraw[fill=\lightgreen]
                ({\w-0.10},-0.16) rectangle (\w,-0.06);
            \filldraw[fill=\deepgreen]
                (\w,-0.16) rectangle ({\w+0.10},-0.06);
            \node[anchor=north,font=\scriptsize]
                at (\w,-0.18) {$i^*$};

            \filldraw[fill=\lightred]
                (2*\w,-0.16) rectangle ({2*\w+0.10},-0.06);
            \filldraw[fill=\deepred]
                ({2*\w+0.10},-0.16) rectangle ({2*\w+0.20},-0.06);
            \node[anchor=north,font=\scriptsize]
                at ({2*\w+0.10},-0.18) {$G_2-i^*$};
        \end{scope}

        \begin{scope}[shift={(3.00,0)}]
            \filldraw[fill=\lightgreen]
                (0,\secondOneBot)
                rectangle (\w,\secondGreenOneTop);
            \filldraw[fill=\deepblue]
                (0,\secondGreenOneTop)
                rectangle (\w,\rOneTop);

            \filldraw[fill=\deepgreen]
                (\w,\rTwoBot)
                rectangle (2*\w,\secondGreenTwoTop);
            \filldraw[fill=\lightblue]
                (\w,\secondGreenTwoTop)
                rectangle (2*\w,\rTwoTop);

            \node[anchor=south,font=\footnotesize]
                at (\w,1.025) {(b2) SI--EF Step 2};

            \node[speedbox]
                at (\w,-0.06) {$\theta^*=\frac{1}{3}$};
        \end{scope}
    \end{tikzpicture}
    \endgroup

    \caption{Illustration of the two instances used in Theorem~\ref{thm:asymptotic-asf-lower-bound}, which have identical state parameters \((\alpha,R_1^*,R_2^*)\) after Step~1 but different SI--EF comparison allocations in Step~2.
    In the first instance, the comparison allocation assigns all Step~2 growth to \(G_2\), resulting in \(\theta^*=0\).
    In the second instance, the comparison allocation assigns most of the remaining resource~1 to \(G_1\) and most of the remaining resource~2 to \(i^*\), with the increment ratio tending to \(1/3\).
    }
    \label{fig:asf-lower-bound}
\end{figure}

\section{Incentive-Corrected Resource Balancing}\label{sec:crb}

We continue to focus on two resources and propose a mechanism that uses the full within-group demand profile to coordinate the two water levels.
The resulting Corrected Resource Balancing (\CRB) mechanism attains the optimal asymptotic order: its \fratio is at most $1+1/n$, matching the order of the universal lower bound $1+\Omega(1/n)$.
We first establish the lower bound, then derive and define \CRB, and prove its incentive and welfare guarantees.

\subsection{The universal lower bound}\label{sec:crb-lower-bound}

We first show that SP entails a positive welfare loss for every two-resource mechanism satisfying SI, EF, and PO.
A two-agent case already exposes the obstruction. Consider $\mathbf I_0=((1,1/2),(1/2,1))$. PO implies that some agent has utility at least $2/3$; after exchanging agent and resource labels if necessary, suppose this is the second agent. Now change the second agent's true demand to $(\varepsilon,1)$, with $0<\varepsilon<1/2$. Let $x,y$ be the truthful utilities. Reporting $(1/2,1)$ recreates $\mathbf I_0$ and gives it utility at least $2/3$, so SP requires $y\ge2/3$. Resource~2 feasibility gives $x/2+y\le1$, and hence $x+y\le2-y\le4/3$. 
In contrast, giving agent~2 its SI baseline and agent~1 all remaining resource~1 proportionally yields utilities $x^*=1-\varepsilon/2$ and $y^*=1/2$. This allocation is feasible, SI, and EF, and its welfare $3/2-\varepsilon/2$ tends to $3/2$ as $\varepsilon\to0$.
Thus the \fratio is at least $9/8$ in this two-agent illustration. For arbitrary populations, adding $n-2$ background agents yields the following general lower bound; its proof is deferred to Appendix~\ref{app:two-lower}.

\begin{theorem}[Deterministic lower bound]\label{two-lower}
For every $n\ge2$, every deterministic two-resource mechanism $f$ satisfying SI, EF, PO, and SP has
\[
\FR(f)\ge \frac{9n+7}{9n+6}
=1+\frac{1}{9n+6}.
\]
\end{theorem}

Theorem~\ref{two-lower} gives a uniform lower bound of $1+\Omega(1/n)$
for all population sizes. At $n=2$, it yields $25/24$, while the
two-agent argument above gives the stronger bound $9/8$.
We now seek a mechanism that approaches the fair optimum while
controlling strategic incentives.

\subsection{Manipulation and the one-agent correction}\label{sec:crb-manipulation}

We begin with a manipulation of the fair optimum, then derive a local correction.

\paragraph{Example~2 (manipulating the fair optimum).}
Consider two agents with $\mathbf d_1=(1,1/2)$ and $\mathbf d_2=(3/5,1)$. The fair optimum gives utilities $(4/7,5/7)$ and exhausts both resources. Reporting $\mathbf d'_1=(1,3/5)$ makes the fair-welfare optimizer give both agents utility $5/8$, increasing agent~1's true utility from $4/7$ to $5/8$. 

By Lemma~\ref{lem:canonical}, the fair optimum can be found among the canonical allocations $A^{\mathrm c}(p,q)$ of Section~\ref{sec:canonical}, with $p,q\in[0,h]$ and $h=1/n$. We call an agent \emph{active} when its utility exceeds $h$. The activation thresholds $p=a_ih$ and $q=b_jh$ partition the level square into cells on which each agent's allocation is affine. Fix a cell with active sets $P\subseteq G_1$ and $Q\subseteq G_2$, and write
\[
k=|P|,\qquad \ell=|Q|,\qquad
U=\sum_{i\in P}\frac1{a_i},\qquad
V=\sum_{j\in Q}\frac1{b_j}.
\]
Let $\mathcal C_{P,Q}\subseteq[0,h]^2$ denote the closure of this cell, specified by interval bounds on $p$ and $q$. The capacities remaining after the inactive agents' baseline allocations are
\[
\begin{aligned}
c_1&=1-h(|G_1|-k)-h\sum_{j\in G_2\setminus Q}b_j,\\
c_2&=1-h(|G_2|-\ell)-h\sum_{i\in G_1\setminus P}a_i.
\end{aligned}
\]
The fair-welfare problem on the cell is the two-variable linear program
\begin{equation*}
\begin{aligned}
\max_{p,q}\quad & Up+Vq+(n-k-\ell)h\\
\text{s.t.}\quad & Up+\ell q\le c_1,\\
& kp+Vq\le c_2,\\
& (p,q)\in\mathcal C_{P,Q}.
\end{aligned}
\tag{OPT-LP}\label{crb-opt-lp}
\end{equation*}
Maximizing over all cells yields $\OPT(\mathbf I)$. 

Suppose $k,\ell>0$ and both resource constraints bind at the optimum. 
Then for an active agent $i\in G_1$, its dominant-resource allocation is
\[
\frac{p}{a_i}=\frac{c_1-\ell c_2/V}{1+a_iD_{\rm opt}}, \qquad 
D_{\rm opt}=U-\frac{1}{a_i}-\frac{k\ell}{V}.\]
Since $U-1/a_i\ge k-1$ and $V\ge\ell$, we have $D_{\rm opt}\ge-1$.
Within this fixed cell and binding regime, a negative $D_{\rm opt}$
allows a profitable deviation. Indeed, for a fixed active-set cell, \(U-1/a_i\), and hence \(D_{\mathrm{opt}}\), is independent of agent \(i\)'s reported \(a_i\).
If $D_{\rm opt}<0$, increasing the reported demand $a_i$ increases both $p/a_i$ and $p$, provided the
optimum remains in the same active-set cell with both resource
constraints binding. The agent then receives strictly more of both
resources and therefore gains in true utility.
In Example~2, $U=2$, $V=5/3$, and $U-1/a_i=0$ for agent~1, so $D_{\rm opt}=-3/5$. The negative term therefore explains the profitable deviation.

The bound $D_{\rm opt}\ge-1$ suggests correcting one active agent's influence on the optimizer.
If we retain $Up+\ell q=c_1$ and replace the second exact-balance equation with
\[
kp+Vq=c_2\quad\longrightarrow\quad(k-1)p+(V+1)q=c_2,
\]
then an active agent $i\in G_1$ receives
\[
\frac{p}{a_i}
=\frac{c_1-\ell c_2/(V+1)}{1+a_iD_{\rm corr}},
\]
where the denominator coefficient is now nonnegative:
\[
D_{\rm corr}=U-\frac{1}{a_i}-\frac{\ell(k-1)}{V+1}\ge0,
\]
since $U-1/a_i\ge k-1$ and $V\ge\ell$. Thus inflating an active agent's non-dominant demand cannot improve its dominant allocation while the active sets and corrected binding equations remain fixed. 
This is only a local argument, because active sets may change. Section~\ref{sec:crb-opt} expresses the same correction directly through the global resource-load functions.

\subsection{Corrected Resource Balancing}\label{sec:crb-opt}

We now extend the local one-agent correction to the full canonical
allocation family. The resulting \CRB mechanism maximizes welfare
under corrected balance whenever feasible and uses an endpoint
allocation otherwise.

For a canonical allocation $A^{\mathrm c}(p,q)$, write
\[
\begin{aligned}
X(p)&=\sum_{i\in G_1}x_i(p), & A(p)&=\sum_{i\in G_1}a_i x_i(p),\\
Y(q)&=\sum_{j\in G_2}y_j(q), & B(q)&=\sum_{j\in G_2}b_j y_j(q).
\end{aligned}
\]
Thus $X,Y$ are the groups' total utilities, while $A,B$ are their
non-dominant resource consumptions. The physical resource loads and
social welfare are
\begin{equation}\label{two-loads-opt}
L_1(p,q)=X(p)+B(q),\qquad
L_2(p,q)=A(p)+Y(q),\qquad
W(p,q)=X(p)+Y(q).
\end{equation}
By Lemma~\ref{lem:canonical}, the fair optimum is obtained by
\begin{equation}\label{crb-opt-opt}
\begin{aligned}
\max_{p,q\in[0,h]}\quad & W(p,q)\\
\text{s.t.}\quad
&L_1(p,q)\le1,\\
&L_2(p,q)\le1.
\end{aligned}
\tag{OPT}
\end{equation}

Define the \emph{balance score}
\begin{equation}\label{crb-balance-opt}
\Phi(p,q):=L_2(p,q)-L_1(p,q)-p+q.
\end{equation}
Corrected balance is the condition $\Phi(p,q)=0$, or equivalently
$L_2(p,q)-p=L_1(p,q)-q$. When the corresponding group is active,
$p$ or $q$ is exactly one active agent's non-dominant-resource
contribution. 
On every active-set cell, \(\Phi(p,q)=0\) is exactly the one-agent corrected binding equation derived in Section~\ref{sec:crb-manipulation}. Thus the following program extends that local correction across activation changes:
\begin{equation}\label{crb-program-opt}
\begin{aligned}
\max_{p,q\in[0,h]}\quad & W(p,q)\\
\text{s.t.}\quad
&L_1(p,q)\le1,\\
&L_2(p,q)\le1,\\
&\Phi(p,q)=0.
\end{aligned}
\tag{CRB}
\end{equation}

We first characterize when corrected balance is feasible. Let $\bar p$
be the largest feasible level with $q=0$, and define $\bar q$
symmetrically with $p=0$.
Call \((p,q)\) physically feasible if \(p,q\in[0,h]\) and \(L_1(p,q),L_2(p,q)\le1\).

\begin{lemma}
\label{lem:crb-balance-feasibility}
The score $\Phi$ is continuous, strictly decreasing in $p$, and strictly
increasing in $q$. 
Consequently, for every physically feasible pair $(p,q)$,
\begin{equation}\label{crb-score-range}
\Phi(\bar p,0)
\le
\Phi(p,q)
\le
\Phi(0,\bar q),
\end{equation}
and
\[
\eqref{crb-program-opt}\text{ is feasible}
\quad\Longleftrightarrow\quad
\Phi(\bar p,0)\le0\le\Phi(0,\bar q).
\]
\end{lemma}

\begin{proof}
Write $\Phi(p,q)=G(q)-F(p)$, where
\[
F(p):=p+\sum_{i\in G_1}(1-a_i)x_i(p),
\qquad
G(q):=q+\sum_{j\in G_2}(1-b_j)y_j(q).
\]
Since the functions $x_i$ and $y_j$ are continuous and nondecreasing,
while the explicit terms $p$ and $q$ are strictly increasing, both
$F$ and $G$ are continuous and strictly increasing. Hence $\Phi$ is
continuous, strictly decreasing in $p$, and strictly increasing in $q$.

Now let $(p,q)$ be any physically feasible pair. Lowering either water
level preserves feasibility, so $(p,0)$ and $(0,q)$ are feasible.
By the definitions of $\bar p$ and $\bar q$, we have $p\le\bar p$ and $q\le\bar q$, so monotonicity gives 
\[\Phi(\bar p,0)\le \Phi(p,0) \le \Phi(p,q)\le  \Phi(0,q) \le \Phi(0,\bar q).\]
Thus $\Phi(\bar p,0)$ and $\Phi(0,\bar q)$ are respectively the minimum
and maximum score over all physically feasible pairs. In particular,
\eqref{crb-program-opt} can be feasible only if $\Phi(\bar p,0)\le0\le\Phi(0,\bar q)$.

Conversely, suppose these inequalities hold. If $\Phi(0,0)=0$, the SI
pair is already balanced. If $\Phi(0,0)>0$, then
\[
\Phi(\bar p,0)\le0<\Phi(0,0),
\]
and continuity of $p\mapsto\Phi(p,0)$ gives some
$p_0\in[0,\bar p]$ with $\Phi(p_0,0)=0$. The pair $(p_0,0)$ is feasible.
If $\Phi(0,0)<0$, symmetrically,
\[
\Phi(0,0)<0\le\Phi(0,\bar q),
\]
so some $q_0\in[0,\bar q]$ satisfies $\Phi(0,q_0)=0$.
In all cases \eqref{crb-program-opt} is feasible.
\end{proof}

We now define the \CRB mechanism.

\begin{definition}[Corrected Resource Balancing]\label{crb-definition-opt}
If the SI allocation exhausts a resource, \CRB returns it. Otherwise:
\begin{enumerate}
\item If \eqref{crb-program-opt} is feasible, select a welfare-maximizing
      solution $(p^{\mathrm c},q^{\mathrm c})$.
\item Otherwise, if $\Phi(\bar p,0)>0$, select
      $(p^{\mathrm c},q^{\mathrm c})=(\bar p,0)$;
      if $\Phi(0,\bar q)<0$, select
      $(p^{\mathrm c},q^{\mathrm c})=(0,\bar q)$.
\end{enumerate}
Return the canonical allocation $A^{\mathrm c}(p^{\mathrm c},q^{\mathrm c})$.
\end{definition}

By Lemma~\ref{lem:crb-balance-feasibility}, the cases are exhaustive:
if corrected balance is infeasible, exactly one of
$\Phi(\bar p,0)>0$ and $\Phi(0,\bar q)<0$ holds.
It remains only to verify that the welfare-maximizing balanced solution
induces a unique allocation and exhausts a resource.

\begin{lemma}[Well-definedness]\label{crb-feasible-opt}
\CRB selects a unique physical allocation and exhausts at least one
resource.
\end{lemma}

\begin{proof}
The SI case is immediate. Otherwise the SI allocation leaves both
resources slack.
For the infeasible case, \CRB selects one of the two fixed endpoints. At $(\bar p,0)$, if
$\bar p<h$ and both resources were slack, $p$ could be increased,
contradicting the definition of $\bar p$; if $\bar p=h$, then
$L_2(h,0)=1$. Thus a resource is exhausted. The argument for
$(0,\bar q)$ is symmetric.

For the feasible case, the balanced feasible set is nonempty and compact. By
Lemma~\ref{lem:crb-balance-feasibility}, any two balanced pairs are
ordered coordinatewise:
\[
p<p' \quad\Longleftrightarrow\quad q<q'.
\]
Hence there is a greatest balanced feasible pair $(p^{\max},q^{\max})$.
Since agents' utilities are nondecreasing in their group water levels,
this pair weakly welfare-dominates every other balanced feasible pair,
with strict inequality whenever the induced physical allocation differs.
Thus every welfare maximizer induces the same physical allocation.

This allocation exhausts a resource. Otherwise, if \(p,q<h\), continuity and strict monotonicity of \(F\) and \(G\) would allow a slightly larger balanced pair that remains feasible, contradicting maximality. If $p=h$, then
\[
L_2(h,q)\ge |G_1|h+|G_2|h=1,
\]
so feasibility implies $L_2(h,q)=1$; the case $q=h$ is symmetric.
\end{proof}

The following sign conditions will be used in the incentive and welfare
analyses in Section~\ref{sec:two-properties-opt}.

\begin{lemma}\label{lem:crb-balance-signs}
For every nonbaseline output of \CRB, the selected water levels satisfy
\begin{equation}\label{crb-output-sign-opt}
q^{\mathrm c}>0\ \Longrightarrow\ \Phi(p^{\mathrm c},q^{\mathrm c})\le0,
\qquad
p^{\mathrm c}>0\ \Longrightarrow\ \Phi(p^{\mathrm c},q^{\mathrm c})\ge0.
\end{equation}
\end{lemma}

\begin{proof}
In the balanced case, $\Phi(p^{\mathrm c},q^{\mathrm c})=0$.
In the two infeasible cases, Definition~\ref{crb-definition-opt}
selects respectively $(\bar p,0)$ with $\Phi(\bar p,0)>0$, or
$(0,\bar q)$ with $\Phi(0,\bar q)<0$. The two implications follow.
\end{proof}

\paragraph{Example~2 continued.}
For $\mathbf d_1=(1,1/2)$ and $\mathbf d_2=(3/5,1)$, the endpoint
levels are $\bar p=7/20$ and $\bar q=9/20$. Panels~(a) and~(d) of
Figure~\ref{fig:crb-water-geometry-opt} show that $\Phi(0,\bar q)=\frac12>0>-\frac12=\Phi(\bar p,0)$.
Thus corrected balance is feasible by
Lemma~\ref{lem:crb-balance-feasibility}, and \CRB uses the balanced
optimization case.
When both agents are active, $\Phi=y_2-x_1$, so corrected balance
requires $x_1=y_2=t$. Resource feasibility gives $t\le5/8$, yielding
$(p^{\mathrm c},q^{\mathrm c})=(5/16,3/8)$ and utilities $(5/8,5/8)$.
Panels~(b) and~(c) compare OPT, with $\Phi=1/7$, and \CRB, with
$\Phi=0$: the correction reduces welfare from $9/7$ to $5/4$, a loss
of $1/28$.
If agent~1 reports $(1,3/5)$, \CRB again gives both agents reported
utility $5/8$. Agent~1's bundle changes from $(5/8,5/16)$ to
$(5/8,3/8)$, but its true utility remains $5/8$. Hence the deviation
that benefits agent~1 under OPT gives no gain under \CRB.

\begin{figure}[!htbp]
    \centering
    \begingroup
    \def\deepred{red!80!black!100}
    \def\lightred{red!45}
    \def\deepblue{blue!70}
    \def\lightblue{blue!45}
    \begin{tikzpicture}[scale=3.80, every node/.style={font=\scriptsize}]
        \def\w{0.22}
        \newcommand{\geometrypanel}[8]{%
            \begin{scope}[shift={(#1,0)}]
                \pgfmathsetmacro{\levelp}{#3}
                \pgfmathsetmacro{\levelq}{#4}
                \pgfmathsetmacro{\utilityOne}{#5}
                \pgfmathsetmacro{\utilityTwo}{#6}
                \pgfmathsetmacro{\redOneTop}{3*\utilityTwo/5}
                \pgfmathsetmacro{\blueOneBot}{1-\utilityOne}
                \pgfmathsetmacro{\blueTwoBot}{1-\utilityOne/2}

                \filldraw[fill=white] (0,0) rectangle (\w,1);
                \filldraw[fill=\lightred] (0,0) rectangle (\w,\redOneTop);
                \filldraw[fill=\deepblue] (0,\blueOneBot) rectangle (\w,1);

                \filldraw[fill=white] (\w,0) rectangle (2*\w,1);
                \filldraw[fill=\deepred] (\w,0) rectangle (2*\w,\utilityTwo);
                \filldraw[fill=\lightblue] (\w,\blueTwoBot) rectangle (2*\w,1);

                \node[anchor=south,font=\footnotesize] at (\w,1.045) {#2};
                \node[anchor=north] at (\w/2,-0.035) {$1$};
                \node[anchor=north] at (3*\w/2,-0.035) {$2$};
                \node[anchor=north] at (\w,-0.145) {$(p,q)=#7$};
                \node[anchor=north] at (\w,-0.325) {$#8$};

                \ifdim\levelq pt>0pt
                    \draw[decorate,decoration={brace,raise=3pt},line width=0.35pt]
                        (0,0) -- (0,\redOneTop);
                    \node[anchor=east] at (-0.065,\redOneTop/2) {$q$};
                \fi
                \ifdim\levelp pt>0pt
                    \draw[decorate,decoration={brace,mirror,raise=3pt},line width=0.35pt]
                        (2*\w,\blueTwoBot) -- (2*\w,1);
                    \node[anchor=west] at (2*\w+0.065,{(\blueTwoBot+1)/2}) {$p$};
                \fi
            \end{scope}
        }
        \geometrypanel{0}{(a) $G_2$ endpoint}{0}{9/20}{1/2}{3/4}
            {\bigl(0,\tfrac{9}{20}\bigr)}
            {\Phi=\tfrac12}
        \geometrypanel{1.05}{(b) OPT}{2/7}{3/7}{4/7}{5/7}
            {\bigl(\tfrac27,\tfrac37\bigr)}
            {\Phi=\tfrac17}
        \geometrypanel{2.10}{(c) CRB}{5/16}{3/8}{5/8}{5/8}
            {\bigl(\tfrac{5}{16},\tfrac38\bigr)}
            {\Phi=0}
        \geometrypanel{3.15}{(d) $G_1$ endpoint}{7/20}{0}{7/10}{1/2}
            {\bigl(\tfrac{7}{20},0\bigr)}
            {\Phi=-\tfrac12}

        \newcommand{\geometrykey}[4]{%
            \filldraw[fill=#2] (#1,-0.62) rectangle +(0.10,0.08);
            \filldraw[fill=#3] ({#1+0.10},-0.62) rectangle +(0.10,0.08);
            \node[anchor=west] at ({#1+0.24},-0.58) {Agent $#4$};
        }
        \geometrykey{0.75}{\deepblue}{\lightblue}{1}
        \geometrykey{2.25}{\lightred}{\deepred}{2}
    \end{tikzpicture}
    \endgroup
    \caption{OPT, \CRB, and the endpoint allocations in Example~2.
    Panels~(a) and~(d) show $(0,\bar q)$ and $(\bar p,0)$; their opposite
    scores establish feasibility of corrected balance. Panels~(b) and~(c)
    compare OPT with \CRB.}
    \label{fig:crb-water-geometry-opt}
\end{figure}

\paragraph{Implementation.}
The functions in \eqref{two-loads-opt} are continuous and piecewise affine,
with breakpoints at the activation thresholds $a_i h$ and $b_j h$.
After sorting these thresholds, single-group water filling computes
the endpoints, and the endpoint test in Lemma~\ref{lem:crb-balance-feasibility}
detects infeasibility. Otherwise the two water levels increase together
along the balance curve, so each activation threshold is crossed at
most once. Scanning the resulting $O(n)$ active-set cells solves
\eqref{crb-program-opt} in $O(n\log n)$ total time; the details are
deferred to Appendix~\ref{app:crb-implementation}.

\begin{lemma}[Running time]\label{two-runtime-opt}
\CRB can be computed in $O(n\log n)$ time.
\end{lemma}

\FloatBarrier

\subsection{Analysis of CRB}\label{sec:two-properties-opt}

We show that the one-agent correction prevents profitable deviations and limits the welfare loss to one SI share \(h=1/n\).

\begin{theorem}[Incentive guarantees]\label{two-SP-opt}
The mechanism \CRB satisfies SI, EF, PO, and SP.
\end{theorem}

\begin{proof}
Every output is a feasible canonical allocation with water levels in
$[0,h]$, so Lemma~\ref{lem:canonical-fairness} gives SI and EF.
Lemma~\ref{crb-feasible-opt} gives resource exhaustion. Since the output
is non-wasteful and every demand coordinate is positive, a Pareto
improvement would require more of the exhausted resource. Thus \CRB
also satisfies PO.

To prove SP, fix a true profile $\mathbf I$, an agent $i$, and a report
that produces profile $\mathbf I'$. Write $\mathbf A$ and $\mathbf A'$
for the two outputs, and suppose that the report is profitable.
A profitable deviation must preserve the dominant-resource group and both outcomes must be nonbaseline: switching groups or ending at the reported SI baseline gives true utility at most \(h\), while a profitable deviation from a truthful SI baseline would Pareto-improve that baseline.

By symmetry, consider a group-preserving report by $i\in G_1$.
Let $(p,q)$ and $(p',q')$ be the selected water levels, and use
subscripts on $\Phi$ to identify the profile at which its loads are
evaluated. Because $\mathbf A_i$ is proportional to the true demand,
a strict utility gain requires $A'_{i1}>A_{i1}$ and $A'_{i2}>A_{i2}$.
In particular, $A'_{i1}>A_{i1}\ge h$, so the deviator is active under
its report and $A'_{i2}=p'$. Hence
\[
p'=A'_{i2}>A_{i2}\ge p.
\]
Every unchanged $G_1$ agent therefore weakly gains. If $q'\ge q$,
every $G_2$ agent would also weakly gain, contradicting PO of
$\mathbf A$. Thus $q'<q$, which gives $q>0$ and $p'>0$.
Applying Lemma~\ref{lem:crb-balance-signs} to the two nonbaseline
outputs yields
\begin{equation}\label{crb-selection-sign-opt}
\Phi_{\mathbf I}(p,q)\le0
\le\Phi_{\mathbf I'}(p',q').
\end{equation}

We now compare the scores term by term. For each unchanged
$k\in G_1$, the contribution to $L_2-L_1$ is $(a_k-1)x_k$;
it weakly decreases because $a_k\le1$ and $p'>p$.
For each $j\in G_2$, the contribution is $(1-b_j)y_j$;
it weakly decreases because $b_j\le1$ and $q'<q$.
For the deviator, combine its load contribution with the correction
$-p$. Using $A'_{i2}=p'$ gives
\[
\begin{aligned}
&\bigl(A'_{i2}-A'_{i1}-p'\bigr)
 -\bigl(A_{i2}-A_{i1}-p\bigr)\\
&\qquad=p-A_{i2}-(A'_{i1}-A_{i1})<0,
\end{aligned}
\]
since $A_{i2}\ge p$ and
$A'_{i1}>A_{i1}$. Finally, the correction $+q$ decreases by
$q'-q<0$. Summing these changes gives
\[
\Phi_{\mathbf I'}(p',q')<\Phi_{\mathbf I}(p,q),
\]
contradicting~\eqref{crb-selection-sign-opt}. The argument for a
group-preserving report in $G_2$ is symmetric, completing the proof.
\end{proof}

The same balance-sign invariant also controls welfare loss. In the SP
proof, it rules out the movement of $\Phi$ required by a profitable
deviation; in the welfare analysis, it bounds the slack of the
unexhausted resource by one water level. Suppose, for example, that
resource~1 is exhausted. We show that any better fair allocation must
shift utility from $G_1$ to $G_2$. The resource~2 released by $G_1$ is
at most its utility loss, so the net welfare gain is bounded by the
unused resource~2 capacity, which Lemma~\ref{lem:crb-balance-signs}
bounds by $h=1/n$.

\begin{theorem}[Welfare guarantee]\label{two-welfare-opt}
For every two-resource instance,
\[
\OPT(\mathbf I)-\SW(\CRB(\mathbf I))
\le \frac1n,
\qquad
\FR(\CRB)\le1+\frac1n.
\]
\end{theorem}

\begin{proof}
If the \CRB output is the SI allocation, then every agent has utility
$h$. Any SI allocation with strictly larger welfare would weakly improve
every agent and strictly improve some agent, contradicting PO of the
\CRB output. Hence the SI output is fair-optimal. We may therefore
assume that the \CRB output is nonbaseline and has welfare strictly
below the fair optimum.

Let $(p,q)$ be the water levels selected by \CRB, and let
$(p^*,q^*)$ induce a canonical fair optimum, whose existence follows
from Lemma~\ref{lem:canonical}. Thus $W(p^*,q^*)=\OPT(\mathbf I)>W(p,q)$.
By Lemma~\ref{crb-feasible-opt}, \CRB exhausts a resource. By symmetry,
suppose $L_1(p,q)=1$.

We first show that $p^*<p$. Suppose instead that $p^*\ge p$.
If $q^*\ge q$, every agent weakly gains and the strict welfare
improvement contradicts PO. Hence $q^*<q$. Since $L_1(p,q)=1$,
resource~1 feasibility and $b_j\le1$ give
\[
X(p^*)-X(p)
\le B(q)-B(q^*)
\le Y(q)-Y(q^*),
\]
contradicting $W(p^*,q^*)>W(p,q)$. Thus $p^*<p$.

Similarly, since $0<a_i\le1$, decreasing the $G_1$ water level releases
no more resource~2 than the corresponding utility loss:
\[
A(p)-A(p^*)\le X(p)-X(p^*).
\]
Consequently,
\[
\begin{aligned}
\OPT(\mathbf I)-W(p,q)
&=X(p^*)-X(p)+Y(q^*)-Y(q)\\
&\le A(p^*)-A(p)+Y(q^*)-Y(q)\\
&=L_2(p^*,q^*)-L_2(p,q)\\
&\le 1-L_2(p,q).
\end{aligned}
\]

It remains to bound the unused resource~2 capacity. Since
$p>p^*\ge0$, Lemma~\ref{lem:crb-balance-signs} gives $\Phi(p,q)\ge0$.
Using $L_1(p,q)=1$ in the definition of $\Phi$,
\[
1-L_2(p,q)
=q-p-\Phi(p,q)
\le q-p
\le h.
\]
Hence $\OPT(\mathbf I)-W(p,q)\le h=\frac1n$.
The case $L_2(p,q)=1$ is symmetric.

Finally, SI gives $W(p,q)\ge nh=1$. Therefore
\[
\frac{\OPT(\mathbf I)}{W(p,q)}
=
1+\frac{\OPT(\mathbf I)-W(p,q)}{W(p,q)}
\le
1+\frac{1/n}{W(p,q)}
\le
1+\frac1n.
\]
Taking the worst case over instances proves the fair-ratio bound.
\end{proof}

Combining Theorem~\ref{two-welfare-opt} with the lower bound from
Theorem~\ref{two-lower} determines the optimal asymptotic order.

\begin{corollary}[Optimal-order two-resource price]\label{two-order-opt}
For every $n\ge2$, the optimal fair-ratio over deterministic
two-resource mechanisms satisfying SI, EF, PO, and SP lies between
\[
1+\frac{1}{9n+6}
\qquad\text{and}\qquad
1+\frac1n.
\]
In particular, it is $1+\Theta(1/n)$.
\end{corollary}

\section{Experimental Evaluation}\label{sec:experiments}

We complement our worst-case guarantees with experiments on synthetic instances and instances generated from Google cluster-usage traces~\cite{reiss2011google}, following the experimental setting of \cite{BeiLL22}.
We compare \(\ATwo\) and \(\CRB\) with \(\DR\), \(\FOne\), \(\FTwos\), and \(\Hyb\).
For each parameter setting, we generate \(1{,}000\) instances.
For each instance, we compute the instance-wise ratio between the optimal SI--EF social welfare and the welfare achieved by the mechanism, and report the arithmetic mean over the \(1{,}000\) ratios. Figure~\ref{fig:exp-sw} summarizes the results; lower values are better, and the ideal value is \(1\).
Overall, \(\CRB\) attains the lowest observed mean ratio at every tested setting in both experiments.

\begin{figure}[t]
    \centering
    \safeincludegraphics[width=\textwidth]{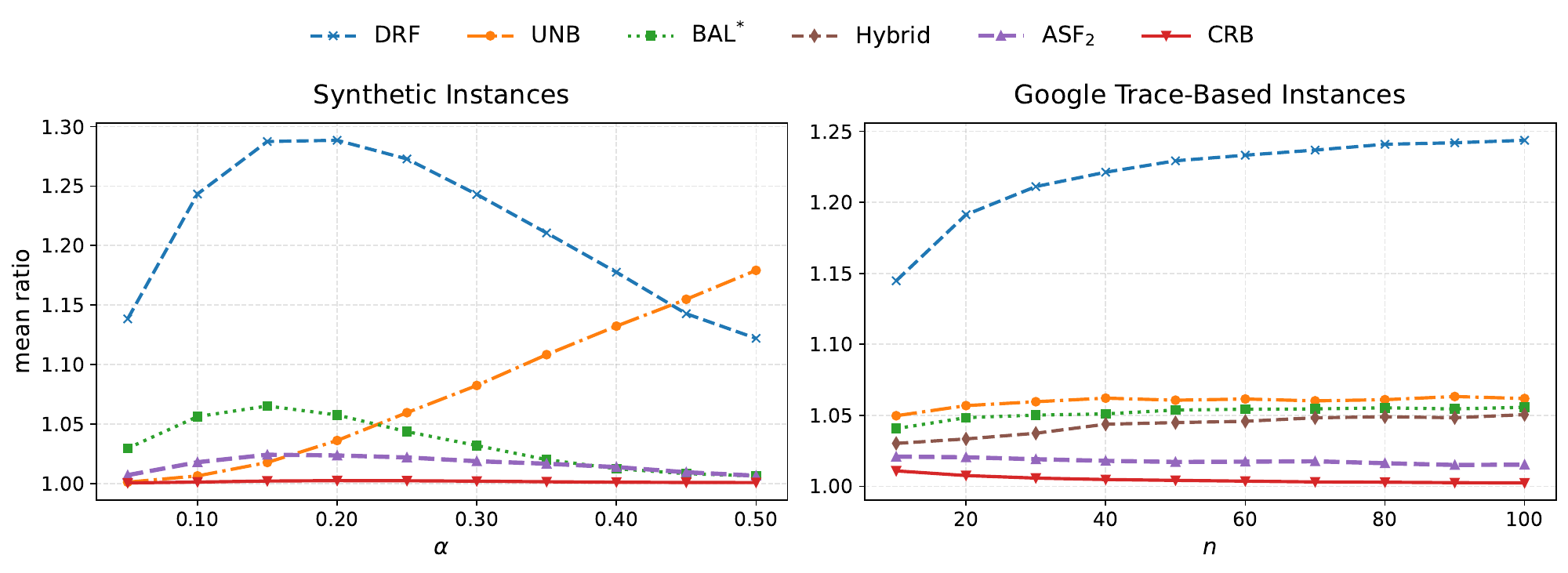}
    \caption{Mean welfare ratios on synthetic instances and Google cluster traces. Each point is the arithmetic mean, over \(1{,}000\) instances, of the optimal SI--EF social welfare divided by the mechanism's social welfare. Lower is better; the ideal value is \(1\).}
    \label{fig:exp-sw}
\end{figure}

\emph{Random instances with different $\alpha$.}
We first test the mechanisms on random instances with fixed $n=100$ and minority-group proportion $\alpha\in\{0.05,0.10,\ldots,0.50\}$. For both groups, the non-dominant entries are sampled uniformly from \(\{0.01,0.02,\ldots,0.99\}\).
The results are shown in the left panel of Figure~\ref{fig:exp-sw}. Among the five mechanisms other than \(\CRB\), \(\FOne\) (and hence \(\Hyb\)) attains the lowest observed mean ratio for \(\alpha\in\{0.05,0.10,0.15\}\).
For intermediate imbalance, \(\alpha\in\{0.20,0.25,0.30,0.35\}\), \(\ATwo\) performs best among these mechanisms.
Near the balanced regime, \(\alpha\in\{0.40,0.45,0.50\}\), \(\FTwos\) (and hence \(\Hyb\)) slightly outperforms \(\ATwo\).
In contrast, \(\CRB\) is much less sensitive to the population regime. It attains the lowest observed mean ratio at every tested value of \(\alpha\), and its mean ratio remains below \(1.003\) throughout. Thus, on these synthetic instances, \(\CRB\) stays within \(0.3\%\) of the fair-optimal welfare across the entire range from strongly imbalanced to nearly balanced populations.

\emph{Instances generated from Google traces.}
We next generate trace-based instances from Google cluster-usage traces~\cite{reiss2011google}, which contain CPU and memory demands of submitted tasks. We normalize these demands into two-resource demand vectors and randomly sample \(n\) vectors to form each instance, for \(n\in\{10,20,\ldots,100\}\).
The results are shown in the right panel of Figure~\ref{fig:exp-sw}. 
Both new mechanisms achieve lower observed mean ratios than all four prior mechanisms at every tested population size.
The performance of \(\ATwo\) on these instances appears to reflect the stronger asymmetry in the Google traces. In the normalized pool, about \(66.6\%\) of tasks are CPU-dominant and about \(33.4\%\) are memory-dominant, while the two groups also differ in their average non-dominant demands (\(0.417\) versus \(0.537\)).
Consequently, sampled instances typically have \(\alpha\approx 1/3\) together with substantially imbalanced augmented remaining resources, with median \(R_1^*/R_2^*\) around \(0.4\). This regime is particularly favorable to \(\ATwo\), which adapts jointly to population imbalance and remaining-resource imbalance.
\(\CRB\) further improves on \(\ATwo\), attaining the lowest observed mean ratio at every tested population size. Its mean ratio decreases from approximately \(1.0109\) at \(n=10\) to \(1.0024\) at \(n=100\). This improvement with \(n\) is consistent with the \(1+O(1/n)\) worst-case guarantee of \(\CRB\), while the observed ratios are considerably smaller than the corresponding worst-case upper bounds.

\section{Beyond Two Resources}
\label{sec:beyond-two}

We now turn to the general case with \(m\ge 3\) resources. The positive results above show that, with two resources, strategy-proof mechanisms can achieve fair-ratio close to \(1\). We show that this phenomenon is specific to the two-resource setting.

\subsection{A Unified Lower Bound for \(m\ge3\)}
\label{sec:m-barrier}

\citet{BeiLL26} proved that, when \(m\ge 4\), every mechanism satisfying SI, EF, PO, and SP has fair-ratio \(m\). 
For \(m=3\), however, their lower bound was only \(2\), leaving open whether the same \(m\)-barrier holds. 
We give a unified lower-bound construction for every \(m\ge3\), proving the stronger statement that every deterministic mechanism satisfying only SI and SP has fair-ratio \(m\).

Our construction combines the prefix-compensation idea from the previous \(m=3\) lower-bound construction with the auxiliary-coordinate idea used for $m \ge 4$.
The key additional ingredient is a harmonic-weight counting argument, which replaces the ordinary counting argument and allows the construction to work uniformly for all \(m\ge3\).
The \(m=3\) case was independently resolved in concurrent work by
\citet{fithiansolving}, using essentially the same harmonic-weight argument; see Section~\ref{sec:related-work} for discussion.
The proof of the following theorem is deferred to Appendix~\ref{app:thm:PoSP}.

\begin{theorem}
\label{thm:PoSP}
With $m\ge 3$ resources, any deterministic mechanism $M$ satisfying SI and SP has
$\AR(M)=m$.
\end{theorem}

Theorem~\ref{thm:PoSP} gives a sharp contrast between two resources and three or more resources. 
For \(m=2\), the mechanisms developed above achieve fair-ratio close to \(1\). For \(m\ge 3\), no deterministic mechanism satisfying SI, EF, PO, and SP can improve on the trivial \(m\) guarantee.

\subsection{Randomized Mechanisms}
\label{sec:randomized}

So far we have considered deterministic mechanisms. A natural next question is whether randomization can circumvent the above impossibility.
Randomization is known to improve approximation guarantees in related
problems, including facility location
\citep{ProcacciaT13,LuSunWangZhu10,Barak26}.
We show that this is not the case here. 

For a randomized mechanism \(M\), we define its fair-ratio using expected welfare:
\[
\AR(M)
  =
    \sup_{\mathbf I\in\mathcal I}
    \frac{
        \max_{\mathbf A:\ \text{SI and EF}} \SW(\mathbf A)
    }{
        \mathbb E[\SW(M(\mathbf I))]
    },
\]
where the expectation is over the internal randomness of \(M\). We say that a randomized
mechanism $M$ satisfies SI, EF, or PO \emph{ex post} if the realized allocation satisfies
the corresponding property.

We consider two standard incentive notions for randomized mechanisms. A mechanism is
\emph{universally strategy-proof} if, for every realization \(\omega\) of its internal randomness, the induced deterministic mechanism \(M^\omega\) is strategy-proof:
\[
    u_i\bigl(M^\omega_i(\mathbf d_i,\mathbf d_{-i})\bigr)
    \ge
    u_i\bigl(M^\omega_i(\mathbf d'_i,\mathbf d_{-i})\bigr),
    \quad
    \forall i,\mathbf d_i,\mathbf d'_i,\mathbf d_{-i},\omega .
\]
A weaker requirement is \emph{truthful in expectation}, under which truthful reporting maximizes
expected utility:
\[
    \mathbb E\bigl[
        u_i\bigl(M_i(\mathbf d_i,\mathbf d_{-i})\bigr)
    \bigr]
    \ge
    \mathbb E\bigl[
        u_i\bigl(M_i(\mathbf d'_i,\mathbf d_{-i})\bigr)
    \bigr],
    \quad
    \forall i,\mathbf d_i,\mathbf d'_i,\mathbf d_{-i}.
\]
Universal strategy-proofness implies truthfulness in expectation.

Our deterministic lower-bound construction can be strengthened to this randomized setting: the same $m$-barrier holds even under the weaker requirement of truthfulness in expectation.
The proof is an expected-value analogue of Theorem~\ref{thm:PoSP}.
The only changes are that SP is replaced by truthfulness in expectation and the resource-budget argument is taken in expectation.
The proof is deferred to Appendix~\ref{app:randomized-lower-bound}.

\begin{theorem}
\label{thm:randomized-PoSP}
For every \(m\ge 3\), every randomized mechanism satisfying ex-post SI and truthfulness in expectation has fair-ratio exactly \(m\). Consequently,
the same bound holds for universally strategy-proof randomized
mechanisms.
\end{theorem}

\section{Conclusion}
\label{sec:con}

We studied the welfare cost of strategyproofness in fair multi-resource allocation with Leontief utilities. 
For two resources, the \ASF framework unifies earlier water-filling mechanisms and achieves an asymptotic \fratio of approximately $1.094$, but its reliance on aggregate state information prevents convergence to the fair optimum. \CRB overcomes this limitation by using the full demand profile and adding a one-agent balance constraint to fair-welfare maximization. It preserves SI, EF, PO, and SP and achieves \fratio at most $1+1/n$; together with the universal lower bound, this establishes the optimal asymptotic order $1+\Theta(1/n)$. 
With three or more resources, the optimal \fratio is exactly $m$, even under the weaker requirements of ex-post SI and truthfulness in expectation. These results establish a sharp boundary: the welfare loss from strategyproofness can vanish with the population size for two resources, whereas three or more resources entail the full worst-case loss.
An important direction for future work is to extend the ideas developed here to richer allocation environments, including indivisible tasks, multiple machines, and dynamic settings.
More broadly, we believe that the principle underlying the \fratio benchmark, evaluating efficiency relative to the best allocation satisfying the required fairness constraints, may be useful in other mechanism-design problems where fairness and incentive compatibility must be achieved simultaneously.

\section*{AI Assistance Disclosure}
The \ASF framework, the sufficient condition for strategy-proofness, and \AOne were developed by the authors. 
Beyond these components, OpenAI's
ChatGPT was used extensively as a research assistant during the development of this work.
All AI-assisted arguments and results were critically examined, independently verified, and rewritten by the authors for correctness and clarity.
The authors take full responsibility for all statements, proofs, and conclusions in the paper.

\subsubsection*{\ackname}
This research was supported by the National Natural Science Foundation of China
(No.~12301412) and the CCF--Huawei Populus Grove Fund
(No.~CCF-HuaweiLK2024007).

\bibliographystyle{splncs04}
\bibliography{DRF}

@article{Bonald2014,
  title={Enhanced cluster computing performance through proportional fairness},
  author={Bonald, Thomas and Roberts, James},
  journal={Performance Evaluation},
  volume={79},
  pages={134--145},
  year={2014},
  publisher={Elsevier}
}

@inproceedings{Bonald2015,
  title={Multi-resource fairness: Objectives, algorithms and performance},
  author={Bonald, Thomas and Roberts, James},
  booktitle={Proceedings of the 2015 ACM SIGMETRICS International Conference on Measurement and Modeling of Computer Systems},
  pages={31--42},
  year={2015}
}

@article{reiss2011google,
  title={Google cluster-usage traces: format + schema},
  author={Reiss, Charles and Wilkes, John and Hellerstein, Joseph L},
  journal={Google Inc., White Paper},
  pages={1--14},
  year={2011}
}

@inproceedings{Friedman2014,
  title={Strategyproof allocation of discrete jobs on multiple machines},
  author={Friedman, Eric and Ghodsi, Ali and Psomas, Christos-Alexandros},
  booktitle={Proceedings of the 15th ACM conference on Economics and Computation (EC)},
  pages={529--546},
  year={2014}
}

@misc{Friedman2011,
  title={Strategyproofness, Leontief economies and the Kalai-Smorodinsky solution},
  author={Friedman, Eric J and Ghodsi, Ali and Shenker, Scott and Stoica, Ion},
  year={2011},
  booktitle={Manuscript},
  publisher={Citeseer}
}

@inproceedings{Ghodsi2011,
  title={Dominant Resource Fairness: Fair Allocation of Multiple Resource Types.},
  author={Ghodsi, Ali and Zaharia, Matei and Hindman, Benjamin and Konwinski, Andy and Shenker, Scott and Stoica, Ion},
  booktitle= {Proceedings of the 8th USENIX conference on Networked Systems Design and Implementation (NSDI)},
  pages = {323--336},
  year = {2011}
}

@article{Grandl2014,
  title={Multi-resource packing for cluster schedulers},
  author={Grandl, Robert and Ananthanarayanan, Ganesh and Kandula, Srikanth and Rao, Sriram and Akella, Aditya},
  journal={ACM SIGCOMM Computer Communication Review},
  volume={44},
  number={4},
  pages={455--466},
  year={2014},
  publisher={ACM New York, NY, USA}
}

@article{Jiang2021,
  title={Multi-resource allocation in cloud data centers: A trade-off on fairness and efficiency},
  author={Jiang, Suhan and Wu, Jie},
  journal={Concurrency and Computation: Practice and Experience},
  volume={33},
  number={6},
  pages={e6061},
  year={2021},
  publisher={Wiley Online Library}
}

@inproceedings{Jin2016,
  title={Efficiency comparison between proportional fairness and dominant resource fairness with two different type resources},
  author={Jin, Youngmi and Hayashi, Michiaki},
  booktitle={2016 Annual Conference on Information Science and Systems (CISS)},
  pages={643--648},
  year={2016}
}

@article{Joe-Wong2013,
  title={Multiresource allocation: Fairness--efficiency tradeoffs in a unifying framework},
  author={Joe-Wong, Carlee and Sen, Soumya and Lan, Tian and Chiang, Mung},
  journal={IEEE/ACM Transactions on Networking},
  volume={21},
  number={6},
  pages={1785--1798},
  year={2013},
  publisher={IEEE}
}

@article{Kash2014,
  title={No agent left behind: Dynamic fair division of multiple resources},
  author={Kash, Ian and Procaccia, Ariel D and Shah, Nisarg},
  journal={Journal of Artificial Intelligence Research},
  volume={51},
  pages={579--603},
  year={2014}
}

@article{Li2013,
  title={Egalitarian division under Leontief preferences},
  author={Li, Jin and Xue, Jingyi},
  journal={Economic Theory},
  volume={54},
  number={3},
  pages={597--622},
  year={2013},
  publisher={Springer}
}

@inproceedings{Li2017,
  title={Multi-resource fair allocation with bounded number of tasks in cloud computing systems},
  author={Li, Weidong and Liu, Xi and Zhang, Xiaolu and Zhang, Xuejie},
  booktitle={National Conference of Theoretical Computer Science (NCTCS)},
  pages={3--17},
  year={2017}
}

@inproceedings{Narayana2021,
  title={Fair and Efficient Allocations with Limited Demands},
  author={Narayana, Sushirdeep and Kash, Ian A},
  booktitle={Proceedings of the 35th AAAI Conference on Artificial Intelligence (AAAI)},
  volume={35},
  number={6},
  pages={5620--5627},
  year={2021}
}

@article{Nicolo2004,
  title={Efficiency and truthfulness with Leontief preferences. A note on two-agent, two-good economies},
  author={Nicol{\'o}, Antonio},
  journal={Review of Economic Design},
  volume={8},
  number={4},
  pages={373--382},
  year={2004},
  publisher={Springer}
}

@article{Parkes2015,
  title={Beyond dominant resource fairness: Extensions, limitations, and indivisibilities},
  author={Parkes, David C and Procaccia, Ariel D and Shah, Nisarg},
  journal={ACM Transactions on Economics and Computation (TEAC)},
  volume={3},
  number={1},
  pages={1--22},
  year={2015},
  publisher={ACM New York, NY, USA}
}

@inproceedings{Tahir2015,
  title={UDRF: Multi-resource fairness for complex jobs with placement constraints},
  author={Tahir, Yad and Yang, Shusen and Koliousis, Alexandros and McCann, Julie},
  booktitle={2015 IEEE Global Communications Conference (GLOBECOM)},
  pages={1--7},
  year={2015}
}

@inproceedings{Tang2016,
  title={Elastic multi-resource fairness: balancing fairness and efficiency in coupled CPU-GPU architectures},
  author={Tang, Shanjiang and He, BingSheng and Zhang, Shuhao and Niu, Zhaojie},
  booktitle={SC'16: Proceedings of the International Conference for High Performance Computing, Networking, Storage and Analysis},
  pages={875--886},
  year={2016}
}

@article{Tang2020,
  title={Fairness-efficiency scheduling for cloud computing with soft fairness guarantees},
  author={Tang, Shanjiang and Yu, Ce and Li, Yusen},
  journal={IEEE Transactions on Cloud Computing},
  pages={1-1},
  year={2020}
}

@inproceedings{Wang2014,
  title={Dominant resource fairness in cloud computing systems with heterogeneous servers},
  author={Wang, Wei and Li, Baochun and Liang, Ben},
  booktitle={IEEE INFOCOM 2014-IEEE Conference on Computer Communications},
  pages={583--591},
  year={2014}
}

@inproceedings{Wang2016,
  title={Multi-resource fair sharing for datacenter jobs with placement constraints},
  author={Wang, Wei and Li, Baochun and Liang, Ben and Li, Jun},
  booktitle={SC'16: Proceedings of the International Conference for High Performance Computing, Networking, Storage and Analysis},
  pages={1003--1014},
  year={2016}
}

@inproceedings{Dolev2012,
title = {No justified complaints: on fair sharing of multiple resources},
author = {Dolev, Danny and Feitelson, Dror G. and Halpern, Joseph Y. and Kupferman, Raz and Linial, Nathan},
booktitle={Innovations in Theoretical Computer Science 2012},
pages = {68--75},
year = {2012}
}

@inproceedings{Gutman2012,
title = {{Fair allocation without trade}},
author = {Gutman, Avital and Nisan, Noam},
booktitle = {Proceedings of the 11th International Conference on Autonomous Agents and Multiagent Systems (AAMAS)},
pages = {816--823},
year = {2012}
}

@inproceedings{ColeGG13,
  author       = {Richard Cole and
                  Vasilis Gkatzelis and
                  Gagan Goel},
  editor       = {Michael J. Kearns and
                  R. Preston McAfee and
                  {\'{E}}va Tardos},
  title        = {Mechanism design for fair division: allocating divisible items without
                  payments},
  booktitle    = {Proceedings of the fourteenth {ACM} Conference on Electronic Commerce,
                  {EC} 2013},
  pages        = {251--268},
  year         = {2013}
}

@inproceedings{fikioris2024incentives,
  title={Incentives in dominant resource fair allocation under dynamic demands},
  author={Fikioris, Giannis and Agarwal, Rachit and Tardos, {\'{E}}va},
  booktitle={International Symposium on Algorithmic Game Theory (SAGT)},
  pages={108--125},
  year={2024}
}

@inproceedings{BeiLL22,
  author       = {Xiaohui Bei and
                  Zihao Li and
                  Junjie Luo},
  title        = {Fair and Efficient Multi-resource Allocation for Cloud Computing},
  booktitle    = {International Conference on Web and Internet Economics {(WINE)}},
  pages        = {169--186},
  year         = {2022}
}

@inproceedings{GuoC10,
  author       = {Mingyu Guo and
                  Vincent Conitzer},
  editor       = {Wiebe van der Hoek and
                  Gal A. Kaminka and
                  Yves Lesp{\'{e}}rance and
                  Michael Luck and
                  Sandip Sen},
  title        = {Strategy-proof allocation of multiple items between two agents without
                  payments or priors},
  booktitle    = {9th International Conference on Autonomous Agents and Multiagent Systems
                  {(AAMAS} 2010), Toronto, Canada, May 10-14, 2010, Volume 1-3},
  pages        = {881--888},
  publisher    = {{IFAAMAS}},
  year         = {2010}
}

@inproceedings{Cheung16,
  author       = {Yun Kuen Cheung},
  editor       = {Subbarao Kambhampati},
  title        = {Better Strategyproof Mechanisms without Payments or Prior - An Analytic
                  Approach},
  booktitle    = {Proceedings of the Twenty-Fifth International Joint Conference on
                  Artificial Intelligence, {IJCAI} 2016, New York, NY, USA, 9-15 July
                  2016},
  pages        = {194--200},
  publisher    = {{IJCAI/AAAI} Press},
  year         = {2016}
}

@inproceedings{HanSTZ11,
  author       = {Li Han and
                  Chunzhi Su and
                  Linpeng Tang and
                  Hongyang Zhang},
  editor       = {Ning Chen and
                  Edith Elkind and
                  Elias Koutsoupias},
  title        = {On Strategy-Proof Allocation without Payments or Priors},
  booktitle    = {Internet and Network Economics - 7th International Workshop, {WINE}
                  2011, Singapore, December 11-14, 2011. Proceedings},
  series       = {Lecture Notes in Computer Science},
  volume       = {7090},
  pages        = {182--193},
  publisher    = {Springer},
  year         = {2011}
}

@article{ProcacciaT13,
  author       = {Ariel D. Procaccia and
                  Moshe Tennenholtz},
  title        = {Approximate Mechanism Design without Money},
  journal      = {{ACM} Trans. Economics and Comput.},
  volume       = {1},
  number       = {4},
  pages        = {18:1--18:26},
  year         = {2013}
}

@article{Friedman2017a,
author = {Friedman, Eric J. and Psomas, Christos Alexandros and Vardi, Shai},
isbn = {9781450345279},
journal = {EC 2017 - Proceedings of the 2017 ACM Conference on Economics and Computation},
pages = {461--478},
title = {{Controlled dynamic fair division}},
year = {2017}
}

@inproceedings{Amanatidis2017Truthful,
  author    = {Amanatidis, Georgios and Birmpas, Georgios and Christodoulou, George and Markakis, Evangelos},
  title     = {Truthful Allocation Mechanisms without Payments: Characterization and Implications on Fairness},
  booktitle = {Proceedings of the 2017 ACM Conference on Economics and Computation},
  series    = {EC '17},
  year      = {2017},
  pages     = {545--562},
  publisher = {ACM},
  address   = {Cambridge, MA, USA}
}

@inproceedings{BabaioffMorag2025Truthful,
  author    = {Babaioff, Moshe and Manaker Morag, Noam},
  title     = {On Truthful Mechanisms without Pareto-efficiency: Characterizations and Fairness},
  booktitle = {Proceedings of the 26th ACM Conference on Economics and Computation},
  series    = {EC '25},
  year      = {2025},
  pages     = {447},
  publisher = {ACM}
}

@article{BeiLL26,
  author  = {Xiaohui Bei and Zihao Li and Junjie Luo},
  title   = {Fair and Efficient Multi-resource Allocation for Cloud Computing: Beyond Dominant Resource Fairness},
  journal = {Mathematics of Operations Research},
  year    = {2026},
  note    = {{Articles in Advance}. Published online February 11, 2026}
}

@inproceedings{LuSunWangZhu10,
  author    = {Pinyan Lu and Xiaorui Sun and Yajun Wang and Zeyuan Allen Zhu},
  title     = {Asymptotically Optimal Strategy-Proof Mechanisms for Two-Facility Games},
  booktitle = {Proceedings of the 11th ACM Conference on Electronic Commerce (EC)},
  pages     = {315--324},
  year      = {2010}
}

@inproceedings{Barak26,
  author    = {Zohar Barak},
  title     = {Facility Location Mechanism Design: Breaking the Deterministic Barrier},
  booktitle = {Proceedings of the 27th ACM Conference on Economics and Computation (EC)},
  year      = {2026},
  note      = {To appear}
}

@techreport{fithiansolving,
  title={Solving Open Problems in Operations Research Using AI},
  author={Fithian, Eric and Niazadeh, Rad and Nuti, Pranav},
  year={2026},
  type= {Working Draft}
}

\appendix
\renewcommand{\theHsection}{appendix.\Alph{section}}
\renewcommand{\theHsubsection}{\theHsection.\arabic{subsection}}
\section{Proofs for Adaptive-Speed Fairness}\label{app:asf}

This appendix follows Section~\ref{sec:asf}: we prove the canonicalization lemma, implement the fixed-speed ASF framework, establish its sufficient SP condition, derive the population-dependent rule, and analyze \ATwo and the information lower bound. The canonical allocations and water levels are those of Section~\ref{sec:canonical}; the utility increments, residual state, speed rule, and constants $c_1,c_2$ are those of Section~\ref{sec:framework}.

\subsection{Canonicalization and water-filling costs}\label{app:canonical}

We use the water-level representation $A^{\mathrm c}(p,q)$ from Section~\ref{sec:canonical}. The group increments $\Delta S_1,\Delta S_2$ are defined in Section~\ref{sec:framework}; we suppress their arguments under the convention stated there. Thus the group totals are $|G_1|h+\Delta S_1$ and $|G_2|h+\Delta S_2$.

\subsubsection*{Proof of Lemma~\ref{lem:canonical}}
\begin{proof}
Let $A$ be a feasible SI--EF allocation. For each agent $i$, put $t_i=u_i(\mathbf A_i)$ and replace its bundle by $\widehat{\mathbf A}_i=t_i\mathbf d_i$. This preserves its utility and weakly reduces every resource allocation. Hence feasibility and SI are preserved. EF is also preserved: each agent's own utility is unchanged, while its valuation of every other agent's bundle can only decrease. The trimmed allocation therefore has the same group total utilities and no greater resource loads. It suffices to prove the claim for this non-wasteful allocation, which we continue to denote by $A$.

If a group is empty, SI and the capacity of the common dominant resource
force every utility to equal $h$, and the claim is immediate. Assume both
groups are nonempty. We first show that replacing the allocation inside
each group by its water-filling allocation preserves that group's total
utility and weakly reduces its non-dominant resource consumption.

Order the agents of $G_1$ so that $a_1\le\cdots\le a_s$. Write their
utilities in $A$ as $x_i$ and their resource-2 allocations as $z_i=a_ix_i$.
Within-group EF gives
\[
 x_1\ge\cdots\ge x_s,\qquad z_1\le\cdots\le z_s.
\]
Indeed, for $a_i\le a_j$, agent $i$ values agent $j$'s bundle at $x_j$,
whereas agent $j$ values agent $i$'s bundle at $z_i/a_j$.

Let $S=\sum_{i\in G_1}x_i$. The water-filling vector with this total
utility has the form
\[
 \bar x_i=\max\{h,p/a_i\},\qquad
 \bar z_i=a_i\bar x_i=\max\{a_ih,p\},\qquad
 \sum_i\bar x_i=S.
\]
Such a level $p\ge0$ exists because the sum is continuous,
starts at $sh\le S$, and tends to infinity. If $S=sh$, choose
$p=0$. Otherwise the sum is strictly increasing once it exceeds
$sh$, so the resulting vector is unique. Every agent with $\bar x_i>h$
has non-dominant allocation $p$, the smallest allocation in the
group. Thus the replacement follows the canonical formula in
Definition~\ref{def:canonical}.

Put $e_i=z_i-\bar z_i$. Preservation of total utility implies
\[
 \sum_i\frac{e_i}{a_i}=0.
\]
The negative entries of $(e_i)$ form a prefix: where $a_ih<p$,
$e_i=z_i-p$ is nondecreasing in $i$; on all remaining indices,
SI gives $e_i=z_i-a_ih\ge0$. The weights $1/a_i$ are nonincreasing.
If a negative entry exists, choose $\mu>0$ between the weights on the
negative and nonnegative sides. Then
\[
 0=\sum_i\frac{e_i}{a_i}\le\mu\sum_i e_i,
\]
so $\sum_i\bar z_i\le\sum_i z_i$. The same conclusion is immediate
when all $e_i\ge0$. The replacement therefore preserves $G_1$'s
resource-1 use and weakly reduces its resource-2 use. Applying the
symmetric argument to $G_2$ gives a level $q$ that preserves its
resource-2 use and weakly reduces its resource-1 use. Hence
$A^{\mathrm c}(p,q)$ has the asserted group totals, weakly decreases
both loads, and is feasible.

It remains to check fairness. The max formulas give SI. Within $G_1$,
$\bar x_i$ is nonincreasing and $\bar z_i$ is nondecreasing with $a_i$,
which gives both directions of pairwise EF by the utility calculations
above. The argument for $G_2$ is symmetric. Feasibility implies
$p\le h$: otherwise every $G_1$ agent would consume more than $h$ of
resource~2, while every $G_2$ agent consumes at least $h$ of it, exceeding
the capacity $nh=1$. Similarly $q\le h$. Thus every non-dominant
allocation is at most $h$. Each agent therefore values every bundle
in the opposite group at most $h$, and has own utility at least $h$,
which proves cross-group EF.

This argument also establishes SI and EF for every feasible canonical
allocation, as used in the ASF construction. Applying the reduction
to a fair optimum preserves its welfare, proving the last claim.
\end{proof}

\paragraph{Resource costs along the canonical trajectories.}
For a nonempty group $G_k$ and a utility increment $t\ge0$, let $C_k(t)$ be the additional non-dominant resource consumed along its canonical trajectory when its total utility is $|G_k|h+t$. That total determines a unique group allocation, so the cost is well defined even at the initial water-level plateau. For an empty group, set $C_k(0)=0$. Consequently, $A^{\mathrm c}(p,q)$, with corresponding increments $\Delta S_1,\Delta S_2$, is feasible exactly when
\[
 \Delta S_1+C_2(\Delta S_2)\le R_1,\qquad
 C_1(\Delta S_1)+\Delta S_2\le R_2.
\]
We record their properties directly in the operational notation; Appendix~\ref{app:two-geometry} relates these costs to the resource loads used in the \CRB construction.

\begin{lemma}\label{lem:asf-cost-curves}
For each nonempty group, $C_k$ is continuous, strictly increasing,
convex, and piecewise linear, with $C_k(0)=0$ and slopes in $(0,1]$.
In particular, for $0\le s\le s'$,
\[
 0\le C_k(s')-C_k(s)\le s'-s.
\]
\end{lemma}
\begin{proof}
On an interval where the set $P$ receiving utility above baseline is
fixed, raising its common non-dominant allocation by $d\lambda$ uses
$|P|d\lambda$ of that resource and increases utility by
$D\,d\lambda$, where $D=\sum_{i\in P}1/d_{i,3-k}$. Thus
$C_k'=|P|/D\in(0,1]$. Agents enter in nondecreasing order of their
non-dominant demand. If the next demand is $a$, then $1/a\le D/|P|$,
so $(|P|+1)/(D+1/a)\ge |P|/D$. Slopes are therefore nondecreasing.
The allocations agree at activation thresholds, which gives
continuity. Integrating the slope bounds proves the inequality.
These properties also hold when the trajectory is continued beyond
the original capacities.
\end{proof}

\subsection{Implementation and basic properties}\label{app:asf-implementation}

The algorithm follows the canonical trajectories in
Section~\ref{sec:canonical}. It computes $\theta$ once from the
initial augmented residual state. The variables $\widehat R_1,
\widehat R_2$ below track the capacities remaining during Step~2;
they are distinct from the fixed $R_1^*,R_2^*$ and the SP constants
$c_1,c_2$. We use $\min\emptyset=+\infty$ for a group with no further
activation threshold.

\begin{algorithm}[t]
\DontPrintSemicolon
\SetAlFnt{\small\linespread{1.15}\selectfont}
\caption{$\mathrm{ASF}_{\Theta}(\mathbf d_1,\ldots,\mathbf d_n)$}
\label{alg:ASF}
Set $h\leftarrow1/n$ and form the dominant-resource groups $G_1,G_2$\;
Set $\mathbf A_i\leftarrow h\mathbf d_i$ for every $i$ and
$\widehat R_r\leftarrow1-\sum_i A_{ir}$ for $r=1,2$\;
\If{a group is empty or $\min\{\widehat R_1,\widehat R_2\}=0$}{
  \Return $\mathbf A$\;
}
Relabel resources and groups so that $|G_2|\le|G_1|$, carrying each
agent's designated dominant resource with the relabeling\;
Set $\alpha\leftarrow |G_2|/n$ and
$R_r^*\leftarrow\widehat R_r+h\min_{i\in G_{3-r}}d_{ir}$ for $r=1,2$\;
Set $\theta\leftarrow\Theta(\alpha,R_1^*,R_2^*)$\;
\While{$\widehat R_1>0$ and $\widehat R_2>0$}{
  \ForEach{$k\in\{1,2\}$}{
    $\ell_k\leftarrow\min_{i\in G_k}A_{i,3-k}$ and
    $P_k\leftarrow\{i\in G_k:A_{i,3-k}=\ell_k\}$\;
    $D_k\leftarrow\sum_{i\in P_k}1/d_{i,3-k}$\;
    $e_k\leftarrow\min_{i\in G_k\setminus P_k}(A_{i,3-k}-\ell_k)$\;
  }
  $\lambda\leftarrow\min\left\{
    \dfrac{D_1e_1}{\theta},\ D_2e_2,\
    \dfrac{\widehat R_1}{\theta+|P_2|/D_2},\
    \dfrac{\widehat R_2}{1+\theta|P_1|/D_1}
  \right\}$\;
  Set $v_1\leftarrow\theta\lambda/D_1$ and $v_2\leftarrow\lambda/D_2$\;
  \ForEach{$k\in\{1,2\}$ and $i\in P_k$}{
    $\mathbf A_i\leftarrow\mathbf A_i+(v_k/d_{i,3-k})\mathbf d_i$\;
  }
  Update $\widehat R_r\leftarrow1-\sum_i A_{ir}$ for $r=1,2$\;
}
\Return $\mathbf A$, expressed in the original resource labels\;
\end{algorithm}

In one iteration, the increases in group utility are
$\theta\lambda$ and $\lambda$. The corresponding resource uses are
$(\theta+|P_2|/D_2)\lambda$ and
$(1+\theta|P_1|/D_1)\lambda$. The four bounds defining $\lambda$
therefore stop at the next activation threshold or resource
exhaustion. This proves that the algorithm implements the fixed
ratio in~\eqref{eq:asf-speed} exactly.

\subsubsection*{Proof of Lemma~\ref{lem:ASF-3-properties}}
\begin{proof}
Algorithm~\ref{alg:ASF} starts at the SI baseline and preserves the canonical within-group trajectories. Its step-size bounds maintain feasibility and make the two utility increments grow in ratio $\theta$. Thus every intermediate allocation has the form $A^{\mathrm c}(p,q)$ at the current water levels, with the corresponding increments $\Delta S_1,\Delta S_2$, and is SI and EF by Lemma~\ref{lem:canonical-fairness}. Feasibility ensures that the levels of nonempty groups do not exceed $h$: a larger level would give that group more than $h$ per agent of its non-dominant resource, while every other agent receives at least $h$ of the same resource.

If the SI baseline exhausts a resource, positivity of demands and SI force every agent's utility to remain $h$. Otherwise, positive speed and positive residual capacities allow Step~2 to continue until a resource is exhausted. Since the output is non-wasteful and every demand coordinate is positive, increasing an agent's utility requires more of that exhausted resource; feasibility then forces another agent's utility to decrease. The allocation is therefore PO.

In each iteration, either a resource is exhausted or a new agent joins one of the active sets. There are at most $n$ activation events and one terminal iteration. Finding the active sets and next thresholds, computing the step, and updating allocations each take $O(n)$ time. Initialization also takes $O(n)$ time, so the total running time, given the speed, is $O(n^2)$.
\end{proof}

\subsection{Sufficient Condition for SP}
\label{sec:sp-condition}

We prove Lemma~\ref{lem:sufficient-SP} using the following bounds on profitable deviations.
Suppose that some agent has a profitable deviation. 
By symmetry, let the deviating agent be \(i_0\in G_1\), with true
demand vector \(\mathbf d_{i_0}=(1,d_{i_0,2})\). Let \(\mathbf d'_{i_0}\)
be a misreport that strictly increases its utility.

\begin{lemma}
\label{lem:no-group-switch}
A profitable deviation cannot switch the deviating agent's dominant-resource group. 
\end{lemma}

\begin{proof}
Under \ASF, every agent receives at least \(1/n\) of its dominant resource and at most \(1/n\) of its non-dominant resource. If an agent reports a different dominant resource, then its true dominant resource is treated as non-dominant under the reported type and hence can be allocated at most \(1/n\). Since the truthful allocation already gives the agent at least \(1/n\) of its true dominant resource, such a report cannot strictly improve its utility, a contradiction.
\end{proof}

By Lemma~\ref{lem:no-group-switch}, \(i_0\) remains in \(G_1\), so \(\alpha\) is unchanged and we may write $\mathbf{d}'_{i_0}=(1,d'_{i_0,2})$.
Let $\mathbf{A}$ and $\mathbf{A}'$ denote the allocations under truthful and manipulated reports, respectively.
The comparisons use fixed physical resource labels; relabeling a majority and minority group is only a change of notation.

Let $\Delta S_1,\Delta S_2$ (resp.\ $\Delta S'_1,\Delta S'_2$) be the total increases of dominant shares for $G_1$ and $G_2$ in Step~2 under $\mathbf{A}$ (resp.\ $\mathbf{A}'$).
Define
\[
x  = \min_{i\in G_1} A_{i,2}, 
\qquad
x' = \min_{i\in G_1} A'_{i,2}.
\]

We first show that a profitable deviation must raise the final
within-group water level of \(G_1\).

\begin{lemma}
\label{lem:x-increase}
Any profitable deviation by \(i_0\) satisfies \(x'>x\).
\end{lemma}

\begin{proof}
Since $i_0$'s utility increases under $\mathbf{A}'$, we must have $A'_{i_0,1}>A_{i_0,1}$ and $A'_{i_0,2}>A_{i_0,2}$.
Because $A_{i_0,1}\ge 1/n$, the increase $A'_{i_0,1}>A_{i_0,1}$ must come from Step~2, so $A'_{i_0,2}$ equals the minimum non-dominant share in $G_1$, i.e., $A'_{i_0,2}=x'$.
By definition of $x$, we have $A_{i_0,2}\ge x$, hence
\[
x' = A'_{i_0,2} > A_{i_0,2} \ge x.
\]
\end{proof}

We next use the monotonicity of \(\Theta\) to identify the direction of the
deviation. Let \(R_1^{\prime *}\) and \(R_2^{\prime *}\) denote the
augmented remaining resources under the deviated report.

\begin{lemma}
\label{lem:F2sext-SP-d-increase}
Any profitable deviation by \(i_0\) satisfies
\[
    d'_{i_0,2}>d_{i_0,2}
    \qquad\text{and}\qquad
    R_2^{\prime *}<R_2^*.
\]
Moreover, \(i_0\) is not a minimizer of the non-dominant demand in \(G_1\)
after the deviation.
\end{lemma}

\begin{proof}
By Lemma~\ref{lem:x-increase}, \(x'>x\). So every agent in $G_1\setminus\{i_0\}$ receives at least as much of both resources in $\mathbf{A}'$ as in $\mathbf{A}$, and $i_0$ itself is strictly better off.
Consequently, $G_1$ as a whole receives more of both resources under $\mathbf{A}'$. Since the truthful outcome exhausts a resource, $G_2$ must receive less of that resource. Its unchanged within-group water-filling trajectory then implies that it receives less of both resources.
Therefore $\Delta S'_1>\Delta S_1$ and $\Delta S'_2<\Delta S_2$, so
\[
\frac{\Delta S'_1}{\Delta S'_2}
\;>\;
\frac{\Delta S_1}{\Delta S_2}.
\]
Since \(\alpha\) and \(R_1^*\) are unchanged, the speed condition
$\tfrac{\Delta S_1}{\Delta S_2}
=\Theta(\alpha,R_1^*,R_2^*)$ and the monotonicity of $\Theta$ in
its third coordinate imply that $R_2^{\prime *}<R_2^*$.
If $G_1$ were a singleton, $R_2^*=h$ would be independent of its report, contradicting the strict decrease. Thus $|G_1|\ge2$. By the definition of $R_2^*$, this can only happen if \(d'_{i_0,2}>d_{i_0,2}\) and $d'_{i_0,2}>\min_{j\in G_1\setminus\{i_0\}} d_{j,2}$.
\end{proof}

We now quantify the decrease in $R_2^*$.
If $i_0$ is not initially a minimizer of the non-dominant demand, the decrease is $(d'_{i_0,2}-d_{i_0,2})/n$.
If $i_0$ is initially a minimizer, increasing its report up to $\min_{j\in G_1\setminus\{i_0\}}d_{j,2}$ leaves $R_2^*$ unchanged; only the increase beyond this threshold contributes to the decrease.
Combining these two cases gives
\begin{equation}
\label{eq:delta}
    \delta \coloneqq R_2^*-R_2^{\prime *} =
 \frac{1}{n}\Bigl(d'_{i_0,2} - \max\{d_{i_0,2},\min_{j\in G_1\setminus\{i_0\}} d_{j,2}\}\Bigr)
 >0.
\end{equation}

Next we analyze the difference between the total Step~2 dominant-share increments of the two groups before and after the manipulation.
\cite{BeiLL22} show that for \FTwos it holds that $\Delta S_1' > \Delta S_1$ and $\Delta S_2' \le \Delta S_2 -\delta$.
To quantify this change, we first bound from below the ratio of the dominant-share increment to the non-dominant-resource increment in each group during Step~2.
Recall the constants $c_1,c_2$ from~\eqref{eq:finite-c12}.
The next lemma gives a uniform lower bound on the dominant-share gain over any interval of the within-group water-filling process.

\begin{lemma}
\label{lem:wf-marginal-bound}
For each group $k\in\{1,2\}$, let $H_k(v)$ denote the utility increment produced when its canonical water-filling trajectory consumes $v$ additional units of non-dominant resource; equivalently, $H_k=C_k^{-1}$. For any $0\le v<v'\le R_{3-k}^*$,
\[
H_k(v')-H_k(v)\ge c_k(v'-v).
\]
\end{lemma}

\begin{proof}
Fix \(k\in\{1,2\}\). 
Let $\gamma_k:=\frac{|G_k|}{n}$,
so that \(\gamma_1=1-\alpha\) and \(\gamma_2=\alpha\).
By the definition of \(R_{3-k}^*\),
\begin{equation}\label{eq:average-demand-bound}
    \frac1{|G_k|}\sum_{i\in G_k}d_{i,3-k}
    =
    1-\frac{R_{3-k}}{\gamma_k}
    =
    1-\frac{R_{3-k}^*-\frac1n\min_{i\in G_k}d_{i,3-k}}{\gamma_k}
    \le
    1-\frac{R_{3-k}^*-1/n}{\gamma_k}
    =
    \frac1{c_k}.
\end{equation}
At any point during the within-group water filling, the active set \(P \subseteq G_k\) consists of agents with the smallest non-dominant demands. Hence,
\[
    \frac1{|P|}\sum_{i\in P}d_{i,3-k}
    \le
    \frac1{|G_k|}\sum_{i\in G_k}d_{i,3-k}
    \le
    \frac1{c_k},
\]
where the second inequality follows from (\ref{eq:average-demand-bound}).
If one additional unit of the non-dominant resource is distributed equally
among the agents in \(P\), the resulting increase in total dominant share is
\[
    \frac{1}{|P|}\sum_{i\in P}\frac1{d_{i,3-k}}.
\]
By the AM--HM inequality,
\[
    \frac{1}{|P|}\sum_{i\in P}\frac1{d_{i,3-k}}
    \ge
    \left(\frac1{|P|}\sum_{i\in P}d_{i,3-k}\right)^{-1}
    \ge
    c_k.
\]
Thus the marginal dominant-share increase per unit of non-dominant resource is at least \(c_k\) throughout the water-filling process. Integrating from $v$ to $v'$ gives the claimed inequality.
\end{proof}

Define $\bar c \coloneqq \min\{c_1c_2,2\}$.
We next derive sharper bounds on \(\Delta S_k'\) relative to \(\Delta S_k\) than those established in \cite{BeiLL22}.

\begin{lemma}
\label{lem:amplified-slack-change}
\(\Delta S_1'>\Delta S_1+c_1\delta\) and \(\Delta S_2'<\Delta S_2-\bar c\delta\).
\end{lemma}

\begin{proof}
Write $a=d_{i_0,2}$ and $a'=d'_{i_0,2}$, and choose $j\in G_1\setminus\{i_0\}$ with the smallest non-dominant demand $b=d_{j,2}$ among these agents. By the definition of $\delta$ in~\eqref{eq:delta}, we have 
\[a'\ge a+n\delta, \qquad a'\ge b+n\delta.\]
Since the deviation is profitable, $A'_{i_0,1}>A_{i_0,1}\ge1/n$, and $i_0$ is active after the deviation. Hence
\[
x'=a'A'_{i_0,1}>(a+n\delta)A_{i_0,1}
\ge A_{i_0,2}+\delta\ge x+\delta,
\]
and
\[
x'=a'A'_{i_0,1}>a'/n\ge b/n+\delta
\]
The water-filling rule gives $A_{j,2}=\max\{b/n,x\}$ and $A'_{j,2}=x'$. 
It follows from the preceding two inequalities that
\[
    A'_{j,2}-A_{j,2}
    =
    x'-\max\{b/n,x\}
    >\delta.
\]
Thus both $i_0$ and $j$ receive more than $\delta$ additional units of resource~2. Every other agent in $G_1$ receives weakly more of both resources.

For the first bound, we translate agent $j$'s increase in resource~2 into
an increase in its dominant resource. Let $s=|G_1|=(1-\alpha)n$. 
Recall that
\[
\begin{aligned}
\frac1{c_1}
&=
1-\frac{R_2^*-1/n}{s/n}\\
&=
1-
\frac{
\frac{s}{n}
-\frac1n\sum_{i\in G_1}d_{i,2}
+\frac1n\min_{i\in G_1}d_{i,2}
-\frac1n
}{
s/n
}\\
&=
\frac{
1+\sum_{i\in G_1}d_{i,2}
-\min_{i\in G_1}d_{i,2}
}{s}.
\end{aligned}
\]
We claim that
\[
    \sum_{i\in G_1}d_{i,2}
    -\min_{i\in G_1}d_{i,2}
    \ge (s-1)b.
\]
Indeed, if $i_0$ is a minimizer, remove the demand of $i_0$; all remaining
agents belong to $G_1\setminus\{i_0\}$ and hence have demand at least $b$.
If $i_0$ is not a minimizer, then $j$ is a minimizer of $G_1$; remove the
demand of $j$, and every remaining demand, including that of $i_0$, is at
least $b$. Therefore, since $b\le1$,
\[
    \frac1{c_1}
    \ge
    \frac{1+(s-1)b}{s}
    =
    b+\frac{1-b}{s}
    \ge b.
\]
Agent $j$ therefore receives more than $\delta/b\ge c_1\delta$ additional units of resource~1. All other agents in $G_1$ receive weakly more of resource~1. Hence
$\Delta S_1'-\Delta S_1>c_1\delta$.

For the second bound, we distinguish which resource is exhausted under truthful reporting.
Suppose first that resource~1 is exhausted. Since $\Delta S_1'-\Delta S_1>c_1\delta$, feasibility after the deviation forces the resource~1 allocation to $G_2$ to decrease by more than $c_1\delta$. The reports of agents in $G_2$ are unchanged, so Lemma~\ref{lem:wf-marginal-bound}, applied to the interval between its two water-filling outcomes, gives
$\Delta S_2-\Delta S_2'>c_1c_2\delta$.
Suppose instead that resource~2 is exhausted. The inequalities above show
that both $i_0$ and $j$ receive more than $\delta$ additional units of
resource~2, while every other agent in $G_1$ receives weakly more.
Hence, the resource~2 allocation to $G_2$ decreases by more than $2\delta$. 
Since resource~2 is dominant for $G_2$, we obtain
$\Delta S_2-\Delta S_2'>2\delta$.
In either case, $\Delta S_2'<\Delta S_2-\bar c\delta$.
\end{proof}

\subsubsection*{Proof of Lemma~\ref{lem:sufficient-SP}}

\begin{proof}
Suppose, toward a contradiction, that some agent $i_0$ has a profitable deviation.
We consider the case \(i_0\in G_1\) and use the first condition; the case \(i_0\in G_2\) is analogous and
uses the second condition.
By Lemma~\ref{lem:no-group-switch}, $i_0$ remains in $G_1$ and \(\alpha\) is unchanged.
By Lemma~\ref{lem:F2sext-SP-d-increase}, its misreport has the form \((1,d'_{i_0,2})\) with \(d'_{i_0,2}>d_{i_0,2}\).
This decreases \(R_2^*\) to \(R_2^*-\delta\) for \(\delta>0\) defined by \eqref{eq:delta}, while leaving \(R_1^*\) unchanged.
Hence the new speed ratio becomes $\Theta(\alpha,R_1^*,R_2^*-\delta)$.
By Lemma~\ref{lem:amplified-slack-change}, \(\Delta S_1'>\Delta S_1+c_1\delta\), and \(\Delta S_2'<\Delta S_2-\bar c\delta\).
If \(\bar c\delta\ge R_2^*\), then \(\Delta S_2\le R_2^*\) implies \(\Delta S_2'<0\), a contradiction.
Therefore, \(\bar c\delta<R_2^*\), so the first condition applies.
On the other hand, 
\[
    \frac{
        \Theta(\alpha,R_1^*,R_2^*-\delta)
    }{
        \Theta(\alpha,R_1^*,R_2^*)
    }
    =
    \frac{\Delta S_1'/\Delta S_2'}
         {\Delta S_1/\Delta S_2}
    >
    \frac{
        1+c_1\delta/\Delta S_1
    }{
        1-\bar c\delta/\Delta S_2
    }
    \ge
    \frac{
        1+c_1\delta/R_1^*
    }{
        1-\bar c\delta/R_2^*
    },
\]
where the strict inequality follows from Lemma~\ref{lem:amplified-slack-change}, and the last inequality follows from \(\Delta S_1\le R_1^*\) and \(\Delta S_2\le R_2^*\).
This contradicts the first condition.
\end{proof}

Corollary~\ref{cor:population-rescaling} follows because a profitable deviation preserves the group sizes, so the positive population factor cancels from both speed ratios. We next verify the balanced rule used as the starting point for that corollary.

The sufficient condition also recovers the SP of previous mechanisms naturally within the \ASF framework. 
For example, \FTwos uses the speed rule $\Theta(\alpha,R_1^*,R_2^*)=\frac{R_1^*}{R_2^*}$. For a deviation by an agent in \(G_1\), 
\[
\frac{
    \Theta(\alpha,R_1^*,R_2^*-\delta)
}{
    \Theta(\alpha,R_1^*,R_2^*)
}
=
\frac{R_2^*}{R_2^*-\delta}
=
\frac{1}{1-\delta/R_2^*}.
\]
Since \(c_1,c_2\ge1\) by~\eqref{eq:finite-c12}, we have \(\bar c=\min\{c_1c_2,2\}\ge1\). Therefore, whenever \(\bar c\delta<R_2^*\),
\[
\frac{1}{1-\delta/R_2^*}
\le
\frac{1}{1-\bar c\delta/R_2^*}
\le
\frac{
    1+c_1\delta/R_1^*
}{
    1-\bar c\delta/R_2^*
}.
\]
The condition for deviations by agents in $G_2$ follows by exchanging the two groups and the two resources. Thus \FTwos satisfies both conditions of Lemma~\ref{lem:sufficient-SP}.

\subsection{Technical Lemmas for Fair-Ratio Analysis}
\label{sec:fair-ratio-tools}

We compare the ASF allocation with a fair-optimal allocation of the same instance. We first show that an optimal SI--EF allocation can be obtained by the same two-step water-filling procedure as \(\ASF\), with its own choice of speed ratio. We then compare the Step~2 gains of these two allocations, and derive a general \fratio upper bound for \(\ASF\) mechanisms, which gives the population-factor optimization in Appendix~\ref{app:asf-speed-rules}.

\begin{lemma}
\label{lem:the-optimal}
For any instance, there exists a possibly boundary speed ratio
\(\theta^*\in[0,\infty]\) such that the corresponding two-step water-filling rule
returns an optimal allocation satisfying SI and EF.
\end{lemma}

Next we compare the \(\ASF\) allocation and the above optimal fair allocation.
Fix an instance \(\mathbf I\), the allocation \(\mathbf A\) produced by \(\ASF\) with speed \(\theta\), and an optimal allocation \(\mathbf A^*\) selected by Lemma~\ref{lem:the-optimal} with speed \(\theta^*\).
Let \(\Delta S_k\) and \(\Delta S_k^*\) be their total Step~2 dominant-share increments in group \(G_k\).
Let
\[
    r=\frac{R_1^*}{R_2^*},
    \qquad
    \theta=\frac{\Delta S_1}{\Delta S_2},
    \qquad
    \vartheta =\frac{\theta}{r}
    =
    \frac{\Delta S_1/R_1^*}{\Delta S_2/R_2^*}.
\]
Here \(\vartheta \) is the speed ratio normalized by the remaining-resource ratio.
The next lemma compares the Step~2 gains of the mechanism and the fair optimum.

\begin{lemma}
\label{lem:S2/S2*}
If resource~1 is exhausted by the \(\ASF\) allocation, then \(\theta\ge\theta^*\), so \(\Delta S_1\ge\Delta S_1^*\) and \(\Delta S_2\le\Delta S_2^*\). The loss on group \(G_2\) is bounded by \(\Delta S_2\ge\Delta S_2^*/(1+\vartheta )-1/n\).
Symmetrically, if resource~2 is exhausted by the \(\ASF\) allocation, then \(\theta\le\theta^*\), so \(\Delta S_1\le\Delta S_1^*\) and \(\Delta S_2\ge\Delta S_2^*\). The loss on group \(G_1\) is bounded by \(\Delta S_1\ge\Delta S_1^*/(1+1/\vartheta )-1/n\).
\end{lemma}

Using the above loss bounds, we can derive the following general upper bound on \fratio for \ASF mechanisms in terms of $\alpha$ and $\vartheta =\theta/r$.

\begin{lemma}
\label{lem:general-ASF-upper-bound}
For every instance $\mathbf I$, the ratio of its optimal SI--EF welfare to the welfare returned by an \ASF mechanism is at most
\[
\max \left\{\frac{2-\alpha}{1+ \frac{1}{1+\vartheta }(1-\alpha)},\frac{1+\alpha}{1+ \frac{\vartheta }{1+\vartheta }\alpha}\right\}+O\left(\frac{1}{n}\right),
\]
where the two terms correspond to the cases in which resource~1 or resource~2 is exhausted, respectively.
\end{lemma}

\subsubsection*{Proof of Lemma~\ref{lem:the-optimal}}
\phantomsection\label{app:lem-the-optimal}
\begin{proof}
By Lemma~\ref{lem:canonical}, choose a canonical fair optimum and let its group increments be $\Delta S_1^*,\Delta S_2^*$. Use their ratio as $\theta^*$, with zero or infinity when one increment is zero. If both are zero, any speed returns the exhausted SI allocation. Otherwise, water filling with this fixed ratio reaches both prescribed increments simultaneously. All earlier allocations use weakly less of each resource, so the execution remains feasible up to that point. If it could continue, its welfare would strictly exceed the fair optimum. Therefore it stops at that optimum. The same argument applies to a boundary speed, with the other group fixed at its SI allocation.
\end{proof}

\subsubsection*{Proof of Lemma~\ref{lem:S2/S2*}}
\phantomsection\label{app:lem-s2-s2-star}

\begin{proof}
If the SI allocation exhausts a resource, both allocations give every agent utility \(1/n\); choosing \(\theta^*=\theta\) makes the claims immediate. Otherwise, both Step~2 increments of \(\ASF\) are positive.
Consider the case in which resource~1 is exhausted under the \(\ASF\) allocation.
By Lemma~\ref{lem:the-optimal}, the two allocations follow the same within-group water-filling rules.
If \(\Delta S_1^*>\Delta S_1\), the optimal allocation must take at least \(\Delta S_1^*-\Delta S_1\) units of resource~1 away from \(G_2\). Moving back along its water-filling trajectory weakly decreases every agent's allocation. Since the baseline has positive resource-1 slack, some $G_2$ agent has non-dominant demand strictly below one. Such an agent is active whenever this group's water-filling trajectory grows. The utility loss in $G_2$ therefore strictly exceeds the resource-1 amount released, and hence exceeds the gain in $G_1$, contradicting optimality.
Thus \(\Delta S_1^*\le\Delta S_1\). Optimality then gives \(\Delta S_2^*\ge\Delta S_2\), and hence \(\theta^*\le\theta\).

During Step~2, the active subset of group~2 expands in water-filling order, so agents with weakly larger non-dominant demand \(d_{j,1}\) enter later.
Consequently, the cumulative ratio of resource~1 to resource~2 used by group~2 is nondecreasing.
Put \(\beta=(R_1-\Delta S_1)/\Delta S_2\), the ratio under \(\ASF\).
Since \(\Delta S_2^*\ge\Delta S_2\), group~2 uses at least \(\beta\Delta S_2^*\) units of resource~1 under the optimal allocation. Feasibility therefore gives
\[
\beta\Delta S_2^*\le R_1-\Delta S_1^*\le R_1\le R_1^*.
\]
Since \(\Delta S_2^*\le R_2\le R_2^*\) and \(\theta R_2^*=\vartheta R_1^*\), it follows that
\[
(\theta+\beta)\Delta S_2^*\le(1+\vartheta )R_1^*.
\]
Every unit of dominant-share growth in \(G_2\) uses at least \(\min_{j\in G_2}d_{j,1}\) of resource~1. Thus \(R_1^*-R_1\le\beta/n\). Resource~1 balance now gives
\[
\Delta S_2
=\frac{R_1}{\theta+\beta}
\ge\frac{R_1^*}{\theta+\beta}-\frac1n
\ge\frac{\Delta S_2^*}{1+\vartheta }-\frac1n.
\]
The resource~2 case is symmetric. For the utility comparison, if Step~2 is nontrivial, the smallest non-dominant demand in \(G_1\) is strictly below one. An agent with this demand is active whenever the group's water-filling allocation grows, so taking resource~2 away from \(G_1\) again causes a strictly larger utility loss.
\end{proof}

\subsubsection*{Proof of Lemma~\ref{lem:general-ASF-upper-bound}}
\phantomsection\label{app:lem:general-ASF-upper-bound}

\begin{proof}
If the mechanism is fair-optimal, its welfare ratio is one and the bound is immediate. Otherwise, suppose first that resource~1 is exhausted.
By Lemma~\ref{lem:S2/S2*}, we have $\Delta S_1 \ge \Delta S_1^*$ and $\Delta S_2 \ge \frac{\Delta S_2^*}{1+\vartheta }-\frac1n$.
Together with \(\Delta S_2^*\le R_2\le 1-\alpha\), this gives
\begin{equation*}
    \begin{aligned}
\frac{1+\Delta S^*_{1}+\Delta S^*_{2}}{1+\Delta S_{1}+\Delta S_{2}}
\le \frac{1+\Delta S^*_{2}}{1+\Delta S_{2}}
&\le \frac{1+\Delta S^*_{2}}{1+\frac{1}{1+\vartheta }\Delta S^*_{2}}+\frac2n
\le\frac{2-\alpha}{1+ \frac{1}{1+\vartheta }(1-\alpha)}+\frac2n.
    \end{aligned}
\end{equation*}
If instead resource~2 is exhausted, an analogous argument yields
\[
\frac{1 + \Delta S_1^* + \Delta S_2^*}
     {1 + \Delta S_1 + \Delta S_2}
\le
\frac{1+\Delta S_1^*}
     {1+\frac{\vartheta }{1+\vartheta }\Delta S_1^*}+\frac2n
\le \frac{1+\alpha}{1+ \frac{\vartheta }{1+\vartheta }\alpha}+\frac2n.
\]
In both cases, the error is at most \(2/n\) because the numerator is at most two and both denominators are at least one.
\end{proof}

\subsection{Population adaptation: the intermediate rule}\label{app:asf-speed-rules}

We derive the factor $f_1(\alpha)$ in~\eqref{eq:asf1-speed}, stated in Section~\ref{sec:mechanisms}. For the speed $\Theta=f(\alpha)r$, the normalized speed in Lemma~\ref{lem:general-ASF-upper-bound} is $\vartheta=f(\alpha)$. Its first bound increases with $\vartheta$, and its second decreases. Their maximum is minimized by equating them:
\[
 \frac{(2-\alpha)(1+\vartheta)}{2-\alpha+\vartheta}
 =\frac{(1+\alpha)(1+\vartheta)}{1+(1+\alpha)\vartheta}
 \quad\Longleftrightarrow\quad
 \vartheta=f_1(\alpha)=\frac{2\alpha-\alpha^2}{1-\alpha^2}.
\]
Let $\AOne$ be the ASF mechanism with this speed factor.

The proof of Theorem~\ref{thm:ASF1-ratio} gives the sharper instance-dependent bound
\[
\frac{\OPT(\mathbf I)}{\SW(\AOne(\mathbf I))}
\le\frac{1+2\alpha-2\alpha^2}{1+\alpha-\alpha^2}+O(1/n),
\]
which is asymptotically tight for each fixed minority proportion $\alpha$.

\subsubsection*{Proof of Theorem~\ref{thm:ASF1-ratio}}
\label{app:thm:ASF1-ratio}

\begin{proof}
SI, EF, and PO follow from Lemma~\ref{lem:ASF-3-properties}.
For SP, the \FTwos speed rule \(\Theta(\alpha,R_1^*,R_2^*)=r\) satisfies the two conditions in Lemma~\ref{lem:sufficient-SP}. Since \(f_1(\alpha)>0\) for \(0<\alpha\le1/2\), Corollary~\ref{cor:population-rescaling} implies SP for \AOne.

We next prove the upper bound. Fix an instance $\mathbf I$ and let $\mathbf A^*$ be a fair-optimal allocation for this instance. For \AOne, the normalized speed in Lemma~\ref{lem:general-ASF-upper-bound} is
\[
    \vartheta =f_1(\alpha)=\frac{2\alpha-\alpha^2}{1-\alpha^2}.
\]
Substituting this value into the two cases of Lemma~\ref{lem:general-ASF-upper-bound} gives the same expression.
Indeed,
\[
    \frac{2-\alpha}
         {1+\frac{1-\alpha}{1+f_1(\alpha)}}
    =
    \frac{1+2\alpha-2\alpha^2}{1+\alpha-\alpha^2},
    \qquad
    \frac{1+\alpha}
         {1+\frac{f_1(\alpha)}{1+f_1(\alpha)}\alpha}
    =
    \frac{1+2\alpha-2\alpha^2}{1+\alpha-\alpha^2}.
\]
Therefore
\[
    \frac{\SW(\mathbf A^*)}{\SW(\AOne(\mathbf I))}
    \le
    \frac{1+2\alpha-2\alpha^2}{1+\alpha-\alpha^2}
    +O\!\left(\frac1n\right).
\]
Finally, the function
\[
    \phi(\alpha)
    =
    \frac{1+2\alpha-2\alpha^2}{1+\alpha-\alpha^2}
\]
is increasing on \(0<\alpha\le1/2\), since
\[
    \phi'(\alpha)
    =
    \frac{1-2\alpha}{(1+\alpha-\alpha^2)^2}
    \ge 0.
\]
Thus \(\phi(\alpha)\le \phi(1/2)=6/5\), proving the claimed upper bound.

To show tightness, fix a rational \(0<\alpha\le1/2\) and choose population sizes with \(n\alpha\) integral.
Group \(G_1\) contains \(n(1-\alpha)\) agents with demand \((1,\varepsilon)\), where \(\varepsilon=1/n^2\).
Group \(G_2\) contains one special agent \(i^*\) with demand
\((1/[n(1-\alpha)],1)\); all other agents in \(G_2\) have demand \((1-\varepsilon,1)\).
Under \(\AOne\), the special agent can receive at most \((1-\alpha)/(1+f_1(\alpha))+o(1)\) of resource~2, whereas a fair allocation can give it almost \(1-\alpha\).

For a feasible SI--EF comparison, give every agent except \(i^*\) its baseline bundle and give \(i^*\) utility \((1-\varepsilon)(1-\alpha)\) according to its demand.
Before allocating to \(i^*\), the remaining resources are
\[
R_1^0=\alpha-\frac{(n\alpha-1)(1-\varepsilon)}n\ge\frac1n,
\qquad
R_2^0=(1-\alpha)(1-\varepsilon)+\frac1n.
\]
The special agent consumes \((1-\varepsilon)/n\) of resource~1 and \((1-\varepsilon)(1-\alpha)\) of resource~2, so the allocation is feasible and satisfies SI for sufficiently large \(n\).
The ordinary \(G_2\) agents and \(i^*\) have the same resource~1 allocation, preventing envy toward \(i^*\) within \(G_2\); \(i^*\) does not envy the baseline agents.
Cross-group EF follows because non-dominant allocations are at most \(1/n\).
Its welfare is at least
\[
1-\frac1n+(1-\varepsilon)(1-\alpha)\ge2-\alpha-\frac2n.
\]

For the actual mechanism,
\[
\begin{aligned}
R_1&=\alpha-\frac{1/[n(1-\alpha)]+(n\alpha-1)(1-\varepsilon)}n
=\frac1n+O(n^{-2}),\\
R_2&=(1-\alpha)(1-\varepsilon),\\
R_1^*&=R_1+\frac1{n^2(1-\alpha)},
\qquad R_2^*=R_2+\frac1{n^3}.
\end{aligned}
\]
Thus \(\theta=f_1(\alpha)R_1^*/R_2^*
=f_1(\alpha)/[n(1-\alpha)]+O(n^{-2})\).
The special agent has the smallest non-dominant coefficient in \(G_2\) for sufficiently large \(n\). Resource~1 feasibility and the fixed speed therefore give
\[
\left(\theta+\frac1{n(1-\alpha)}\right)\Delta S_2\le R_1.
\]
Consequently,
\[
\SW(\AOne(\mathbf I))
=1+(1+\theta)\Delta S_2
\le1+\frac{(1+\theta)R_1}{\theta+1/[n(1-\alpha)]}
=1+\frac{1-\alpha}{1+f_1(\alpha)}+O(1/n).
\]
Combining the two welfare estimates yields the limiting lower bound
\[
\frac{2-\alpha}{1+(1-\alpha)/(1+f_1(\alpha))}
=\frac{1+2\alpha-2\alpha^2}{1+\alpha-\alpha^2}.
\]
Taking \(\alpha=\lfloor n/2\rfloor/n\) gives \(6/5-O(1/n)\); the estimates above are uniform for \(\alpha\) near \(1/2\).
Together with the upper bound, this proves \(\FR(\AOne)=6/5+O(1/n)\).
\end{proof}

\subsection{Properties and welfare guarantee of \ATwo}
\label{app:atwo-fair-ratio}

We first establish the welfare bounds used in the proof of Theorem~\ref{thm:ASF2-summary}, then verify the mechanism's fairness, strategyproofness, and welfare guarantee together. The following refinement of the resource-exhaustion comparison applies to every \ASF mechanism.

\begin{lemma}[Refined upper bound]
\label{lem:asf-resource-exhaustion-comparison}
If resource~1 is exhausted by the \ASF allocation, then
\begin{equation}
\label{eq:asf-r1-ratio-bound}
\frac{1+\Delta S_1^*+\Delta S_2^*}
{1+\Delta S_1+\Delta S_2}
\le
\max_{0\le \beta\le 1}
\frac{
1+R_2^*\cdot
\min\left\{r+1-\beta,\frac{r}{\beta}\right\}
}{
1+R_2^*\cdot
\frac{r(1+r\vartheta )}
{r\vartheta +\beta}
}+\frac2n,
\end{equation}
where \(r/\beta=+\infty\) when \(\beta=0\).
If resource~2 is exhausted by the \ASF allocation, then
\begin{equation}
\label{eq:asf-r2-ratio-bound}
\frac{1+\Delta S_1^*+\Delta S_2^*}
{1+\Delta S_1+\Delta S_2}
\le
\max_{0\le \beta\le 1}
\frac{
1+R_2^*\cdot
\min\left\{1+r(1-\beta),\frac{1}{\beta}\right\}
}{
1+R_2^*\cdot
\frac{1+r\vartheta }
{1+\beta r\vartheta }
}+\frac2n,
\end{equation}
where \(1/\beta=+\infty\) when \(\beta=0\).
\end{lemma}

\begin{proof}
Use the cost functions $C_1,C_2$ of Appendix~\ref{app:canonical}. Their convexity and slope bounds follow from Lemma~\ref{lem:asf-cost-curves}.

If the mechanism is fair-optimal, its welfare ratio is one. In either bound, the fractional expression equals one at \(\beta=1\), so the claim follows.

We may therefore assume \(\Delta S_1^*+\Delta S_2^*>\Delta S_1+\Delta S_2\).
Suppose first that resource~1 is exhausted by the \ASF allocation. Denote 
\[\beta=\frac{C_2(\Delta S_2)}{\Delta S_2} \in [0,1].\]
Thus group \(G_2\) consumes \(C_2(\Delta S_2)=\beta\Delta S_2\) units of resource~1 in Step~2.
Since \(\Delta S_1=r\vartheta \Delta S_2\), resource~1 balance gives
\[
    R_1
    =
    \Delta S_1+\beta\Delta S_2
    =
    (r\vartheta +\beta)\Delta S_2 .
\]
As in the proof of Lemma~\ref{lem:S2/S2*}, \(R_1^*-R_1\le\beta/n\). Since \(\beta\le1\), the \ASF welfare denominator satisfies
\[
    \Delta S_1+\Delta S_2
    =\frac{(1+r\vartheta )R_1}{r\vartheta +\beta}
    \ge\frac{r(1+r\vartheta )}{r\vartheta +\beta}\,R_2^*
    -\frac{\beta(1+r\vartheta )}{n(r\vartheta +\beta)}
    \ge\frac{r(1+r\vartheta )}{r\vartheta +\beta}\,R_2^*-\frac1n.
\]

It remains to bound the fair optimum's welfare.
Since resource~1 is exhausted, Lemma~\ref{lem:S2/S2*} gives \(\Delta S_2^* \ge \Delta S_2\).
Convexity and \(C_2(0)=0\) then imply
\[
    \frac{C_2(\Delta S_2^*)}{\Delta S_2^*}
    \ge
    \frac{C_2(\Delta S_2)}{\Delta S_2}
    =
    \beta,
\]
and therefore
\[
    \Delta S_2^*-C_2(\Delta S_2^*)
    \le
    (1-\beta)\Delta S_2^*.
\]
Feasibility gives \(\Delta S_2^*\le R_2^*\), and the inequalities \(C_2(\Delta S_2^*)\le R_1^*\) and \(C_2(\Delta S_2^*)\ge\beta \Delta S_2^*\) give \(\Delta S_2^*\le R_1^*/\beta\) when \(\beta>0\).
Thus
\[
\Delta S_2^* \le \min\left\{R_2^*,\frac{R_1^*}{\beta}\right\}.
\]
Since \(\Delta S_1^*+C_2(\Delta S_2^*)\le R_1^*\), it follows that 
\[
    \Delta S_1^*+\Delta S_2^*
    \le
    R_1^*+(1-\beta)
    \min\left\{R_2^*,\frac{R_1^*}{\beta}\right\}.
\]
Substituting \(R_1^*=rR_2^*\) in the numerator bound gives
\[
    R_1^*+(1-\beta)
    \min\left\{R_2^*,\frac{R_1^*}{\beta}\right\}
    =
    R_2^*\cdot\min\left\{r+1-\beta,\frac{r}{\beta}\right\}.
\]
The resulting numerator bound is at most two, since \(R_1^*+R_2^*\le1\). Both the actual welfare denominator and the denominator in~\eqref{eq:asf-r1-ratio-bound} are at least one, and the latter exceeds the former by at most \(1/n\). The welfare ratio is therefore at most the displayed fractional expression plus \(2/n\), proving~\eqref{eq:asf-r1-ratio-bound}.

The case in which resource~2 is exhausted follows by the same argument with the two resources and the two groups exchanged. Denote
\[\beta=\frac{C_1(\Delta S_1)}{\Delta S_1}\in[0,1].\]
Then resource~2 balance gives
\[
    R_2
    =
    \beta\Delta S_1+\Delta S_2
    =
    (1+\beta r\vartheta )\Delta S_2 .
\]
Now \(R_2^*-R_2\le\beta/n\), so the \ASF welfare denominator satisfies
\[
    \Delta S_1+\Delta S_2
    =\frac{(1+r\vartheta )R_2}{1+\beta r\vartheta }
    \ge\frac{1+r\vartheta }{1+\beta r\vartheta }\,R_2^*
    -\frac{\beta(1+r\vartheta )}{n(1+\beta r\vartheta )}
    \ge\frac{1+r\vartheta }{1+\beta r\vartheta }\,R_2^*-\frac1n.
\]
By the symmetric envelope argument applied to \(C_1\), and using resource~2 feasibility of the fair optimum,
\[
    \Delta S_1^*+\Delta S_2^*
    \le
    R_2^*+(1-\beta)
    \min\left\{R_1^*,\frac{R_2^*}{\beta}\right\}.
\]
Using \(R_1^*=rR_2^*\), the numerator bound becomes
\[
    R_2^*+(1-\beta)
    \min\left\{R_1^*,\frac{R_2^*}{\beta}\right\}
    =
    R_2^*\cdot\min\left\{1+r(1-\beta),\frac{1}{\beta}\right\}.
\]
The same denominator comparison adds at most \(2/n\), proving~\eqref{eq:asf-r2-ratio-bound}.
\end{proof}

Put
$X=R_1^*$, $Y=R_2^*$, $h=1/n$, and $r=X/Y$.
The speed rule is symmetric under exchanging the groups and resources, so it suffices to consider $0<r\le1$; the argument below does not require $\alpha\le1/2$.
In this case, the normalized speed is
\[
\vartheta =\frac{\Theta_2}{r}=r^t,
\qquad
t=t_1=1-\frac1{c_1}=\frac{Y-h}{1-\alpha}\in[0,1].
\]
Since $rY=X\le\alpha$ and $Y=t(1-\alpha)+h$,
\begin{equation}
\label{eq:asf2-new-envelope}
Y\le\frac{t+h}{1+rt}
\le\frac{t}{1+rt}+h.
\end{equation}
Also, $X+Y\le1$ gives $Y\le1/(1+r)$.

The two bounds in Lemma~\ref{lem:asf-resource-exhaustion-comparison} correspond to different resource-exhaustion cases. The following scalar estimates bound their fractional expressions by $12/11$ and $2/(2\sqrt2-1)$, respectively.

We use the Bernstein basis
\[
B_i^d(x)=\binom di x^i(1-x)^{d-i}.
\]
If a polynomial on \([0,1]^3\) has only nonnegative Bernstein coefficients, then it is nonnegative throughout the cube. The exact coefficient data used below are collected in Section~\ref{app:asf2-bernstein-certificates}.

\begin{lemma}
\label{lem:asf2-new-scalar}
Let $0<r\le1$, $0\le t\le1$, $\vartheta =r^t$, and $0\le\beta\le1$.
If $0\le Y\le t/(1+rt)$, then
\[
\Psi_1(r,\vartheta ,Y,\beta):=
\frac{1+Y\min\{r+1-\beta,r/\beta\}}
     {1+Yr(1+r\vartheta )/(r\vartheta +\beta)}
\le\frac{12}{11}.
\]
If $0\le Y\le1/(1+r)$, then
\[
\Psi_2(r,\vartheta ,Y,\beta):=
\frac{1+Y(1+r(1-\beta))}
     {1+Y(1+r\vartheta )/(1+\beta r\vartheta )}
\le\frac{2}{2\sqrt2-1}.
\]
Here $r/\beta=+\infty$ when $\beta=0$.
\end{lemma}

\begin{proof}
For either ratio, a value at most one needs no further bound. Otherwise, the ratio is increasing in $Y$, so we may use the corresponding upper bound on $Y$.

\emph{Resource 1.}
The ratio $\Psi_1$ is nondecreasing in $\vartheta $, since
\[
\frac{\partial}{\partial \vartheta }\frac{r(1+r\vartheta )}{r\vartheta +\beta}
=\frac{r^2(\beta-1)}{(r\vartheta +\beta)^2}\le0.
\]
For $\beta\ge r$,
\[
\Psi_1-1
=\frac{Yr^2\vartheta (1-\beta)}
       {\beta\bigl(r\vartheta +\beta+Yr(1+r\vartheta )\bigr)},
\]
which is nonincreasing in $\beta$. It therefore suffices to consider
$0\le\beta\le r$ and $Y=t/(1+rt)$.

Set $s=\sqrt r$ and $v=\beta/r$.
Convexity of $r^t$ in $t$ gives the following upper chords:
\[
r^t\le
\begin{cases}
1-2t+2ts, & 0\le t\le1/2,\\
(2-2t)s+(2t-1)s^2, & 1/2\le t\le1.
\end{cases}
\]
Write $t_0=u/2$, $t_1=(1+u)/2$, $Q_0=1-u+us$, and
$Q_1=(1-u)s+us^2$, where $u\in[0,1]$.
After substituting $\vartheta =Q_e$, $t=t_e$, and $\beta=rv$, and clearing positive denominators, the inequality $\Psi_1\le12/11$ is equivalent to
\begin{equation}
\label{eq:asf2-polynomial}
P_e(s,u,v):=
(1+s^2t_e)(Q_e+v)
+t_e\bigl[12(1+s^2Q_e)
-11(1+s^2-s^2v)(Q_e+v)\bigr]\ge0,
\end{equation}
for $e=0,1$ on $[0,1]^3$.
The coefficient certificate in Section~\ref{app:asf2-bernstein-certificates} proves~\eqref{eq:asf2-polynomial}, and hence the first bound.

\emph{Resource 2.}
The ratio $\Psi_2$ is nonincreasing in $\vartheta $, since
\[
\frac{\partial}{\partial \vartheta }\frac{1+r\vartheta }{1+\beta r\vartheta }
=\frac{r(1-\beta)}{(1+\beta r\vartheta )^2}\ge0.
\]
As $\vartheta =r^t\ge r$, it suffices to take $\vartheta =r$ and $Y=1/(1+r)$.
Put $z=r\beta\in[0,r]$. Then
\[
\Psi_2(r,r,1/(1+r),\beta)-1
=\frac{(r-z)(1-r+rz)}
       {2+r+r^2+rz(1+r)}.
\]
Let $\kappa=(4\sqrt2-5)/7$, so $1+\kappa=2/(2\sqrt2-1)$.
For $0<r\le3/5$, subtracting the numerator from $\kappa$ times the denominator gives
\[
(1+\kappa)(r-\sqrt2+1)^2
+\bigl[1-(1-\kappa)r(1+r)\bigr]z+rz^2\ge0,
\]
because $r(1+r)\le24/25<1$.

For $3/5\le r\le1$, the stronger bound with $\kappa_0=3/32<\kappa$ holds.
Write
\[
g(r)=\kappa_0(2+r+r^2)-r+r^2,
\qquad
b(r)=1-(1-\kappa_0)r(1+r).
\]
The corresponding difference is $rz^2+b(r)z+g(r)$.
The coefficient certificate in Section~\ref{app:asf2-bernstein-certificates} gives $4rg(r)-b(r)^2>0$ on this interval. Therefore
\[
rz^2+b(r)z+g(r)
=r\left(z+\frac{b(r)}{2r}\right)^2
+\frac{4rg(r)-b(r)^2}{4r}\ge0.
\]
This proves the second bound.
\end{proof}

\subsubsection*{Proof of Theorem~\ref{thm:ASF2-summary}}
\label{app:thm:ASF2-summary}

\begin{proof}
SI, EF, and PO follow from Lemma~\ref{lem:ASF-3-properties}.

\emph{Strategyproofness.} We verify that $\Theta_2$ is positive, monotone, and satisfies both conditions of Lemma~\ref{lem:sufficient-SP}.
The speed is positive whenever \(R_1^*,R_2^*>0\). The state satisfies \(1/n\le R_1^*\le\alpha\) and \(1/n\le R_2^*\le1-\alpha\).
Since \(c_1,c_2\ge1\) and \(2-1/c\le c\) for \(c\ge1\),
\[
1\le 2-\frac1{c_k}\le\min\{c_k,2\}\le\bar c,
\qquad k=1,2.
\]
We first verify monotonicity. On \(r\le1\),
\[
\log\Theta_2
=(2-1/c_1)\log r
=\left(1+\frac{R_2^*-1/n}{1-\alpha}\right)\log r.
\]
Hence
\[
\frac{\partial\log\Theta_2}{\partial R_1^*}
=\frac{2-1/c_1}{R_1^*}>0,
\]
and
\[
\frac{\partial\log\Theta_2}{\partial R_2^*}
=\frac{\log r}{1-\alpha}-\frac{2-1/c_1}{R_2^*}<0.
\]
The \(r\ge1\) branch follows by exchanging the groups and resources and taking the reciprocal of the speed. The two branches agree at \(r=1\), where \(\Theta_2=1\), proving monotonicity.

We next prove the two conditions in Lemma~\ref{lem:sufficient-SP}. 
We provide the proof for the first condition. Exchanging the groups and resources replaces \((\alpha,R_1^*,R_2^*,c_1,c_2)\) by \((1-\alpha,R_2^*,R_1^*,c_2,c_1)\) and replaces \(\Theta_2\) by its reciprocal, proving the second condition.
For the first condition, consider a feasible decrease from \(R_2^*\) to \(R_2^*-\delta\), with \(\bar c\delta<R_2^*\).
Bernoulli's inequality gives
\begin{equation}
\label{eq:asf2-bernoulli}
\left(\frac{R_2^*}{R_2^*-\delta}\right)^{\bar c}
\le\frac1{1-\bar c\delta/R_2^*}.
\end{equation}
There are three cases.

\emph{Both states have \(r\le1\).}
Here \(R_1^*\le R_2^*-\delta\le R_2^*\).
Let \(c_1'\) be the post-deviation value of \(c_1\). Then
\(1/c_1'=1/c_1+\delta/(1-\alpha)\), so
\[
\frac{\Theta_2(\alpha,R_1^*,R_2^*-\delta)}
{\Theta_2(\alpha,R_1^*,R_2^*)}
=
\left(\frac{R_2^*}{R_2^*-\delta}\right)^{2-1/c_1'}
\left(\frac{R_2^*}{R_1^*}\right)^{1/c_1'-1/c_1}.
\]
Since \(2-1/c_1'\le2-1/c_1\le\bar c\), the first factor is at most the right-hand side of~\eqref{eq:asf2-bernoulli}.
Put \(\lambda=1/c_1'-1/c_1=\delta/(1-\alpha)\in[0,1]\).
Concavity gives
\[
\begin{aligned}
\left(\frac{R_2^*}{R_1^*}\right)^\lambda
&\le1+\lambda\left(\frac{R_2^*}{R_1^*}-1\right)\\
&=1+\frac{\delta(R_2^*-R_1^*)}{(1-\alpha)R_1^*}
\le1+\frac{c_1\delta}{R_1^*},
\end{aligned}
\]
where the last inequality uses
\(R_2^*-R_1^*\le R_2^*-1/n=(1-\alpha)(1-1/c_1)\le c_1(1-\alpha)\).
Multiplying the two bounds proves the first condition.

\emph{Both states have \(r\ge1\).}
Here \(R_2^*-\delta\le R_2^*\le R_1^*\).
The exponent \(2-1/c_2\) is unchanged, and
\[
\frac{\Theta_2(\alpha,R_1^*,R_2^*-\delta)}
{\Theta_2(\alpha,R_1^*,R_2^*)}
=
\left(\frac{R_2^*}{R_2^*-\delta}\right)^{2-1/c_2}
\le\frac1{1-\bar c\delta/R_2^*}
\le\frac{1+c_1\delta/R_1^*}{1-\bar c\delta/R_2^*}.
\]

\emph{The deviation crosses \(r=1\).}
Here \(R_2^*-\delta\le R_1^*\le R_2^*\).
Both bases below are at least one, and therefore
\[
\begin{aligned}
\frac{\Theta_2(\alpha,R_1^*,R_2^*-\delta)}
{\Theta_2(\alpha,R_1^*,R_2^*)}
&=
\left(\frac{R_1^*}{R_2^*-\delta}\right)^{2-1/c_2}
\left(\frac{R_2^*}{R_1^*}\right)^{2-1/c_1}\\
&\le\left(\frac{R_2^*}{R_2^*-\delta}\right)^{\bar c}
\le\frac1{1-\bar c\delta/R_2^*}.
\end{aligned}
\]
This also implies the required condition.
The calculation holds for every \(0<\alpha<1\).
Thus Lemma~\ref{lem:sufficient-SP} implies SP for \ATwo.

\emph{Welfare guarantee.} We now prove the upper bound and its asymptotic tightness.

Consider first $r\le1$.
If resource~1 is exhausted, use~\eqref{eq:asf-r1-ratio-bound}.
For fixed $r,\vartheta ,\beta$, write its fractional expression as
$(1+YA)/(1+YB)$, where $A=\min\{r+1-\beta,r/\beta\}\le2$ and $B\ge0$.
Its derivative with respect to $Y$ is $(A-B)/(1+YB)^2\le2$.
Thus~\eqref{eq:asf2-new-envelope} and Lemma~\ref{lem:asf2-new-scalar} give
\[
\Psi_1(r,\vartheta ,Y,\beta)
\le\frac{12}{11}+\frac2n.
\]
If resource~2 is exhausted, use~\eqref{eq:asf-r2-ratio-bound}.
Since $r\le1$ and $\beta\le1$, its numerator is $1+Y(1+r(1-\beta))$.
The bound $Y\le1/(1+r)$ and Lemma~\ref{lem:asf2-new-scalar} give
$\Psi_2\le2/(2\sqrt2-1)$.
Including the additional $2/n$ in~\eqref{eq:asf-r1-ratio-bound} and~\eqref{eq:asf-r2-ratio-bound}, and using $12/11<2/(2\sqrt2-1)$, both cases have the required upper bound.
Exchanging the resources proves the same bound when $r\ge1$.

For tightness, take $n_2/n=\alpha=1-1/\sqrt2+O(1/n)$, give every agent in $G_1$ demand $(1,\varepsilon)$ and every agent in $G_2$ demand $(\varepsilon,1)$, and let $\varepsilon=n^{-2}$.
Then $r=\sqrt2-1+O(1/n)$, $t=1+O(1/n)$, and
$\Theta_2=r^2+O(1/n)$.
The mechanism exhausts resource~2 first, with welfare
\[
1+(1-\alpha)(1+r^2)+O(1/n)
=\frac{2+r+r^2}{1+r}+O(1/n).
\]
An allocation giving each group total dominant share $1/(1+\varepsilon)$ exhausts both resources, satisfies SI and EF for all sufficiently large $n$, and has welfare $2/(1+\varepsilon)$.
Summing the two resource constraints shows that no allocation has greater welfare.
The resulting ratio is therefore
\[
\frac{2(1+r)}{2+r+r^2}+O(1/n)
=\frac{2}{2\sqrt2-1}+O(1/n),
\]
as claimed.

\end{proof}

\subsection{Exact Bernstein certificate data}
\label{app:asf2-bernstein-certificates}

For the polynomials $P_e$ in~\eqref{eq:asf2-polynomial}, partition the $s$ and $u$ intervals into four equal parts. For
$i,j\in\{0,1,2,3\}$, expand
\[
P_e\left(\frac{i+x}{4},\frac{j+y}{4},z\right)
=\sum_{k=0}^4\sum_{\ell=0}^2\sum_{m=0}^4
 A^{eij}_{k\ell m}x^ky^\ell z^m.
\]
Its Bernstein coefficient of index $(a,b,c)$ and degree $(4,2,4)$ is
\[
 B^{eij}_{abc}
 =\sum_{k=0}^a\sum_{\ell=0}^b\sum_{m=0}^c
 A^{eij}_{k\ell m}
 \frac{\binom ak\binom b\ell\binom cm}
      {\binom4k\binom2\ell\binom4m}.
\]
The following matrices list $\min_{a,b,c}B^{eij}_{abc}$, with row $i$ and column $j$:
\[
\begin{gathered}
e=0:\quad
\begin{pmatrix}
429/512&565/2048&115/1024&337/2048\\
1865/2048&153/512&209/2048&83/1536\\
1&39/64&71/256&11/64\\
1&1097/1024&7799/8192&435/512
\end{pmatrix},\\[2mm]
e=1:\quad
\begin{pmatrix}
65/128&4701/8192&1555/2048&8717/8192\\
37/512&161/2048&17/128&147/512\\
51/512&11/128&55/512&415/2048\\
26001/32768&1599/2048&6543/8192&28353/32768
\end{pmatrix}.
\end{gathered}
\]
All entries are positive. Since the Bernstein basis functions are nonnegative and sum to one, this proves~\eqref{eq:asf2-polynomial} on every subrectangle.

For $g(r)$ and $b(r)$ in the resource~2 case of Lemma~\ref{lem:asf2-new-scalar}, substitute $r=(3+2x)/5$. The polynomial $4rg(r)-b(r)^2$ has degree-four Bernstein coefficients
\[
\frac{731}{10000},\quad
\frac{3799}{16000},\quad
\frac{5317}{12800},\quad
\frac{317}{512},\quad
\frac{215}{256}.
\]
All five coefficients are positive, so $4rg(r)-b(r)^2>0$ for $3/5\le r\le1$.

\Needspace{7\baselineskip}
\subsection{The information limitation of ASF}
\subsubsection*{Proof of Theorem~\ref{thm:asymptotic-asf-lower-bound}}
\label{app:asymptotic-asf-lower-bound}
\begin{proof}
Let \(n\ge6\) be divisible by three and put \(\varepsilon_n=n^{-2}\).
We construct two instances \(I_n\) and \(J_n\) with identical state parameters after Step~1.
In both instances, \(G_1\) contains \(2n/3\) agents with demand \((1,\varepsilon_n)\), and \(G_2\) contains \(n/3\) agents.
In \(I_n\), every \(G_2\) agent has demand \((1/3,1)\).
In \(J_n\), one special \(G_2\) agent has demand \((\varepsilon_n,1)\), and all other \(G_2\) agents have demand \((1/3,1)\).
The remaining amounts of resource~1 are \(2/9\) in \(I_n\) and
\(2/9+1/(3n)-\varepsilon_n/n\) in \(J_n\).
The minimum non-dominant coefficient in \(G_2\) is \(1/3\) in \(I_n\) and
\(\varepsilon_n\) in \(J_n\).
Both instances therefore have exactly the same state
\[
\alpha=\frac13,\qquad
R_1^*=\frac29+\frac1{3n},\qquad
R_2^*=\frac23(1-\varepsilon_n)+\frac{\varepsilon_n}{n}.
\]
Therefore, any \ASF{} rule depending only on \((\alpha,R_1^*,R_2^*)\) must choose the same Step~2 speed \(\theta>0\) on \(I_n\) and \(J_n\).
Write \(\SW^*\) for the optimal SI--EF welfare.

We first consider instance \(I_n\).
Give every \(G_1\) agent utility \(1/n\) and divide total dominant share \(1-2\varepsilon_n/3\) equally among \(G_2\).
The resulting loads are \(1-2\varepsilon_n/9\) and \(1\), and welfare is \(5/3-2\varepsilon_n/3\).
It satisfies SI and within-group EF immediately.
Moreover, every \(G_2\) bundle contains at most \(1/n\) of resource~1,
whereas every \(G_1\) bundle contains only \(\varepsilon_n/n\) of
resource~2, so there is no cross-group envy.
Thus $\SW^*(I_n)\ge \frac53-\frac{2\varepsilon_n}{3}$.

For the mechanism, resource~1 feasibility gives
\[
\left(\theta+\frac13\right)\Delta S_2\le\frac29,
\]
and hence
\[
\SW(\ASF(I_n))
\le1+\frac{(1+\theta)(2/9)}{\theta+1/3}.
\]
Therefore,
\[
\frac{\SW^*(I_n)}{\SW(\ASF(I_n))}
\ge\frac{15\theta+5}{11\theta+5}-\frac23\varepsilon_n.
\]

We next consider instance \(J_n\).
Give every \(G_1\) agent utility \(4/(3n)\), every ordinary \(G_2\) agent utility \(1/n\), and the special agent utility
\[
\frac23+\frac1n-\frac89\varepsilon_n.
\]
This allocation exhausts resource~2. Its load on resource~1 is
\[
1-\frac1{3n}
+\varepsilon_n\left(\frac23+\frac1n-\frac89\varepsilon_n\right)\le1,
\]
since the special agent's utility is at most one and \(\varepsilon_n\le1/(3n)\).
Thus the allocation is feasible and satisfies SI for \(n\ge6\).
It is also EF.
An ordinary \(G_2\) agent evaluates the special bundle at at most \(3\varepsilon_n\le1/n\), while the special agent does not envy the ordinary agents' baseline bundles.
Across groups, every \(G_2\) bundle contains at most \(1/(3n)\) of resource~1, and the \(G_1\) bundles contain only \(4\varepsilon_n/(3n)\) of resource~2.
Hence, $\SW^*(J_n)\ge \frac{17}{9}-\frac{8\varepsilon_n}{9}$.

For the mechanism, all \(G_1\) agents have non-dominant coefficient \(\varepsilon_n\).
Resource~2 feasibility therefore gives
\[
(1+\varepsilon_n\theta)\Delta S_2
\le\frac23(1-\varepsilon_n).
\]
In particular,
\[
\SW(\ASF(J_n))
\le1+\frac{(1+\theta)(2/3)(1-\varepsilon_n)}
{1+\varepsilon_n\theta}
\le\frac53+\frac23\theta.
\]
Consequently,
\[
\frac{\SW^*(J_n)}{\SW(\ASF(J_n))}
\ge\frac{17}{15+6\theta}-\frac89\varepsilon_n.
\]

The function \((15\theta+5)/(11\theta+5)\) is increasing, whereas
\(17/(15+6\theta)\) is decreasing.
They intersect at the positive root of
\[
45\theta^2+34\theta-5=0,
\qquad
\theta=\frac{\sqrt{514}-17}{45}.
\]
Their common value is \((191-2\sqrt{514})/135\).
For every \(n\) divisible by three, at least one of \(I_n,J_n\) therefore has welfare ratio at least
\[
\frac{191-2\sqrt{514}}{135}-\frac89\varepsilon_n.
\]
Since \(\varepsilon_n=n^{-2}\) and \(n\) can be arbitrarily large, the claimed lower bound follows.
\end{proof}

\section{Proofs for Corrected Resource Balancing}\label{app:crb}
\label{app:missing-proofs-section3}

This appendix supplements Section~\ref{sec:crb}: we prove the universal lower bound, record the water-filling cost and cell formulas, derive the cellwise form of CRB, and give the implementation.

\subsection{Proof of Theorem~\ref{two-lower}}
\phantomsection\label{app:two-lower}

\begin{proof}
Consider an instance with $n-2$ background agents and two agents with different dominant resources:
\[
\mathbf I_0=
\bigl(\underbrace{(1,1),\ldots,(1,1)}_{n-2},
(1,\tfrac12),(\tfrac12,1)\bigr).
\]
A fair-optimal allocation gives every background agent utility $1/n$ and the two special agents utilities $x^*=y^*=4/(3n)$, yielding $\OPT(\mathbf I_0)=1+2/(3n)$.

Write the PO output $f(\mathbf I_0)$ in the non-wasteful form fixed above. At least one resource is exhausted. Let $x$ and $y$ be the utilities of the two special agents, and let $T$ be the total utility of the background agents. Suppose first that resource~1 is exhausted. Then
\[
T+x+\frac12y=1,
\]
and therefore the welfare produced by the mechanism is
\[
\begin{aligned}
\SW(f(\mathbf I_0))
&=T+x+y\\
&=\left(T+x+\frac12y\right)+\frac12y\\
&=1+\frac y2.
\end{aligned}
\]
Because $f$ is deterministic, its output on $\mathbf I_0$ is fixed. We perturb the agent of type $(1/2,1)$ in the resource~1 case. If resource~2 is exhausted instead, the same argument perturbs the agent of type $(1,1/2)$. We present the first case.

For $0<\varepsilon<1/2$, define
\[
\mathbf I_1=
\bigl(\underbrace{(1,1),\ldots,(1,1)}_{n-2},
(1,\tfrac12),(\varepsilon,1)\bigr).
\]
If the last agent reports $(1/2,1)$, the reported instance is exactly $\mathbf I_0$. A bundle giving reported utility $y$ contains at least $y/2$ units of resource~1 and $y$ units of resource~2. Under the true demand $(\varepsilon,1)$, where $\varepsilon<1/2$, this bundle therefore gives utility at least $y$. SP requires the agent's truthful utility on $\mathbf I_1$ to be at least $y$ as well. Let $T_1$ be the total utility of the background agents and let $x_1,y_1$ be the truthful utilities of the two special agents on $\mathbf I_1$. SI and SP give $T_1\ge(n-2)/n$ and $y_1\ge y$, while resource~2 feasibility gives
\[
T_1+\frac12x_1+y_1\le1.
\]
Consequently, $x_1/2\le2/n-y$, and hence
\[
\begin{aligned}
\SW(f(\mathbf I_1))
&=T_1+x_1+y_1\\
&=\left(T_1+\frac12x_1+y_1\right)+\frac12x_1\\
&\le1+\frac2n-y.
\end{aligned}
\]

For the fair optimum of $\mathbf I_1$, give every background agent utility $1/n$ and allocate the remaining capacities optimally to the two special agents. Their optimal utilities are
\[
x^*_{\varepsilon}=\frac{4(1-\varepsilon)}{n(2-\varepsilon)},
\qquad
y^*_{\varepsilon}=\frac{2}{n(2-\varepsilon)},
\]
so
\[
\OPT(\mathbf I_1)
=\frac{n-2}{n}+x^*_{\varepsilon}+y^*_{\varepsilon}
\longrightarrow1+\frac1n
\qquad\text{as }\varepsilon\to0.
\]

Thus the same utility $y$ creates opposite pressures. A larger $y$ improves the mechanism on $\mathbf I_0$, but SP then forces the perturbed agent to retain more utility on $\mathbf I_1$. Taking $\varepsilon\to0$ gives
\[
\FR(f)\ge
R_0(y):=\frac{1+\frac{2}{3n}}{1+\frac y2},
\qquad
\FR(f)\ge
R_1(y):=\frac{1+\frac1n}{1+\frac2n-y}.
\]
The first bound decreases with $y$ and the second increases with $y$, so the smallest possible worst-case ratio occurs when they are equal. Solving $R_0(y)=R_1(y)$ gives
\[
y=\frac{2(5n+4)}{n(9n+7)},
\qquad
R_0(y)=R_1(y)=\frac{9n+7}{9n+6},
\]
which proves the theorem.
\end{proof}

\subsection{Water-filling costs and active sets}\label{app:two-geometry}

The cost functions $C_1,C_2$ are defined in Appendix~\ref{app:canonical}. With the group totals from Section~\ref{sec:crb-opt},
\[
\begin{aligned}
\Delta S_1(p)&=X(p)-|G_1|h,&
C_1(\Delta S_1(p))&=A(p)-A(0),\\
\Delta S_2(q)&=Y(q)-|G_2|h,&
C_2(\Delta S_2(q))&=B(q)-B(0).
\end{aligned}
\]
Thus feasibility in increment coordinates is
\begin{equation}\label{two-surplus-feasible}
\Delta S_1+C_2(\Delta S_2)\le R_1,\qquad
C_1(\Delta S_1)+\Delta S_2\le R_2,
\end{equation}
where $R_1,R_2$ are the SI residual capacities in~\eqref{two-residuals}. Lemma~\ref{lem:asf-cost-curves} gives continuity, piecewise linearity, and slopes in $(0,1]$ for each nonempty group's cost function, including when its trajectory is continued beyond the original capacities.

On a cell with active sets $P,Q$, let $k,\ell,U,V$ be as in Section~\ref{sec:crb-manipulation}. Then
\[
\begin{aligned}
X(p)&=(|G_1|-k)h+Up,&
A(p)&=h\sum_{i\in G_1\setminus P}a_i+kp,\\
Y(q)&=(|G_2|-\ell)h+Vq,&
B(q)&=h\sum_{j\in G_2\setminus Q}b_j+\ell q.
\end{aligned}
\]
These affine formulas agree at all activation thresholds, including simultaneous activations.

\subsection{Cellwise form of CRB}\label{app:crb-balanced-lp}

After the SI case, Lemma~\ref{lem:crb-balance-feasibility} characterizes the
infeasibility cases by the strict endpoint tests
$\Phi(\bar p,0)>0$ and $\Phi(0,\bar q)<0$.
If neither test holds, \eqref{crb-program-opt} is feasible, including
when an endpoint score is zero. On a fixed active-set cell
$\mathcal C_{P,Q}$, the formulas in Appendix~\ref{app:two-geometry}
make the objective and all constraints affine. The restriction is
therefore the two-variable linear program
\begin{equation}\label{crb-cell-lp}
\begin{aligned}
\max_{p,q}\quad & W(p,q)\\
\text{s.t.}\quad & L_1(p,q)\le1,\quad L_2(p,q)\le1,\\
&\Phi(p,q)=0,\quad (p,q)\in\mathcal C_{P,Q}.
\end{aligned}
\tag{CRB-LP}
\end{equation}
The finitely many closed cells cover $[0,h]^2$, and their affine
formulas agree at shared boundaries. Thus maximizing over all feasible
regional programs is equivalent to \eqref{crb-program-opt};
Lemma~\ref{crb-feasible-opt} guarantees a unique maximizing physical
allocation. The next subsection shows that only $O(n)$ cells need to
be examined.

\subsection{Implementation of CRB}\label{app:crb-implementation}

Algorithm~\ref{two-algorithm} implements
Definition~\ref{crb-definition-opt}. It first handles the SI case and
uses the strict endpoint tests to detect infeasibility; otherwise it
solves the global corrected-balance problem through the regional
programs~\eqref{crb-cell-lp}. Sorting the distinct activation thresholds
$a_i h$ for $p$ and $b_j h$ for $q$, together with $0$ and $h$,
partitions the two level ranges into closed intervals
$I_1,\ldots,I_K$ and $J_1,\ldots,J_M$.

Balanced pairs are ordered in both coordinates, as used in the proof
of Lemma~\ref{crb-feasible-opt}. Starting with
$I_1\times J_1$, the algorithm scans cells in that order. If a regional
optimizer exhausts a resource, it is also a global optimizer: at least
one agent is above SI, so any larger balanced pair would strictly
increase both physical loads and be infeasible, while a smaller balanced
pair has weakly lower welfare. Otherwise the scan advances.

Let $p^+,q^+$ be the upper endpoints of the current intervals. If
$\Phi(p^+,q^+)>0$, monotonicity rules out balanced pairs with
$p\le p^+$ and $q\ge q^+$, so the algorithm advances to the next
$p$-interval. If $\Phi(p^+,q^+)<0$, it rules out balanced pairs with
$p\ge p^+$ and $q\le q^+$, so the algorithm advances to the next
$q$-interval. At equality, both intervals advance. These steps preserve
every later balanced pair not already examined, including pairs at
activation thresholds. Since a global optimizer exhausts a resource
by Lemma~\ref{crb-feasible-opt}, the scan reaches a cell where it returns
that allocation.

\begin{algorithm}[htbp]
\DontPrintSemicolon
\SetAlFnt{\small\linespread{1.2}\selectfont}
\caption{$\CRB(\mathbf d_1,\ldots,\mathbf d_n)$}
\label{two-algorithm}
Set $h\leftarrow1/n$ and form $G_1,G_2$\;
\If{$G_1=\varnothing$ or $G_2=\varnothing$ or $\max\{L_1(0,0),L_2(0,0)\}=1$}{
  \Return the SI allocation induced by $(0,0)$\;
}
Sort the distinct activation thresholds, including $0$ and $h$,
to form intervals $I_1,\ldots,I_K$ and $J_1,\ldots,J_M$\;
Precompute the affine expressions for $X,A$ on each $I_i$
and for $Y,B$ on each $J_j$\;
\tcp{Detect the two infeasibility cases}
Compute $\bar p,\bar q$ by single-group water filling\;
\If{$\Phi(0,\bar q)<0$}{
  \Return the canonical allocation induced by $(0,\bar q)$\;
}
\If{$\Phi(\bar p,0)>0$}{
  \Return the canonical allocation induced by $(\bar p,0)$\;
}
$(i,j)\leftarrow(1,1)$\;
\While{$i\le K$ and $j\le M$}{
  \tcp{Maximize welfare under corrected balance in the current region}
  Solve the regional program~\eqref{crb-cell-lp} on $I_i\times J_j$\;
  \If{the LP is feasible}{
    Let $(p,q)$ be an optimal solution\;
    \If{$\max\{L_1(p,q),L_2(p,q)\}=1$}{
      \Return the canonical allocation induced by $(p,q)$\;
    }
  }
  \tcp{Advance to the next region that may contain balanced pairs}
  Let $p^+,q^+$ be the upper endpoints of $I_i,J_j$\;
  $\phi\leftarrow\Phi(p^+,q^+)$\;
  \If{$\phi\ge0$}{
    $i\leftarrow i+1$\;
  }
  \If{$\phi\le0$}{
    $j\leftarrow j+1$\;
  }
}
\end{algorithm}

\paragraph{Proof of Lemma~\ref{two-runtime-opt}.}

\begin{proof}
Sorting the activation thresholds takes $O(n\log n)$ time. The affine coefficients on the resulting intervals can be computed in $O(n)$ time using prefix sums. The two single-group water-filling computations also take $O(n)$ time. During the region search, neither index decreases, and each iteration either returns an allocation or increases at least one index. Since $K\le |G_1|+1$ and $M\le |G_2|+1$, at most $K+M-1\le n+1$ regions are examined. In each region, the LP has two variables and a constant number of constraints, so it can be solved in constant time. Checking resource exhaustion and evaluating $\Phi$ also take constant time. Constructing the final allocation takes $O(n)$ time. The total running time is therefore $O(n\log n)$.
\end{proof}

\section{Proofs for Section~\ref{sec:beyond-two}}

\subsection*{Proof of Theorem~\ref{thm:PoSP}}
\label{app:thm:PoSP}

\begin{proof}
The upper bound follows from SI: every SI mechanism has $\SW\ge 1$, while every feasible allocation has $\SW\le m$.
We prove the matching lower bound. In fact, only SI and SP are used.

\paragraph{Instance Construction.}
Fix $m\ge 3$. Let \(n_2\) be a sufficiently large integer and set
\[
n=n_2^2,
\qquad
n_1=n-(m-1)n_2.
\]
There are $m$ groups $G_1,\ldots,G_m$, where $|G_1|=n_1$ and $|G_k|=n_2$ for every $k\ge 2$.
Let
\[
\epsilon_1=\frac1n,
\qquad
\epsilon_2=n^{-2n_2-3},
\qquad
\beta=n^2,
\qquad
 a_i=\epsilon_2\beta^i \quad (i\in[n_2]).
\]
Then $a_i\le n^{-3}$ and $\sum_{i=1}^{n_2}a_i\le n_2/n^3$.
Agents in $G_1$ have demand vector $\mathbf{d}=(1,\epsilon_1,\ldots,\epsilon_1)$.
For every $k\in\{2,\ldots,m\}$, the $i$-th agent in $G_k$ has demand vector
\[
 \mathbf{d}_{k,i}
 =
 \left(
 \frac12,
 a_i,\ldots,a_i,
 \underbrace{1}_{\text{resource }k},
 a_i,\ldots,a_i
 \right).
\]
To exploit SP, we will select one candidate agent from each group \(G_k\), \(k\ge2\), and modify only that agent's demand.
The modification keeps resource \(k\) as the dominant resource while scaling
all other coordinates by a factor 
\[
\gamma_i=\frac{n_2}{in} \quad (i\in[n_2]),
\]
which is smaller than $1$ since $n_2 \le n$.
The changed demand of candidate $i\in G_k$ is
\[
 \mathbf{d}'_{k,i}
 =
 \left(
 \frac{\gamma_i}{2},
 \gamma_i a_i,\ldots,\gamma_i a_i,
 \underbrace{1}_{\text{resource }k},
 \gamma_i a_i,\ldots,\gamma_i a_i
 \right).
\]
For a tuple $\mathbf{t}=(i_2,\ldots,i_m)\in[n_2]^{m-1}$, let $\mathbf{I}(\mathbf{t})$ be the instance in which candidate $i_k$ in each group $G_k$, $k=2,\ldots,m$, is changed.
For fixed $k$, let $\mathbf{t}_{-k}$ denote the tuple with $i_k$ removed; $\mathbf{I}(\mathbf{t}_{-k})$ is the instance where all groups except $G_k$ are changed according to $\mathbf{t}_{-k}$, while $G_k$ remains original; and $\mathbf{I}(\mathbf{t}_{-k},i)$ is obtained by further changing candidate $i$ in $G_k$.

\paragraph{Lower bound for the optimal fair allocation.}
We first show that every $\mathbf{I}(\mathbf{t})$ admits an SI and EF allocation with welfare $m-o(1)$.
Fix $\mathbf{t}=(i_2,\ldots,i_m)$ and put
\[
\eta=m\frac{n_2}{n}+\epsilon_1+(m-2)\sum_{i=1}^{n_2}a_i,
\qquad
s=1-\eta.
\]
Since \(m\) is fixed and \(\sum_{i=1}^{n_2}a_i\le n_2/n^3\), we have
$\eta=o(1)$ and $s=1-o(1)$.
Give every agent in $G_1$ utility $1/n$.
Now fix $G_k$ with $k\ge 2$ and let $i_k$ be its changed candidate.
Set
\[
 x_k=\max\left\{\frac1n,\gamma_{i_k}s\right\}.
\]
The changed candidate receives utility $s$, with $\gamma_{i_k}s/2$ units of resource $1$, $s$ units of resource $k$, and $\gamma_{i_k}a_{i_k}s$ units of every other resource.
Each original agent $\ell<i_k$ in $G_k$ receives utility $x_k$ according to its original demand vector.
Each original agent $\ell>i_k$ receives its SI allocation.

This allocation is feasible. For resource $1$, group $G_1$ uses $n_1/n$.
Each group $G_k$, $k\ge 2$, uses
\[
\frac{(i_k-1)x_k+\gamma_{i_k}s}{2}+\frac{n_2-i_k}{2n}
\le \frac{i_kx_k}{2}+\frac{n_2-i_k}{2n}
\le \frac{n_2}{n}
\]
of resource $1$, because either $x_k=1/n$ or $x_k=\gamma_{i_k}s=n_2s/(i_kn)$.
Hence resource $1$ is not overused.
For a resource $r\ge 2$, agents in $G_1$ use at most $\epsilon_1$.
Group $G_r$ uses at most $s+mn_2/n$ of resource $r$, and every other group uses at most $\sum_i a_i$ of resource $r$.
By the definition of $s$, the total usage is at most $1$.

The allocation satisfies SI by construction. We next verify EF.
Agents in $G_1$ do not envy any non-$G_1$ agent: for any allocation in a group $G_k$, there is an auxiliary resource $r\in\{2,\ldots,m\}\setminus\{k\}$, and the allocated amount of that resource is at most $n^{-3}$; an agent in $G_1$ therefore obtains utility at most $n^{-2}<1/n$ from that bundle.
No agent in $G_k$, $k\ge 2$, envies an agent outside $G_k$: bundles of agents in $G_1$ contain only $n^{-2}$ of resource $k$, and bundles from other groups contain at most $n^{-3}$ of resource $k$.

It remains to check envy within a fixed group $G_k$.
A prefix agent $\ell<i_k$ does not envy any other prefix agent, since every prefix bundle contains $x_k/2$ units of resource~1, which caps its utility at $x_k$. The changed candidate receives $\gamma_{i_k}s/2\le x_k/2$ units of resource~1, so its bundle also gives the prefix agent utility at most $x_k$.
It also does not envy any suffix agent, whose resource-$k$ allocation is only $1/n\le x_k$.
The changed candidate does not envy a prefix agent $p<i_k$: using any auxiliary resource $r\ne k$, its utility from $p$'s bundle is at most
\[
 \frac{a_px_k}{\gamma_{i_k}a_{i_k}}
 =\frac{x_k}{\gamma_{i_k}}\beta^{p-i_k}
 \le s,
\]
where the inequality is immediate if $x_k=\gamma_{i_k}s$, and if $x_k=1/n$ it follows from $x_k/\gamma_{i_k}=i_k/n_2\le 1$ and $\beta^{p-i_k}\le 1/\beta$.
It does not envy suffix agents because they receive only $1/n<s$ of resource $k$.
Finally, a suffix agent $\ell>i_k$ does not envy the changed candidate or any prefix agent: using an auxiliary resource $r\ne k$, the utilities are at most
\[
\frac{\gamma_{i_k}a_{i_k}s}{a_\ell}
=\gamma_{i_k}s\beta^{i_k-\ell}<\frac1n,
\qquad
\frac{a_px_k}{a_\ell}=x_k\beta^{p-\ell}<\frac1n.
\]
Suffix agents also do not envy one another because their SI bundles have the same resource-$1$ amount.
Thus the allocation is SI and EF. Using $\epsilon_1=1/n$ and $\sum_i a_i\le n_2/n^3$, its welfare gives
\[
\begin{aligned}
\SW^*(\mathbf I(\mathbf t))
&\ge \frac{n_1}{n}+(m-1)s
=1-\frac{(m-1)n_2}{n}+(m-1)s \\
&\ge m-\frac{(m^2-1)n_2}{n}
  -\frac{m-1}{n}
  -(m-1)(m-2)\frac{n_2}{n^3}=m-o(1).
\end{aligned}
\]
This bound holds for every tuple $\mathbf t$.

\paragraph{Upper bound for any SI and SP mechanism.}
We now upper bound any SI and SP mechanism $M$.
Fix a constant $C>(m-1)^2$ and define
\[
T_i=\frac{C}{\ln n_2}\cdot\frac1i\cdot\frac{n_2}{n},
\qquad
U=\frac{2C}{\ln n_2}.
\]
Fix $k\in\{2,\ldots,m\}$ and a partial tuple $\mathbf{t}_{-k}$.
Call an index $i\in[n_2]$ bad for $\mathbf{t}_{-k}$ if, in $\mathbf{I}(\mathbf{t}_{-k},i)$, candidate $i$ obtains utility larger than $U$.
Let $S_k(\mathbf{t}_{-k})$ be the set of bad indices.

If $i\in S_k(\mathbf{t}_{-k})$, then in $\mathbf{I}(\mathbf{t}_{-k})$ the same agent can misreport as $\mathbf{d}'_{k,i}$ and induce exactly the instance $\mathbf{I}(\mathbf{t}_{-k},i)$.
Since $\mathbf{d}'_{k,i}\ge \gamma_i \mathbf{d}_{k,i}$ coordinatewise, any allocation giving utility $y$ to the changed demand gives utility at least $\gamma_i y$ to the original demand.
Thus the misreport gives true utility greater than
\[
\gamma_i U
=
\frac{n_2}{in}\cdot\frac{2C}{\ln n_2}
=2T_i.
\]
By SP, the truthful utility of candidate $i$ in $\mathbf{I}(\mathbf{t}_{-k})$ is also greater than $2T_i$.
Because its original demand uses $1/2$ of resource $1$ per unit utility, it must receive more than $T_i$ of resource $1$ in $\mathbf{I}(\mathbf{t}_{-k})$.

In $\mathbf{I}(\mathbf{t}_{-k})$, SI forces agents in $G_1$ to consume at least $n_1/n$ of resource $1$.
Hence all agents outside $G_1$ receive total resource $1$ at most $(m-1)n_2/n$.
Therefore
\[
\sum_{i\in S_k(\mathbf{t}_{-k})}T_i
\le
(m-1)\frac{n_2}{n},
\qquad\text{and hence}\qquad
\sum_{i\in S_k(\mathbf{t}_{-k})}\frac1i
\le
\frac{(m-1)\ln n_2}{C}.
\]


We next use a probabilistic argument with \emph{harmonic} weights to show that there exists a tuple \(\mathbf t\) for which no changed candidate is bad.
This is the main difference from the previous lower-bound proof in \cite{BeiLL26}, which uses identical weights and can only establish a weaker lower bound of $2$ for $m=3$.
Let $H_{n_2}=\sum_{i=1}^{n_2}\frac1i$.
Choose each coordinate \(i_k\), \(k=2,\ldots,m\), independently according to
\[
\Pr[i_k=i]=\frac{1/i}{H_{n_2}},
\qquad i\in[n_2].
\]
Fix \(k\) and condition on \(\mathbf t_{-k}\).
Then
\[
\Pr[i_k\in S_k(\mathbf t_{-k})\mid \mathbf t_{-k}]
=
\frac{\sum_{i\in S_k(\mathbf t_{-k})}1/i}{H_{n_2}}
\le
\frac{(m-1)\ln n_2}{C H_{n_2}}
<
\frac{m-1}{C},
\]
where the last inequality uses \(H_{n_2}>\ln n_2\).
Hence the unconditional probability that the changed candidate in \(G_k\) is bad is also less than \((m-1)/C\).
Taking a union bound over \(k=2,\ldots,m\), the probability that some changed candidate is bad is
\[
\Pr[\text{some changed candidate is bad}]
<
\frac{(m-1)^2}{C}
<1.
\]
Therefore, with positive probability no changed candidate is bad, and hence there exists a tuple \(\mathbf t\) with this property.

For this tuple, agents in $G_1$ have total utility at most their total allocation of resource $1$, hence at most $1$.
Every unchanged agent outside $G_1$ has demand $1/2$ for resource $1$, so their total utility is at most twice their total allocation of resource $1$, which is at most $2(m-1)n_2/n$.
No changed candidate is bad, so each of the $m-1$ changed candidates has utility at most $U$.
Therefore
\[
\SW(M(\mathbf{I}(\mathbf{t})))
\le
1+2(m-1)\frac{n_2}{n}+(m-1)\frac{2C}{\ln n_2}
=1+o(1).
\]
Combining this with $\SW^*(\mathbf{I}(\mathbf{t}))\ge m-o(1)$ gives $\AR(M)\ge m-o(1)$.
Letting \(n_2\to\infty\) yields \(\FR(M)\ge m\).
Together with the upper bound $m$, this proves the theorem.
\end{proof}

\subsection*{Proof of Theorem~\ref{thm:randomized-PoSP}}
\label{app:randomized-lower-bound}

\begin{proof}
Let $M$ be any randomized mechanism satisfying ex-post SI and truthfulness in expectation.
The upper bound follows immediately from ex-post SI: every realized allocation has social welfare at least \(1\), while every feasible allocation has welfare at most \(m\). Hence \(\FR(M)\le m\).

For the lower bound, use the same instance family \(\mathbf I(\mathbf t)\) as in the proof of Theorem~\ref{thm:PoSP}.
The SI--EF comparison constructed there is deterministic, so
\[\SW^*(\mathbf I(\mathbf t))\ge m-o(1)\]
for every tuple \(\mathbf t\).
It remains to show that some tuple has expected mechanism welfare $1+o(1)$.

Fix \(C>(m-1)^2\), and define \(T_i\) and \(U\) as in the proof of Theorem~\ref{thm:PoSP}. 
For fixed \(k\) and \(\mathbf t_{-k}\), call \(i\in[n_2]\) \emph{bad} if the changed candidate has expected utility greater than \(U\) in \(\mathbf I(\mathbf t_{-k},i)\). 
Suppose \(i\) is bad. In \(\mathbf I(\mathbf t_{-k})\), candidate \(i\) can misreport as \(\mathbf d'_{k,i}\), thereby inducing \(\mathbf I(\mathbf t_{-k},i)\). 
Since $\mathbf d'_{k,i}\ge \gamma_i\mathbf d_{k,i}$ coordinatewise, this misreport gives expected utility, evaluated under the true demand, greater than $\gamma_i U=2T_i$.
Truthfulness in expectation therefore implies that candidate \(i\)'s truthful expected utility in \(\mathbf I(\mathbf t_{-k})\) is also greater than \(2T_i\). 
Since its true demand uses \(1/2\) unit of resource~1 per unit of utility, its expected allocation of resource~1 is greater than \(T_i\).

By ex-post SI, agents in \(G_1\) consume at least \(n_1/n\) units of resource~1 in every realization. 
Hence the expected total amount of resource~1 allocated outside \(G_1\) is at most 
$\frac{(m-1)n_2}{n}$.
Consequently, if \(S_k(\mathbf t_{-k})\) denotes the set of bad indices, 
\[
\sum_{i\in S_k(\mathbf{t}_{-k})}T_i
\le
(m-1)\frac{n_2}{n},
\qquad\text{and hence}\qquad
\sum_{i\in S_k(\mathbf{t}_{-k})}\frac1i
\le
\frac{(m-1)\ln n_2}{C}.
\]
The same harmonic-probability argument as in Theorem~\ref{thm:PoSP} now gives a tuple \(\mathbf t\) for which no changed candidate is bad.  

For this tuple, the expected total utility of agents in \(G_1\) is at most \(1\). 
Every unchanged agent outside \(G_1\) uses \(1/2\) unit of resource~1 per unit of utility, so their expected total utility is at most $2\frac{(m-1)n_2}{n}$. 
Finally, the \(m-1\) changed candidates contribute at most \((m-1)U\) in expectation. 
Thus \[ \mathbb E[\SW(M(\mathbf I(\mathbf t)))] \le 1+ 2\frac{(m-1)n_2}{n} +(m-1)U = 1+o(1). \] 
Combining this with \(\SW^*(\mathbf I(\mathbf t))\ge m-o(1)\) gives $\FR(M)\ge m-o(1)$. 
Letting \(n_2\to\infty\) yields \(\FR(M)\ge m\), completing the proof. 
\end{proof}

\end{document}